\documentclass[journal]{IEEEtran}

\usepackage[usenames,dvipsnames]{xcolor}
\usepackage{amsmath,amssymb,amsthm}
\usepackage{epstopdf}
\usepackage{courier}
\usepackage[normalem]{ulem}
\usepackage{soul} 
\usepackage{paralist}
\usepackage{balance}
\usepackage{enumerate}
\usepackage{graphicx}
\usepackage{epstopdf}
\usepackage{caption,subcaption}
\usepackage{multirow}
\usepackage{enumerate}
\usepackage{algorithm}
\usepackage{algorithmicx}
\usepackage{algpseudocode}
\usepackage{url}

\def\etc{\emph{etc}}
\def\eg{\emph{e.g. }}

\def\U{\mathcal{U}}
\def\T{\mathcal{T}}

\def\G{\mathcal{G}}

\newcommand{\vi} [2]{ {#1}^{(#2)} }

\newtheorem{corollary}{Corollary}

\newtheorem{lemma}{Lemma}

\title{An Efficient Web Traffic Defence Against Timing-Analysis Attacks}
 \author{Saman Feghhi and Douglas J. Leith\thanks{Work supported by SFI grant 11/PI/1177 and 13/RC/2077.}\\Trinity College Dublin, Ireland} 
\begin{document}

\maketitle

\begin{abstract}
We introduce a new class of lower overhead tunnel that is resistant to traffic analysis.  The tunnel opportunistically reduces the number of dummy packets transmitted during busy times when many flows are simultaneously active while maintaining well-defined privacy properties.   We find that the dummy packet overhead is typically less than $20$\% on lightly loaded links and falls to zero as the traffic load increases i.e. the tunnel is capacity-achieving.  The additional latency incurred is less than 100ms.  We build an experimental prototype of the tunnel and carry out an extensive performance evaluation that demonstrates its effectiveness under a range of network conditions and real web page fetches.
\end{abstract}


\section{Introduction}

Encrypting traffic to protect it from eavesdroppers is one of the basic building blocks of secure networks.   However, while encryption conceals the contents of transmitted packets other features of the transmitted packet stream, such as the packet sizes and timings, often still contain revealing information.   Taking advantage of this, over the last 10 years a number of increasingly powerful attacks have been demonstrated against encrypted packet streams, and encrypted web traffic in particular.   While attacks against HTTPS were initially based on packet sizes and counts, attacks using only packet timing information have recently been demonstrated, \eg \cite{feghhi161}.  The importance of the latter is that they are immune to defences currently being rolled out, such as padding packets to make them all the same size, plus they do not require \emph{a priori} knowledge of the start and end of each web fetch.   Such developments motivate revisiting how we transmit encrypted traffic.   

Three main approaches to defence have been considered to date.  One is to obfuscate the timing of information packets by transmitting packets at a constant rate, using buffering and insertion of dummy packets as needed \cite{dyer12,cai14,cai14b}.  While effective at disrupting attacks, this is recognised to involve such high overheads as to be unsuitable for general use.  A second, application-layer, approach is to randomise the pipelining of the requests forming the fetch a web page and also to inject dummy requests \cite{luo11}.   However, while incurring a low overhead (few dummy packets and low latency) this is largely ineffective against modern attacks \cite{cai12,wang14,juarez14}.   A third line of work is based on shaping traffic so that the transmitted trace is similar, in some appropriate sense, to that for a different web page or application \cite{wright09,nithyanand14,wang14,wang15}.   However, either these approaches are of limited effectiveness \cite{cai12,wang14,feghhi161} or again entail high overheads \cite{nithyanand14,wang14,wang15}.  In summary, existing defences are either of limited ineffectiveness or carry an excessive overhead. 

In this paper we introduce a new class of lower overhead tunnel that is resistant to traffic analysis.  The basic idea used is to ensure that, given an observed packet trace, many different sequences of web fetches could reasonably generate this trace.   Users are therefore provided strong deniability that a specific web page was fetched.  This indistinguishability idea is not new, but the approach used to realise it is novel.  Important features are that (i) our approach is lightweight and does not, for example, require pre-clustering of web pages (and so maintainance of a packet trace database etc) and (ii) packets from multiple web fetches/flows are aggregated thereby allowing us to opportunistically reduce the number of dummy packets transmitted during busy times when many flows are simultaneously active while maintaining well-defined privacy properties.   We find that the dummy packet overhead is typically less than $20$\% on lightly loaded links.  Further, the fraction of dummy packets falls to zero as the traffic load increases i.e. the tunnel is capacity-achieving.  The additional latency introduced by the tunnel is less than 100ms.  Importantly, we build an experimental prototype of the tunnel and carry out an extensive performance evaluation that demonstrates its effectiveness under a range of network conditions and real web page fetches.   {The importance of using a working implementation is twofold.  Firstly, the buffering associated with traffic shaping affects the behaviour of the web fetch, e.g. delay of a request has knock-on effects on the timing of the response which may in turn affect subsequent requests, and such interactions are not captured by the usual simulation approach of replaying a recorded packet sequence.  Secondly, practical issues such as handling of DNS requests and interactions between traffic shaping and TCP congestion control, which have received little attention in the traffic analysis literature, are highlighted.}

{The rest of the paper is organised as follows.   In Section II we formalise the threat model considered and set the present work in the context of the existing literature.   Section III introduces the use of a trace-based traffic shaping approach to achieve indistinguishability and in Section IV we show that it is possible to  perform traffic shaping in a way that, while ensuring indistinguishability, reduces the dummy packet overhead during busy times and importantly that the fraction of dummy packets falls to zero as the traffic load increases.  Section V presents experimental measurements on the performance of a prototype VPN implementing the proposed approach, evaluating a number of design choices and presenting performance data for both TCP and UDP traffic.   Section VI evaluates the performance of the tunnel against a state of the art traffic analysis attack and also compares its performance with other proposed defences based on indistinguishability.  Our conclusions are summarised in Section VII.
}

\begin{figure}
\centering
\includegraphics[width=0.7\columnwidth]{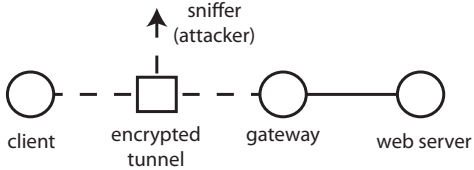}
\caption[Schematic illustrating threat model]{Schematic illustrating attacker of the type considered.  A client machine is connected to an external network via an encrypted tunnel (ssh, SSL, IPSec etc.).  The attacker can sniff packets traversing the tunnel, but has no other information about the clients activity.}\label{fig:intoexample}
\end{figure}


\section{Preliminaries}
\subsection{Threat Model}
The privacy disclosure scenario we consider arises naturally from the prevalence of relatively easy access to the network path connecting a client device to the internet.  We imagine two main types of attack/threat:
\begin{enumerate}
\item The first hop is a WiFi link and a nearby device equipped with a wireless interface sniffs packets transmitted across the wireless hop.  
\item ISPs and hot spot providers sniff the packets carried on the DSL/cable/ethernet link connecting the edge network to the internet.  
\end{enumerate}
Notice that the sensitive asset here is the packet trace, consisting of the time when a packet is transmitted together with the packet header and contents.   We assume that packet contents are encrypted and packets are padded to be of the same size, that is the user already has baseline defences in place by routing traffic via a suitable tunnel as illustrated schematically in Figure \ref{fig:intoexample}.  The relevant part of the packet trace for an attacker therefore consists of the timing of packet transmissions.   

We focus on web traffic, since this is a major source of traffic in the internet and has been the subject of much of the literature on traffic analysis attacks.    We do not seek to conceal the fact that the client device is browsing the web, but rather to prevent an attacker from inferring which pages are being browsed.   Packet timing data is a rich source of information, as discussed in more detail in Section \ref{sec:http} below, and in particular it is known to be sufficient to allow an attacker to infer with high probability the web page being browsed by a user \cite{feghhi161} while in \cite{dyer12} it is shown that packet padding is not sufficient to hide coarse grained features such as bursts in traffic or the total size and load time of a page.   

Although a client device may transmit packets from multiple flows, e.g. web and email, across its internet link, we want to avoid making strong assumptions about the number and characteristics of flows since this can change over time in ways that are difficult to predict i.e. we require a defence that provides a reasonable level of protection whether a single web fetch is in progress or many parallel flows are active.   That said, we expect that flow re-identification attacks may sometimes be harder to carry out when many flows are active simultaneously and of course a defence can and should opportunistically take advantage of this e.g. to reduce the overheads induced by the defence.

Note that we do not seek to address attacks that target the end host itself (viruses etc), nor attacks based on active packet injection by the attacker ( e.g. to force packet loss so as to manipulate user or web server behaviour).

\subsection{Anatomy of a Web Page Fetch}\label{sec:http}
We briefly review the reasons why packet timing information can be highly informative.    When traffic is encrypted the packet source and destination addresses and ports and the packet payload are hidden and the packets may be padded to be of equal size, so that packet size information is also concealed.   An attacker sniffing such encrypted traffic is therefore able only to observe the direction and timing of packets.   Figure \ref{fig:difweb} plots the timestamps of the uplink packets sent during the course of fetching five different health-related web pages (see below for details of the measurement setup).    The $x$-axis indicates the packet number $k$ within the stream and the $y$-axis indicates the corresponding timestamp.   It can be seen that these timestamp traces are distinctly different for each web site.

\begin{figure}
\centering
	\includegraphics[width=0.8\columnwidth]{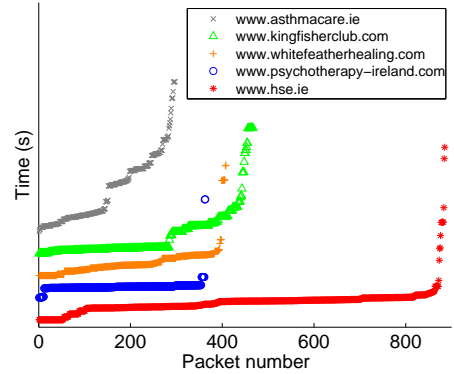}
	\caption{ Time traces of uplink traffic from 5 different Irish health-related web sites are shown. It can be seen that the web site time traces exhibit distinct patterns.  The traces are shifted vertically to avoid overlap and facilitate comparison. Figure reprinted from \protect\cite{feghhi161}.}
	\label{fig:difweb}
\end{figure}

To gain insight into the differences between the packet timestamp sequences in Figure \ref{fig:difweb}, it is helpful to consider the process of fetching a web page in more detail.    To fetch a web page the client browser starts by opening a TCP connection with the server indicated by the URL and issues an HTTP GET or POST request to which the server then replies.  As the client parses the server response it issues additional GET/POST requests to fetch embedded objects (images, css, scripts \etc.).   These additional requests may be to different servers from the original request (\emph{e.g.} when the object to be fetched is an advert or is hosted in a separate content-delivery network),  in which case the client opens a TCP connection to each new server in order to issue the requests.   Fetching of these objects may in turn trigger the fetching of further objects.   Note that asynchronous fetching of dynamic content using, \emph{e.g.} AJAX, can lead to a complex sequence of server requests and responses even after the page has been rendered by the browser.   Also, typically the TCP connections to the various servers are held open until the page is fully loaded so that they can be reused for later requests (request pipelining in this way is almost universally used by modern browsers).    It is this sequence of events while fetching the objects in a page that creates a packet timing ``signature'' which can be used to deanonymise the encrypted traffic and estimate the web page fetched with high probability.  For example, in \cite{feghhi161} web pages are correctly identified with $>90\%$ success rate.

\subsection{Related Work}
\label{sec:related}
\subsubsection{Attacks}
With regard to attacks, a fairly large body of literature now exists.  Some of the earliest work specifically focussed on attacks against encrypted web traffic appears to be that of Hintz \cite{hintz02}.   In this setup (i)  web page fetches occur sequentially with the start and end of each web page fetch known, and for each packet (ii) the client-side port number, (iii) the direction (incoming/outgoing) and (iv) the size is observed.   A web page signature is constructed consisting of the aggregate bytes received on each port (calculated by summing packet sizes), effectively corresponding to the number and size of each object within the web page.   

Subsequently,  Bissias \emph{et al} \cite{bissias06} considered an encrypted tunnel setup where (i) web page fetches occur sequentially with the start and end of each web page fetch known, and for each packet (ii) the size, (iii) the direction (incoming/outgoing) and (iv) the time (and so also the packet ordering)  is observed.    The sequence of packet inter-arrival times and packet sizes from a web page fetch is used to create a profile for each web page in a target set and the cross correlation between an observed traffic sequence and the stored profiles is then used as a measure of similarity.   

Most later work has adopted essentially the same model as \cite{bissias06}, making use of packet direction and size information and assuming that the packet stream has already been partitioned into individual web page fetches.  For example in \cite{wang14} the timing information is not considered in the feature set, hence the attack can be countered with defences such as BuFLO in \cite{dyer12} leading to a success rate of only $10\%$. In  \cite{liberatore06,herrmann09} Bayes  classifiers based on the direction and size of packets are considered while in \cite{panchenko11} an SVM classifier is proposed.  In \cite{lu10} classification based on direction and size of packets is studied using Levenshtein distance as the similarity metric, in \cite{miller14} using a Gaussian Bag-of-Words approach and in \cite{wang14} using $K$-NN classification.  In \cite{cai12} using a SVM approach a classification accuracy of over 80\% is reported for both SSH and Tor traffic and the  defences considered were generally found to be ineffective.   Similarly, \cite{dyer12} considers Bayes and SVM classifiers and finds that a range of proposed defences are ineffective.  In \cite{gong10} remote inference of packet sizes from queuing delay is studied.   

Recently, powerful attacks based on using packet timing information alone have been demonstrated \cite{feghhi161}.  Such attacks are immune to defences based on packet padding etc, do not require a priori knowledge of the start/end of a web fetch and have low false positive rates in open-world settings where the set of web sites being browsed is not restricted to a defined set.

\subsubsection{Defences}
With regard to defences, one fairly obvious approach to obfuscating the timing of information packets is to transmit packets at a constant rate, inserting dummy packets when an information packet is not available to send (when encrypted dummy packets are indistinguishable from information packets to an attacker), and buffering information packets when a transmit slot is not available.  This type of strategy is considered in detail by \cite{dyer12}, although to reduce the bandwidth overhead transmission only continues until either a page has loaded or a minimum time has elapsed.   It is concluded that this approach, referred to as BuFLO (Buffered Fixed-Length Obfuscator), involves excessive bandwidth overhead ($>$100\%) and high latency (a $\times$2-3 increase) in return for uncertain security (the duration of pages longer than the minimum time is still revealed).   Tamaraw \cite{cai14}  and CS\_BuFLO \cite{cai14b} both attempt to refine BuFLO.  Tamaraw ends transmission only at fixed multiples of a design parameter $L$ which allows for stronger security guarantees at the cost of increased overhead.  CS\_BuFLO adjusts the transmit rate based on traffic load.

In view of the high overheads associated with constant rate transmission, a number of alternative approaches have been considered.  In HTTPOS \cite{luo11} and the work of Wright et al \cite{wright09} packet timing is modified by injecting dummy HTTP requests and randomising the pipelining of the requests forming the fetch a web page.  Randomised pipelining is also used in Tor.  However, this defence has been shown to be of limited effectiveness at protecting against attacks in several subsequent evaluations \cite{cai12,wang14,juarez14}.    Another line of work is based on shaping traffic so that the transmitted trace is similar, in some appropriate sense, to that for a different web page or application.  In early work \cite{wright09} traffic morphing is considered whereby packet padding/splitting is used so that the distribution of packet sizes over a web fetch resembles that of a target web page.  However this was found to be ineffective against later attacks \cite{cai12,wang14} and packet padding also provides little defence against attacks based only on timing information such as that of \cite{feghhi161}.  In \cite{nithyanand14} anonymity sets are created by clustering the packet traces for web page fetches and morphing them to look like the centroid of their cluster.  A similar approach is considered by \cite{wang14,wang15}.  However, in addition to still entailing a relatively high bandwidth overhead ($>$60\%) these approaches also require maintaining a database of packet traces that may be both large and strongly dependent on the location of a client device (e.g. when moved from work to home, or between hot spots with different qualities of network connection).

On links carrying multiple web fetches interleaving of packets from these fetches can be expected to make it harder for an attacker to infer the web pages being browsed.  Intuitively, the inference task of the attacker becomes that of estimating the components of a mixture, which can be expected to increase in difficulty as the number of components (web fetches) increases.   The great advantage of this approach is that it avoids the costs associated with buffering and injecting dummy packets.   However, it relies on the availability of sufficiently many active fetches at any given time.  Further, while only a limited number of studies have been carried out that evaluate the protection provided by combining multiple web fetches, generally the results point towards this approach being of limited effectiveness against modern attacks, see for example \cite{cai12}.  
In summary,  it is difficult to quantify the degree of privacy provided at any given time when relying on interleaving of multiple fetches, and due to its opportunistic nature it is probably inherently impossible to provide up front guarantees of privacy to users using this type of approach.

\section{Achieving Privacy}
\label{sec:privacy}


{
In this section we discuss how traffic shaping can be used to ensure indistinguishability i.e that a given packet trace can be explained by several combinations of web fetches.   While this can be trivially achieved by sending packets at a constant rate and using buffering plus insertion of dummy packets as required, the overhead involved is excessive.  The challenge is to strike a balance between indistinguishability and the associated traffic shaping overheads of buffering delay and dummy packets.   

We make two main observations.  Firstly, traffic shaping can be adapted on the fly using what we refer to as a trace-based approach.  As alternative is offline optimisation of a traffic profile using a database of web fetch packet traces.  However, while this can lead to traffic profiles which use a lower number of dummy packets, it can be difficult to deal with previously unseen web pages.   Secondly, rather than applying traffic shaping to each web fetch individually we instead apply shaping to the aggregate traffic sharing a tunnel.    This small change has significant consequences.   When shaping is applied to web fetches individually then the overhead per fetch is fixed and thus the overall tunnel overhead of scales with the number of fetches e.g. if the overhead is 100\% for a fetch of 100 packets then for 100 fetches the overhead is 10,000 packets.   In contrast, when shaping is applied to the aggregate traffic then as the traffic load increases the overhead falls to zero since there are always useful packets available to send and so no need to insert dummy packets.

The rest of the section proceeds as follows.  In Section \ref{sec:Indistinguishability} we discuss indistinguishability in more detail.  Then in Section \ref{sec:traces} we introduce our online adaptive trace-based approach and in Section \ref{sec:over} the reduction in dummy packet overhead obtained by shaping of aggregated traffic.
}

\begin{figure}
\centering
\includegraphics[width=\columnwidth]{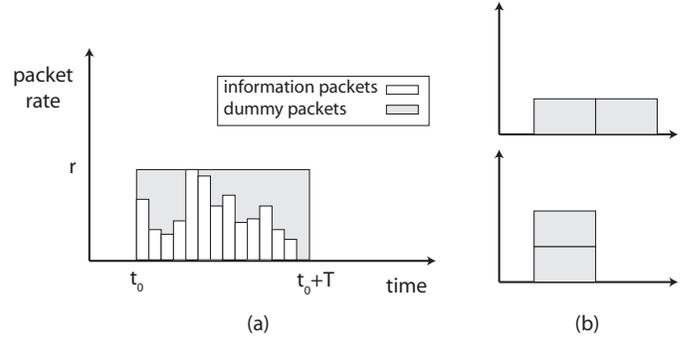}
\caption[Use of dummy packets to construct an uninformative pattern]{Illustrating use of dummy packets to force the number and timing of transmitted packets to conform to an uninformative pattern or ``trace''. When insufficient information packets are available, dummy packets are inserted as needed {as indicated in figure (a)}.  When too many information packets are available they are buffered until a transmission opportunity becomes available. {Additional traces can be activated when the duration of a web fetch exceeds that of a single trace, top figure in (b), or when the rate is greater than that of a single trace, bottom figure in (b). }}\label{fig:trace_schematic}
\end{figure}

\subsection{Indistinguishability}\label{sec:Indistinguishability}

As already noted, we are not concerned with concealing the fact that the user is browsing the web but rather with concealing the particular pages visited.   
Privacy is therefore embodied in the indistinguishability of sequences of web fetches given with regard to the packet sequences which they generate.    


It is important to note that we do not insist that fetches of single web pages individually be indistinguishable.    Rather, we require that the observed packet sequence corresponding to the fetch of each web page can be reasonably explained by other combinations of web fetches.   So, for example, the packet sequence generated by the fetch of a large web page could equally have been generated by sufficiently many sequences of fetches of small web pages.   Fetches of single web pages must still be indistinguishable from one another when they cannot plausibly be generated by combinations of other page fetches.  That is, we require fetches of sufficiently small web pages to be atomic in the sense that they are indistinguishable from one another.


\subsection{Defence By Use of ``Traces''}\label{sec:traces}
\label{sec:traces}
The basic problem with the packet sequence generated by a web fetch is that it contains many packets (often several hundred, frequently more) and so amounts to an observation of a point in high dimensional space.  It is increasingly recognised that high dimensional data carries a major risk of de-anonymisation since data points tend to be sparsely distributed in high dimensions (each point can be individually distinguished).  Our approach is therefore to reduce the dimension of the observed packet sequence data.  We do this by gathering packets into larger groups and transmitting these groups in a uniform way.   

{
The approach is illustrated in Figure \ref{fig:trace_schematic}a.   Here, a web fetch starts at time $t_0$.   Dummy packets (indicated in grey) are inserted so that the observed sequence of packets, which is a combination of dummy plus user packets, has a uniform profile and duration.   We refer to this uniform profile as a ``trace''.   Formally, a trace consists of a sequence of $\{(\tau_i,p_i),i=1,2,\dots,n\}$, with $\tau_j>\tau_i$ for $i>j$.    This says that at time $\tau_i$ after the start of the trace {$p_i$ packets are to be transmitted.  The duration of the trace is $\tau_n-\tau_1$ and the number of packets contained in the trace is $P=\sum_{i=1}^n p_i$.   In the rest of this paper we will use traces with uniform $p_i=p$, $i=1,\dots,n$ and evenly spaced $\tau_i=(i-1)T/n$, where $T$ is the trace duration.}}   The rate $p$ and duration $T$ of the trace are design parameters, that we will return to later\footnote{We might also use a trace where the $p_i$ are not constant and the $\tau_i$ are not evenly spaced, and also allow the trace used to be drawn from a pre-defined set of traces.}.   In this example the web fetch completes by time $t_0+T$ and so fits within a single trace.  However, if the duration exceeded $T$ then a second trace would be started so as to mask the additional user packets, see upper schematic in Figure \ref{fig:trace_schematic}b.   Similarly, if the rate at which user packets arrive much exceeds the rate of a trace then two traces can be started in parallel, see lower schematic in Figure \ref{fig:trace_schematic}b.   With this approach the observed packet sequence now consists of a sequence of traces.

\begin{figure}
\centering
\includegraphics[width=0.65\columnwidth]{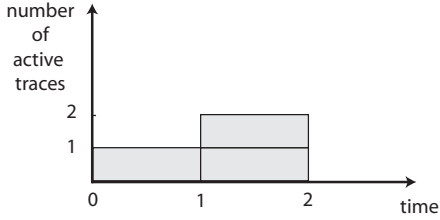}
\caption{Example illustrating an observed sequence of traces.}\label{fig:traceexample}
\end{figure}

We note that \cite{nithyanand14,wang14,wang15} also consider indistinguishability, doing this through the creation of so-called anonymity sets by clustering the packet traces for web page fetches.   However the approach considered here differs in a number of important and fundamental respects.  One is that the sequence of traces is constructed on the fly as packets are transmitted.  There is no need to pre-cluster web fetches and construct a trace in advance for each cluster, therefore also no need to maintain a history of web fetches (needed to perform pre-clustering) nor to explicitly adapt to location and time (which may change the link characteristics and so the packet traces and their clusters).   A second key difference is that packets from multiple web fetches/flows can be packed together into the same trace, thereby allowing us to opportunistically reduce the number of dummy packets transmitted during busy times when many flows are simultaneously active while maintaining well-defined privacy properties.   We discuss these in more detail below, where we show that they allow a tunnel based on our trace-based approach to be capacity achieving i.e. as the traffic load increases the number of dummy packets sent falls to zero.

Given an observed sequence of traces, the indistinguishability question becomes: how many different sequences of web fetches can reasonably generate the observed trace sequence.    {The analysis is basically that of a packing problem, namely given an observed sequence of traces what combinations of web fetches can be packed into that sequence.    Since all web pages may not be equally likely to be fetched, if we know the probability of fetching each page we can use this to calculate the probability of each combinations of web fetches.   For plausible deniability we need there to be several combinations with reasonably high probability.   There are two main difficulties with actually carrying out such calculations however.  One is that the probability of fetching a page is usually unknown and user dependent, although it might be estimated from historical data.  The other is that the pages fetched may be correlated, and again would have to be estimated from historical data.   With these caveats in mind, we give a simple example calculation for illustrative purposes. } Consider the sequence of traces shown in Figure \ref{fig:traceexample}.   Time is slotted, with each slot corresponding to the duration of a trace.  Traces start at the beginning of a slot, and suppose a maximum of $n$ traces are possible in each slot (e.g. limited by the link capacity).   Suppose we have a set $W$ of web pages of interest.  Of these, fetches for web pages $W_1\subset W$  can be covered by one trace, $W_2\subset W$ by two traces and so on.   Let $q_w$ be the probability that web page $w\in W$ is fetched and assume, for simplicity, that pages are fetched independently.  Let $X_w$ be a random variable which takes value 1 when web page $w$ is fetched and $0$ otherwise.  Let $T\subset \{0,n\}^k$ denote the observed sequence of traces of length $k$ slots, with $T=(1,2,0,\dots,0)$ in Figure \ref{fig:traceexample} .   Then,
\begin{align}
\textsc{P}(X_w=1 | T, w\in W_1) &= 1-\prod_{j=1}^k\left(1-\frac{q_w}{\sum_{i\in W_1} q_i}\right)^{T_j}
\end{align}
and similarly we can calculate $\textsc{P}(X_w=1 | T, w\in W_2)$ etc.   In the example in Figure \ref{fig:traceexample} suppose $w_w=1/|W|$.  Then for a web page $w\in W_1$ that fits inside a single trace the probability that it was fetched is $1-\left(1-\frac{1}{|W_1|}\right)\left(1-\frac{1}{|W_1|}\right)^2$.  For example, when $|W_1|=100$ this probability is $0.03$ and when $|W_1|=1000$ the probability is $0.003$ and it can be seen that provided the cardinality of $W_1$ is reasonably large a user can with high plausibility deny that they fetched a given web page.   


Note that in the above example an attacker can infer that the user did not fetch any web pages that cannot fit inside three traces or fewer.   However, given the extremely large number of possible web pages and the fact that most of these pages are not fetched by a given user we argue that the absence of a fetch for a web site is much less informative to an attacker than the presence of a fetch.   

We note also that the option exists to inject dummy traces into the tunnel link to add a further level of deniability.   Using such an approach the potential exists to formulate indistinguishability as a form of differential privacy.  Namely, such that the addition of the fetch of any individual web page has limited impact on the sequence of traces observed on the link.   However, we leave this as future work and do not pursue it further here. 


\subsection{Impact of Rate and Duration of a Trace on Overhead}\label{sec:over}

\begin{figure}
\centering
\begin{subfigure}[]{0.46\columnwidth}
\includegraphics[width=0.9\columnwidth]{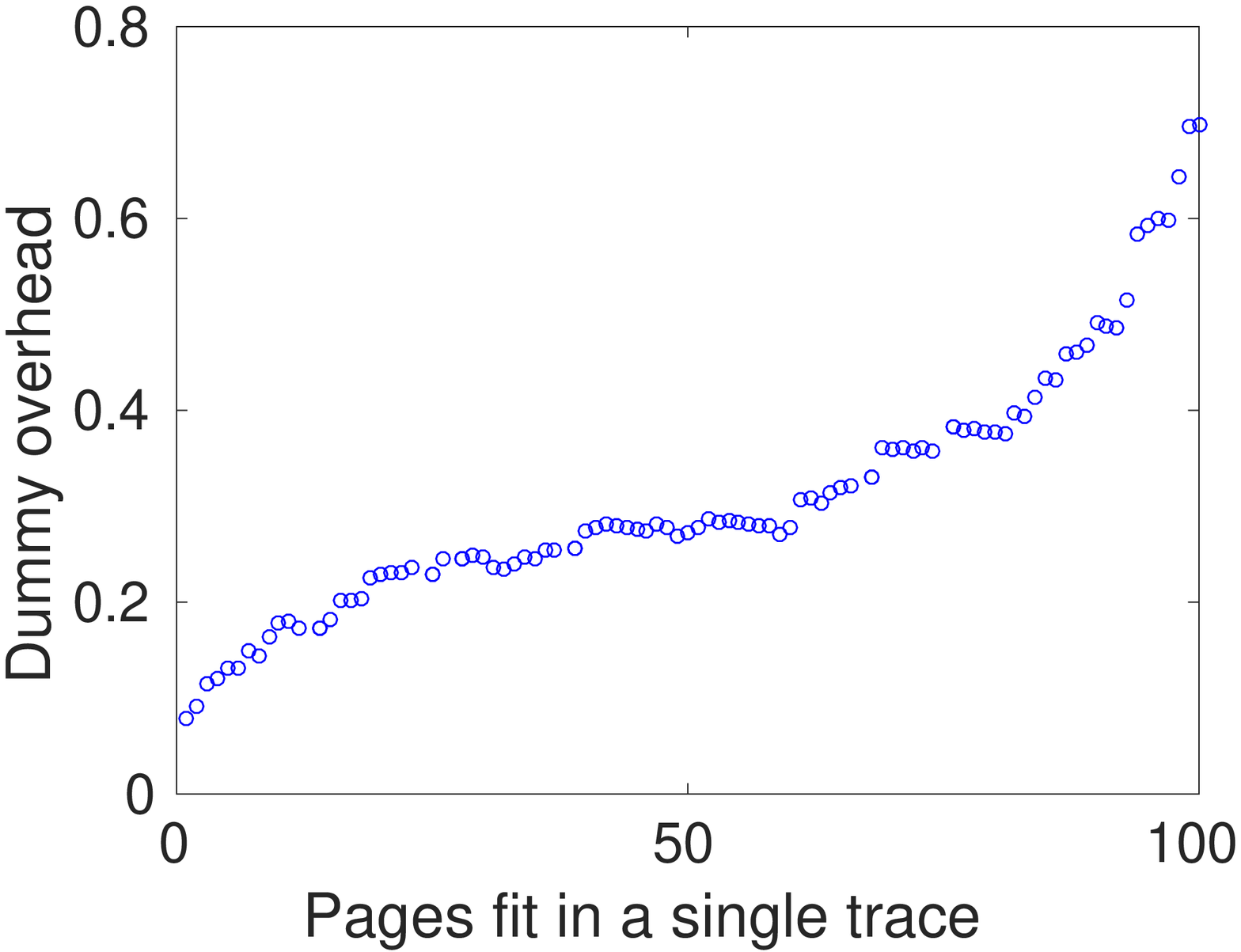}
	\caption{}
	\label{fig:dummy_vs_fitted_pf}
\end{subfigure}
\begin{subfigure}[]{0.46\columnwidth}
\includegraphics[width=0.9\columnwidth]{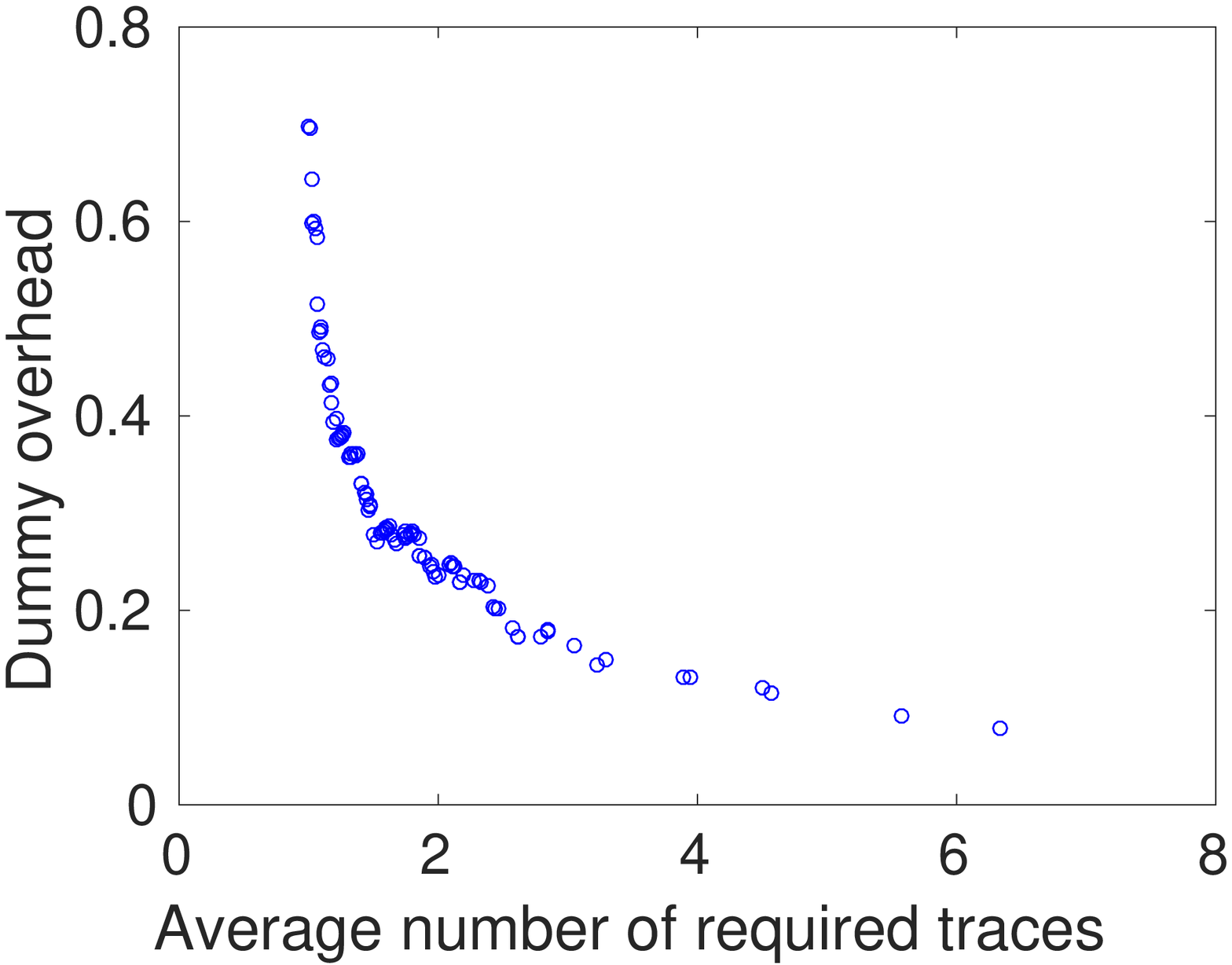}
	\caption{}
	\label{fig:dummy_vs_fitted_at}
\end{subfigure}
\caption[Overhead of dummy packets vs trace duration for 100 web pages]{Illustrating the overhead of dummy packets vs trace duration. Data is for Alexa's top 100 web pages. }\label{fig:dummy_vs_fitted}
\end{figure}

By buffering user packets at the tunnel ingress until enough are available to fill a trace we can avoid the overhead of sending dummy packets.  However, this comes at the cost of increased delay, perhaps much increased delay.   However, in modern networks quality of service is frequently achieved by over-provisioning of network capacity and as a result bandwidth is often relatively plentiful.   Conversely, users are known to be sensitive to delay when using online services.  For example, Amazon estimates that a 100ms increase in delay reduces its revenue by 1\% \cite{amazon}, Google measured a 0.74\% drop in web searches when delay was artificially increased by 400ms \cite{google} while Bing saw a 1.2\% reduction in per-user revenue when the service delay was increased by 500ms \cite{bing}.  Hence, in order to enhance resistance to traffic analysis attacks it is preferable to sacrifice some bandwidth by sending dummy packets rather than incurring excessive delay.

To gain some insight into the overhead of dummy packets associated with different durations of trace we fetched the home pages from the Alexa top 100 finance and health web sites in Ireland.   Figure \ref{fig:dummy_vs_fitted_pf}  plots the average fraction of packets in a trace which are dummy packets vs the number of web sites that fit inside the trace.   A trace of duration 10890ms is needed to cover all 100 web sites and a trace of duration 4947ms to cover 50 web sites.   As might be expected, the fraction of dummy packets increases as the duration of the trace is increased to include more web pages.    When all 100 web pages fit inside a single trace the overhead is around 70\% but this falls to around 30\% for a trace that covers 50 web pages.  Figure \ref{fig:dummy_vs_fitted_at} plots the fraction of dummy packets vs the maximum number of consecutive traces used to cover all 100 web pages (so the value for one trace corresponds to the data in Figure \ref{fig:dummy_vs_fitted_pf}).  It can be seen that as the number of traces used increases the overhead falls.  However, use of smaller traces reduces the level of indistinguisability provided and so a trade-off exists between privacy and dummy packet overhead.   

\begin{figure}
\centering
\includegraphics[clip,width=0.65\columnwidth]{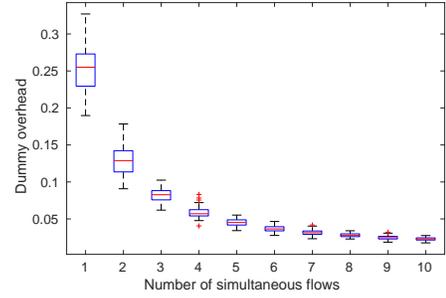}
\caption{Illustrating the overhead of dummy packets vs the number of simultaneous flows. Each flow is a sequence of fetched pages randomly selected from Alexa's top 100 web pages. Uniform traces of duration 5000ms and rate of 0.2 pkt/s are activated every 5 seconds, number of which corresponds to the incoming traffic load. Each box plot represent 100 measurements.}\label{fig:dummy_vs_flows}
\end{figure}

Importantly, the trade-off between privacy and dummy packet overhead also depends on the number of flows simultaneously active, since packets from multiple flows can be packed together into the same trace thereby reducing the number of dummy packets needed without compromising privacy.  This is illustrated in Figure \ref{fig:dummy_vs_flows}, which plots the fraction of dummy packets against the number of simultaneous flows (here a flow is a sequence of fetched pages randomly selected from Alexa's top 100 web pages).  It can be seen that as the number of flows increases, the level of dummy traffic falls towards zero.   We will show analytically in the next section that the trace-based method is, in fact, capacity achieving as the traffic load increases, in contrast to previous traffic shaping approaches that come with privacy guarantees.   


Since it is hard to analytically quantify the trade-off between privacy, dummy packet overhead and delay (the delay aspect is especially difficult to analyse mathematically), we will revisit this trade-off shortly using experimental data.

\section{Throughput and Dummy Packet Overhead}

A privacy-enhanced tunnel using the trace approach starts up traces\footnote{There is no need to manage the termination of traces, each trace simply stops once its alloted time has expired.} dynamically so as to manage user delay and dummy packet overhead while ensuring privacy.     In this section we consider policies for start-up of traces and analyse how the dummy packet overhead changes with the traffic load (analysis of delay is left to the experimental section below).  We develop a scheduling policy that has the following important properties:
\begin{enumerate}
 \item \emph{Adaptive}.  The number of dummy packets transmitted adapts with the traffic load.  When the traffic load is light the scheduling policy inserts dummy packets as needed to fill out the active packet trace(s) and thus maintain privacy.
 \item \emph{Capacity Achieving}. The number of dummy packets falls to zero as the traffic load increases i.e. user packets are able to make full use of the available network throughput capacity as traffic load increases.   
\end{enumerate}

Our analysis is made complicated by (i) the non-preemptive nature of traces, namely once a trace has started it will continue until its alloted time has expired, and (ii) a new trace can start at any time and so when multiple traces are active their start (and so end) times need not be not aligned i.e. the traces may only partially overlap.  We therefore begin by relaxing (ii) and requiring the synchronised start/stop of traces, and later relax this requirement.

\subsection{Network Setup and Notation}
\begin{figure}
\centering
\includegraphics[width=0.95\columnwidth]{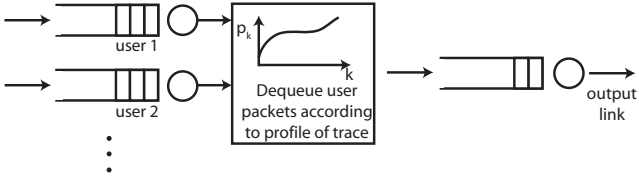}
\caption{Schematic illustrating the VPN gateway setup considered.}\label{fig:setup}
\end{figure}

The setup considered is illustrated in Fig \ref{fig:setup}.  Let $\U=\{1,2,\cdots,n_u\}$ denote the set of users, $\vi{a}{u}_{k}\in\mathbb{N}$ denote the number of packet arrivals for user $u$ at time $k$ and $\vi{\bar{a}}{u}:=\lim_{K\rightarrow\infty}\frac{1}{K}\sum_{k=1}^K \vi{a}{u}_k$ the average number of packet arrivals.  Packets for each user $u$ are held in separate queues, with the queue occupancy for user $u$ at time $k$ being denoted by $\vi{q}{u}_k$.  Letting $\vi{w}{u}_{k}\in\mathbb{N}$ be the number of packets dequeued at time $k$ then $\vi{q}{u}_{k+1}=[\vi{q}{u}_k+\vi{a}{u}_k-\vi{w}{u}_k]^+$ where the notation $[\cdot]^+$ has the usual meaning that $[x]^+=x$ when $x\ge0$ and $[x]^+=0$ otherwise.

User packets are dequeued according to a predefined \emph{trace}.  Multiple traces may be active simultaneously in order to increase the rate at which user packets are dequeued.  A trace $f$ is a sequence ${p}_j\in\mathbb{N}$, $j=1,2,\dots,{n}$  where ${p}_j$ is the number of packets to be transmitted at time $j$ and ${n}$ is the duration of the trace.  Importantly, when no user packets are available to send at time $j$, then dummy packets will be transmitted to ensure that ${p}_j$ packets are always transmitted.  Let ${P}=\sum_{j=1}^{{n}} {p}_j$ denote the total number of packets in a trace.   

Packets from a trace are queued at the output link before transmission over the tunnel, see Fig \ref{fig:setup}.  As is usual in Internet links we leave rate control to the end hosts.  If the aggregate send rate persistently exceeds the output link capacity then the queue at this link will eventually overflow and cause packet loss which end hosts can then use as a congestion indicator, e.g. by using TCP congestion control.

New traces are started as needed to service user packet arrivals, and multiple traces may be active at the same time so as to increase the sending rate if needed.   Indexing these active traces by $1,2,\cdots$ let $\T$ denote the set of trace indices, $\tau_t$ the start time of trace $t\in\T$.  At time $k$ the set of active traces is $\T_k:=\{t\in\T: k\in\{\tau_t,\dots,\tau_t+{n}\}\}$ and the number of packets that are transmitted at time $k$ is $\sum_{t\in\T_k} {p}_{k-\tau_t}$, therefore $\sum_{u\in\U} \vi{w}{u}_k \le \sum_{t\in\T_k} p_{k-\tau_t}$.   The average rate of packet transmissions can be no more than the capacity $c$ in packets/slot of the outgoing link i.e. $\lim_{K\rightarrow\infty} \frac{1}{K}\sum_{k=1}^K \sum_{t\in\T_k} {p}_{k-\tau_t} < c$.  

Since $\sum_{t\in\T_k} {p}_{k-\tau_t}$ packets are sent then if there are insufficient user packets available to send at time $k$ we need to send $\sum_{t\in\T_k} {p}_{k-\tau_t}-\sum_{u\in\U} \vi{w}{u}_k$ dummy packets.  Our task is to activate traces and thereby adjust the number of user packets transmitted $\vi{w}{u}_{k}$ so as to minimise the number of dummy packets sent while still servicing all of the user packets in a timely manner.


\subsection{Synchronised Scheduler}

To proceed we first assume that active traces start/stop in a synchronised fashion.   Since traces started in slot $1$ finish at slot $n$, traces started at slot $n+1$ finish at slot $2n$ and so on then by restricting ourselves to starting traces at times $1,n+1,2n+1,\dots$ we can ensure synchronised start/stop of traces and avoid partial overlapping of traces.   We will relax this restriction later, but the absence of overlapping traces simplifies the initial analysis and assists with gaining insight into the design issues to be considered.

\begin{figure}
\centering
\includegraphics[width=0.4\columnwidth]{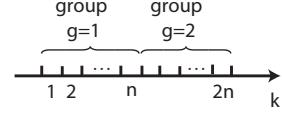}
\caption[Gathering slots into groups of $n$ consecutive slots]{Gathering slots into groups of $n$ consecutive slots.  In the synchronised case a trace can only be activated in the first slot of a group, and will have completed by the last slot in that same group since traces are of duration $n$ slots.}\label{fig:sync_ex}
\end{figure}

Formally, partition time slots $k=1,2,\dots$ into groups of $n$ slots, see Fig \ref{fig:sync_ex}.  Denote the first group of slots by $\G_1:=\{1,\dots,n\}$, the second by $\G_2:=\{n+1,\dots,2n\}$ and so on.  A new trace can only be activated at the start of a group of slots.    
Letting $y_g$ denote the number of traces active in group $g$, consider the following trace scheduler:

 \begin{align}
&y_g\in\  \arg\min_{y\in [0,y_{max}]}(\gamma-Q_gP  )x \label{eq:delta1}\\
&Q_{g+1} = [Q_g+b_{g} - {y}_{g}P]^+\label{eq:mult1}
\end{align}
where $\gamma >0$ is a design parameter, ${b}_g:=\sum_{u\in\U}\sum_{k\in \G_g}\vi{a}{u}_k$ is the number of user packet arrivals in group of slots $g$ and recall $P$ is the total number of packets in a trace.  

The following lemma establishes that the scheduler maintains a bounded queue for all packet arrival rates up to $y_{max}P$.
\begin{lemma}[Queue Boundedness]\label{lem:bounded}
Consider update (\ref{eq:delta1})-(\ref{eq:mult1}).  Suppose that $b_g\le b_{max}$, $\bar{b}=E[{b}_g]$ and $y_{max}P>\bar{b}+\epsilon$, $\epsilon>0$. Then as $g\rightarrow \infty$ we have $|\gamma-Q_gP|\le2\delta$ where $\delta:={b}_{max}+y_{max}P$.
\end{lemma}
\begin{proof} 
Case(i) $\gamma-Q_gP>\delta$.  Then $y_g=0$ and $Q_{g+1} = [Q_g+b_{g}]^+=Q_g+b_g$ since $Q_g, b_g>0$.   Hence, $\gamma-Q_{g+1}P=\gamma-Q_gP-b_gP$.   Taking expectations conditioned on $Q_g$, $\gamma-E[Q_{g+1}|Q_g]P=\gamma-Q_gP-\bar{b}P$

Case (ii) $\gamma-Q_gP<-\delta$.  Then $y_g=y_{max}$ and $Q_{g+1} = [Q_g+b_{g}-y_{max}P]^+$.  Since $\gamma-Q_gP<-\delta$ then $Q_g>(\gamma+\delta)/P>y_{max}$.  Hence, $Q_{g+1}=Q_g+b_{g}-y_{max}P$ and $\gamma-Q_{g+1}P=\gamma-Q_gP-b_{g}P+y_{max}P^2$.  Taking expectations, $\gamma-E[Q_{g+1}|Q_g]P=\gamma-Q_gP-\bar{b}P+y_{max}P^2\ge\gamma-Q_gP+\epsilon P$.

Case (iii) $|\gamma-Q_gP|<\delta$.  Since $Q_{g+1}\ge 0$ then $\gamma-Q_{g+1}P\le \gamma$.  For the left-hand bound, we use the fact that $|Q_{g+1} -Q_g|\le b_{max}+y_{max}P$ and so $\gamma-Q_{g+1}P \ge -\delta-(b_{max}+y_{max}P)$.  

Hence, when $|\gamma-Q_gP|>\delta$ then $\gamma-E[Q_{g+1}|Q_g]P$ is strictly decreasing and when $|\gamma-Q_gP|\le\delta$ then it remains in the ball $|\gamma-E[Q_{g+1}|Q_gP|\le 2\delta$.    
\end{proof}

We have the following immediate consequence.  
\begin{corollary}[Capacity Achieving]
Suppose $y_{max}>\bar{b}/P$.  Then the mean service rate provided by the traces is sufficient to serve the mean rate of user packet arrivals.   That is, the trace-based tunnel with scheduler (\ref{eq:delta1})-(\ref{eq:mult1}) is capacity achieving.
\end{corollary}
\begin{proof}
Observe from (\ref{eq:mult1}) that,
\begin{align}
&{Q}_{g+1} \ge {Q}_g+b_g - {y}_g P
\end{align}
and so when $Q_1=0$ it follows that
\begin{align}
\frac{1}{g} \sum_{i=1}^{g-1} (b_i - {y}_iP)=\bar{b}_g-y^\diamond_g P
 \le Q_g/g
\end{align}
where $\bar{b}_g:=\frac{1}{g} \sum_{i=1}^{g-1}{b}_g$ and $y^\diamond_g:=\frac{1}{g} \sum_{i=1}^{g-1}{y}_g$.  Since, by Lemma \ref {lem:bounded}, $Q_{g}$ is bounded then $\frac{Q_{g}}{g}\rightarrow 0$ as $g\rightarrow\infty$ and therefore $\bar{b}_g \le {y}^\diamond_gP$ as $g\rightarrow\infty$.   That is, the mean service rate provided by the traces is sufficient to serve the mean rate of user packet arrivals.   This holds for every $\epsilon>0$ such that $y_{max}P>\bar{b}+\epsilon$.
\end{proof}

\subsection{Discussion}
The foregoing analysis establishes that the trace-based tunnel approach is capacity achieving.  Intuitively, as the traffic load increases the probability increases that a user packet is available in the queue when a transmission slot arises and so the number of dummy packets decreases, eventually falling to zero.

Scheduler (\ref{eq:delta1})-(\ref{eq:mult1}) is easy to analyse but rather harsh in that it selects either $y_{max}$ or $0$ active traces in each group of slots.  This can be alleviated by modifying the scheduler as follows:
 \begin{align}
&x_g\in\  \arg\min_{x\in [-1,1]}(\gamma-Q_gP  )x \label{eq:delta2}\\
&{y}_{g+1}=[{y}_g+ x_g]^{[0,y_{max}]}\\
&Q_{g+1} = [Q_g+b_{g+1} - {y}_{g+1}P]^+\label{eq:mult2}
\end{align}
where $[x]^{[0,y_{max}]}=x$ when $x\in[0,y_{max}]$ and equals 0 when $x<0$ and $y_{max}$ when $x>y_{max}$.   This scheduler is intuitive, although the use of incremental updates to $y_g$ is somewhat unusual.  Namely, queue $Q_g$ measures the accumulated mismatch between the arrival of user packets and service of packets.  When this mismatch grows too large and exceeds $\gamma/P$ then we increase the number of active traces, when the mismatch falls to less than $\gamma/P$ we decrease the number of active traces.   

%
%



We can also relax the synchronisation assumption made in the previous section by allowing new traces to start in any slot.  This means that traces can, for example, now partially overlap.   The payoff is that the potential exists to exploit this extra freedom to achieve improved performance, especially improved delay performance as a result of being able to start new traces in a more timely way.

As our baseline unsynchronised scheduler we use the following update,
 \begin{align}
&x_k\in\  \arg\min_{x\in [-1,1]}(\gamma-Q_kP  )x \\
&{z}_{k+1}=[{z}_k+  x_k]^{[0,y_{max}]}\\
&\hat{Q}_{k+1} = [\hat{Q}_k+\alpha(\sum_{u\in\U}\vi{{a}}{u}_k - \sum_{t\in\T_k} p_{k-\tau_t}) ]^+\label{eq:mult3}
\end{align}
where a new trace is started in slot $k$ (added to $\T_k$) when $z_k$ increases or when a trace completes and the number of active traces $|\T_k|$ falls below $z_k$.

\section{Implementing A Fully Functional Protoype}

We built a prototype VPN implementation of the scheduler in (\ref{eq:delta2})-(\ref{eq:mult2}) and also of its unsynchronised variant. The entire project including codes, scripts and datasets for all measurements in this paper is available at \cite{implementation}.   {Many of the evaluation studies of proposed traffic analysis defences in the literature use a simulation approach whereby previously collected packet traces are replayed.  While easy to implement and fast to run, this replay approach ignores a number of important effects when traffic shaping is used.  Namely, the buffering associated with traffic shaping affects the behaviour of the web fetch process, e.g. delay of a request has knock-on effects on the timing of the response which may in turn affect subsequent requests, and such interactions are not captured by simply replaying a recorded packet sequence.  Further, traffic shaping can interact with TCP congestion control in unexpected ways that impact on timing and performance.  Translating performance data, not only on latency and dummy packet overhead but also on resistence to attack, derived from replay-based evaluations to real-world usage can therefore be far from straightforward.  While developing a full implementation involves significant extra time and effort, we argue that it allows more accurate prediction of real-world performance.

In this section we begin by describing the hardware and software setup used in our tests and outlining the implementation approach used.  We then discuss how the flexibility in the scheduler design noted in the previous section can be exploited to improve latency and how we carry out selection of trace parameters (duration and rate).   We finish by presenting experimental measurements for both UDP and TCP traffic.   These measurements allow us to evaluate the trade-off between the delay and dummy packet overheads and how this can be controlled via the choice of scheduler parameters. 
}

\subsection{Hardware Setup}
\begin{figure}
\centering
\includegraphics[width=0.8\columnwidth]{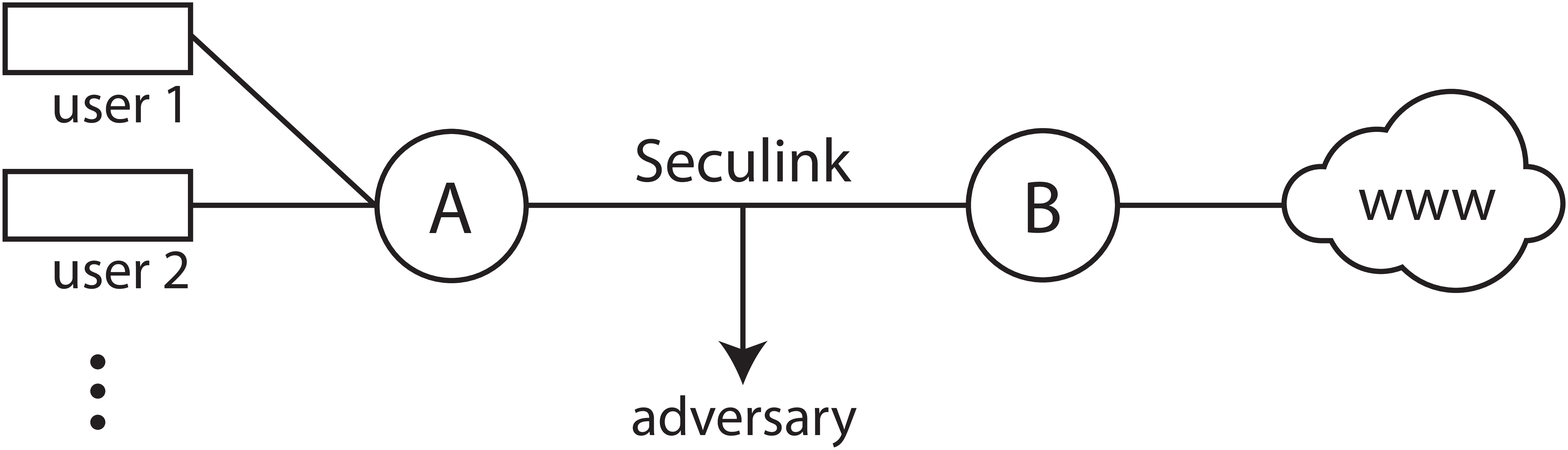}
\caption{Schematic illustrating the tunnel setup considered.}\label{fig:tunnel}
\end{figure}

We begin by describing the hardware setup used.   The network topology is shown schematically in Figure \ref{fig:tunnel}.	 Clients are connected to the VPN gateway A using an IPSec protected link. The link between nodes A and B (labelled Seculink) is the encrypted privacy-enhanced VPN tunnel.  This is a public link, \textit{i.e.} it is exposed to sniffing from adversaries.  Seculink gateway A is a commodity server with an Intel(R) Core 2 Quad @2.66GHz CPU and 4GB RAM running Ubuntu 14.04.4 LTS.  The other end of the tunnel, node B, is another commodity server with an Intel(R) Core 2 Duo @ 3.00GHz and 2GB RAM running Ubuntu 12.04.5 LTS.   Nodes A and B use a mix of TP-LINK TG-3468 1000Mbps external ethernet cards and on-board gigabit ethernet cards (RTL8111 PCI Express and Intel 82567LM-3 respectively) for network connectivity.  Traffic from node B is routed to the Internet via a campus gateway using a 100Mbps link (it is this link which limits the network rate since the tunnel between A and B is a gigabit connection).   A NETGEAR JGS516 Gigabit switch is used to connect client machines to Seculink gateway A.  The Seculink tunnel uses IPSec and traffic shaping to protect against DPI and traffic analysis attacks. Outgoing traffic from clients is sent to node B via node A and then forwarded by B to the campus gateway.  Incoming traffic arriving at node B is forwarded to node A and then sent to the corresponding client. {The proposed tunnel is designed to perform the scheduling and traffic shaping on aggregated traffic from multiple clients. However in order to compare traffic patterns before and after traffic shaping we use a single client generating multiple flows in the experiments of this section unless stated otherwise. The client is a Dell Inspiron 15 5000-Series laptop with Intel Core i3-4005U Processor and 4GB RAM running an Ubuntu 14.04.4 LTS.}

\subsection{Software Setup}
Both Seculink and the links between clients and node A are protected with IPSec protocols AES-256 and SHA-256.  Traffic shaping on Seculink is implemented using the \texttt{netfilter} and \texttt{netfilter}\_\texttt{queue} libraries in C in conjunction with the \texttt{iptables} module in Linux.  Gateway A queues packets from clients and transmits them over Seculink using our traffic shaping scheduler with slots of 1ms duration.  When no client packets are queued for transmission yet the active traces require a packet to be sent then a dummy packet is generated and sent.  Traffic received at node B is first filtered to remove dummy packets and remaining packets are then forwarded to the campus gateway over an unencrypted link.  Responses received from the Internet are treated similarly, transmitted by node B to node A using our traffic shaping scheduler.   In the UDP experiments traffic is generated using the PERIODIC method of the MGEN traffic generator.   TCP traffic is generated by using \texttt{wget} to fetch dummy files of various sizes from a server located in the campus network.

\subsection{Synchronised vs Unsynchronised Schedulers}
\label{sec:syncvsunsync}
\begin{figure}
\centering
\begin{subfigure}[Synchronised]{0.5\columnwidth}
  \includegraphics[clip, width=0.98\columnwidth]{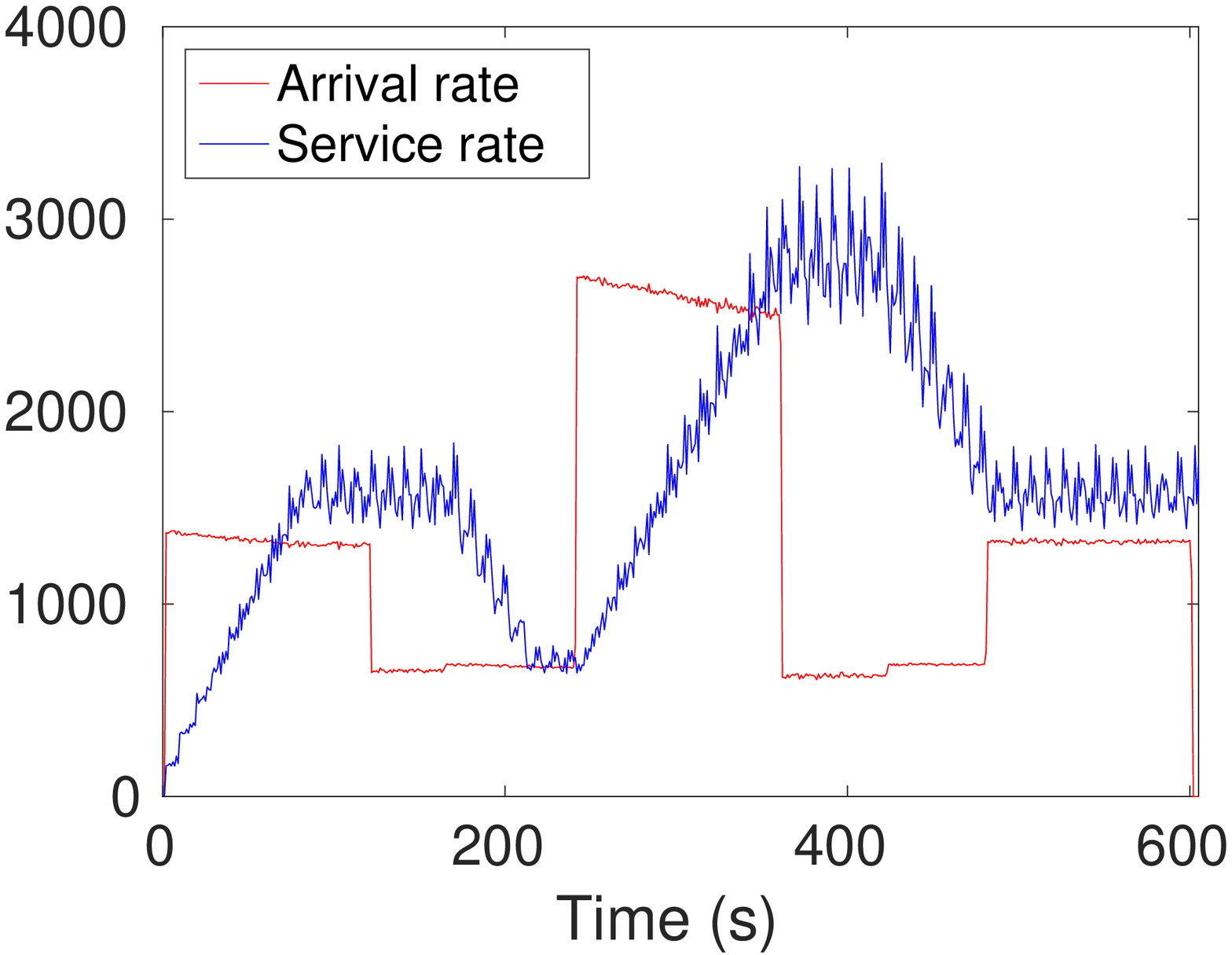}
	\label{fig:sync}
	\caption{Synchronised}
\end{subfigure}%
\begin{subfigure}[Unsynchronised]{0.5\columnwidth}
  \includegraphics[clip, width=0.98\columnwidth]{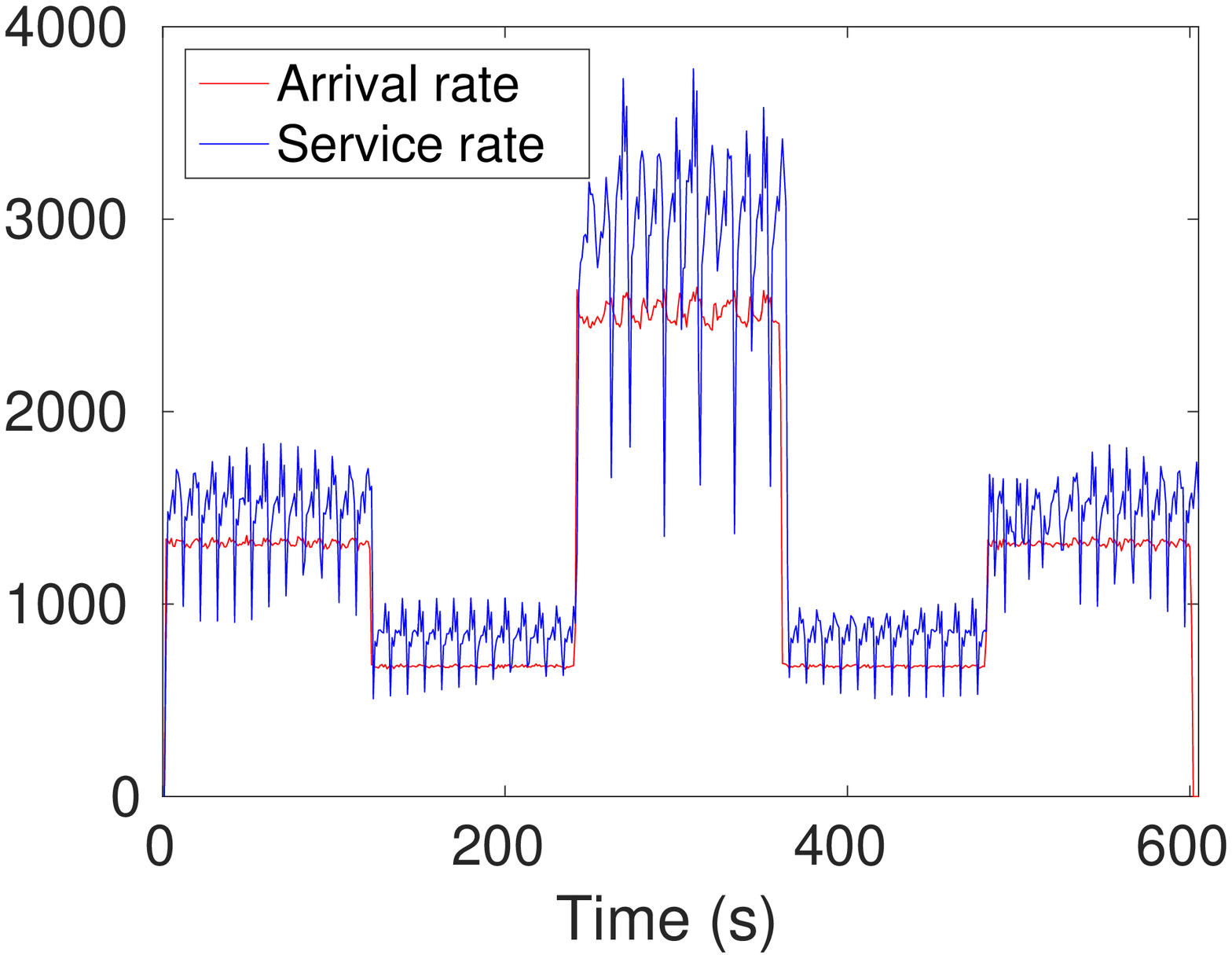}
	\label{fig:unsync}
	\caption{Unsynchronised}
\end{subfigure}
\caption[Rate adaptation in synchronised and unsynchronised scenarios]{{Comparison of tunnel rate adaptation in synchronised and unsynchronised scenarios.$\alpha$ in synchronised scenario. (duration: 10min, $\alpha=1$, $P = 1682, n=9615, c = 20, \gamma = 100$)}}
\label{fig:sync_unsync}
\end{figure}

With a synchronised scheduler traces only start and finish at the beginning of each cycle, a cycle being the duration of a trace. This means that when the arrival rate is changed in the middle of a cycle we have to wait until the start of the next cycle to update the number of active traces and so when new flows start they may experience significant delay.   This is particularly an issue when no traces are active in the tunnel and a new client flow starts since in this case the scheduler cannot transmit any packets until the next cycle starts (when some traces are already active the scheduler has the option to transmit some packets from the new flow e.g. by substituting for packets from other flows or for dummy packets).    Similarly, when the rate of incoming traffic is less than the rate of a single trace then sufficient packets need to be queued before a trace is activated, causing delay.

To mitigate such delays on lightly loaded links, two unsynchronised modifications are implemented to ``wake'' the channel from silence upon sensing new incoming traffic:

\begin{enumerate}
\item \emph{DNS Triggering}.  Upon receiving a DNS packet when no traces are active, a new trace is immediately activated.  
\item \emph{Use of Running Average When Lightly-Loaded}.  We maintain a running average of packet arrivals, $\bar{a} = (1-\zeta)  \bar{a} + \zeta  \sum_{u\in\U}\vi{a}{u}_k$, where $\zeta>0$ is a design parameter.  When $\bar{a}$ exceeds threshold $a^*$ and no traces are active then a new trace is activated.   In our experiments we use $\zeta= 0.001$, $a^*=0.005$.
\end{enumerate}

In addition, to avoid adding new traces in response to small spikes in arrivals, when $\gamma - \hat Q_k P < 0$ we only activate a new trace when also $\hat Q_{k}-\hat Q_{k-m} > 0$, where $m$ is a design parameter.  In this way we only respond to a sustained increase in $\hat Q$.  Similarly, when $\gamma - \hat Q_k P > 0$ we only decrease $x_k$ when also  $\hat Q_{k}-\hat Q_{k-m} < 0$.   In our experiments we use $m = 100$ unless otherwise stated, which is found to provide a good compromise between responsiveness and sensitivity to small fluctuations (since slots are 1ms duration, this correponds to a window of 100ms duration). 

To reduce sensitivity to small fluctuations in arrivals we also modify the $\hat Q$ update to use $\max\{\bar{a} , \sum_{u\in\U}\vi{a}{u}_k\}$, namely:
\begin{align*}
\hat Q_{k+1} = [\hat Q_k + \alpha(\max\{\bar{a} , \sum_{u\in\U}\vi{a}{u}_k\} - \sum_{t\in\T_k} p_{k-\tau_t})]^+
\end{align*}
In this way when the number of packet arrivals falls temporarily $\bar{a}$ is used instead of $a_k$.  

Figure \ref{fig:sync_unsync} illustrates the impact of these changes, comparing time histories of the arrival and services rates when using a synchronised scheduler vs an unsynchronised scheduler.   It can be seen that the unsynchronised scheduler is much more responsive to changes to traffic load.  Not only is the delay in increasing the service rate in response to increases in arrival rate reduced (so reducing the delay experienced by user packets) but also the delay in reducing the service rate in response to a fall in the arrival rate is also reduced (so reducing the number of dummy packets transmitted).  {Since it has better performance, from now on we will make use of this unsynchronised scheduler unless otherwise stated.}

\subsection{Choice of Trace Parameters}

\begin{figure}
\centering
\begin{subfigure}[Uplink]{0.48\columnwidth}
\centering
  \includegraphics[clip, width=0.98\columnwidth]{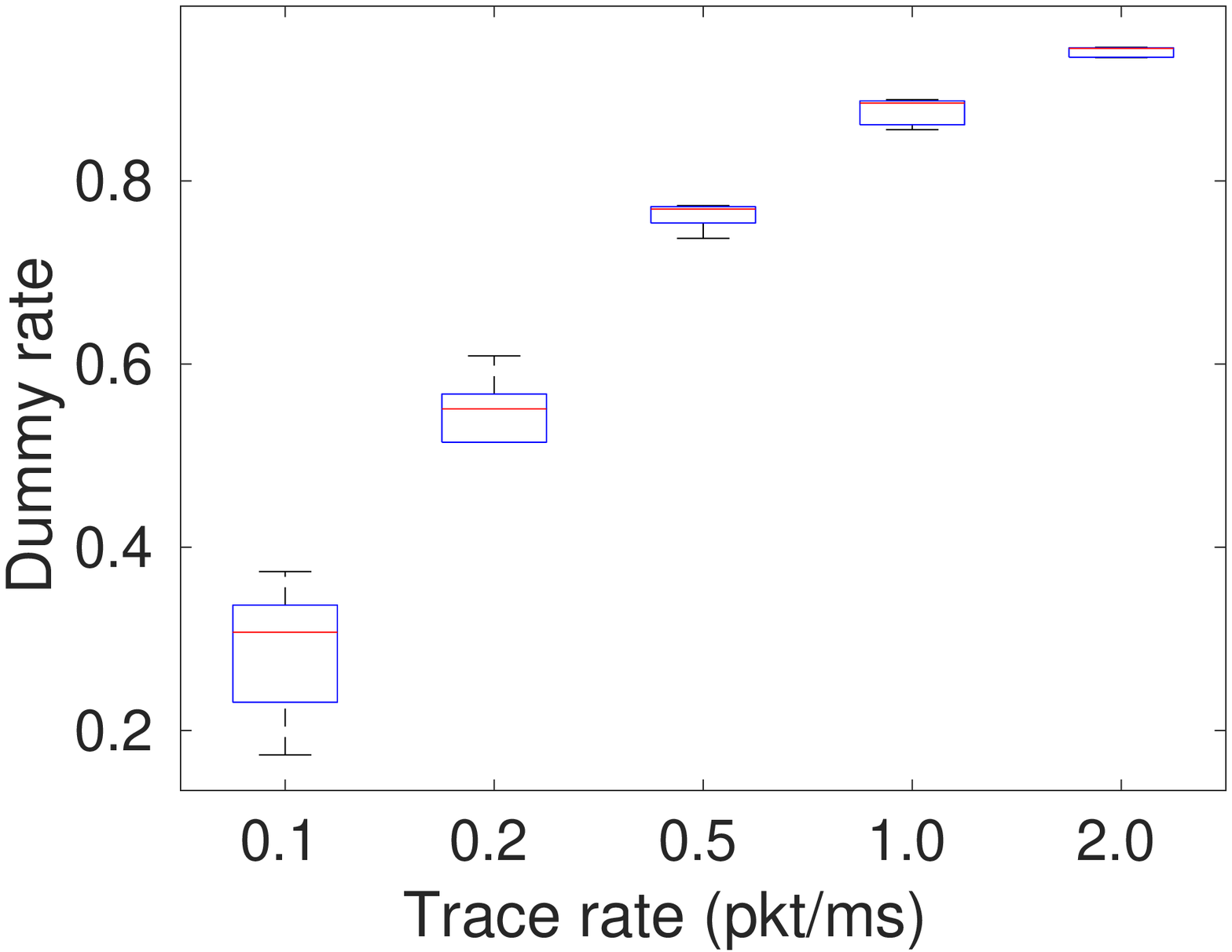}
  	\caption{Uplink}
\end{subfigure}%
\begin{subfigure}[Downlink]{0.48\columnwidth}
\centering
  \includegraphics[clip, width=0.98\columnwidth]{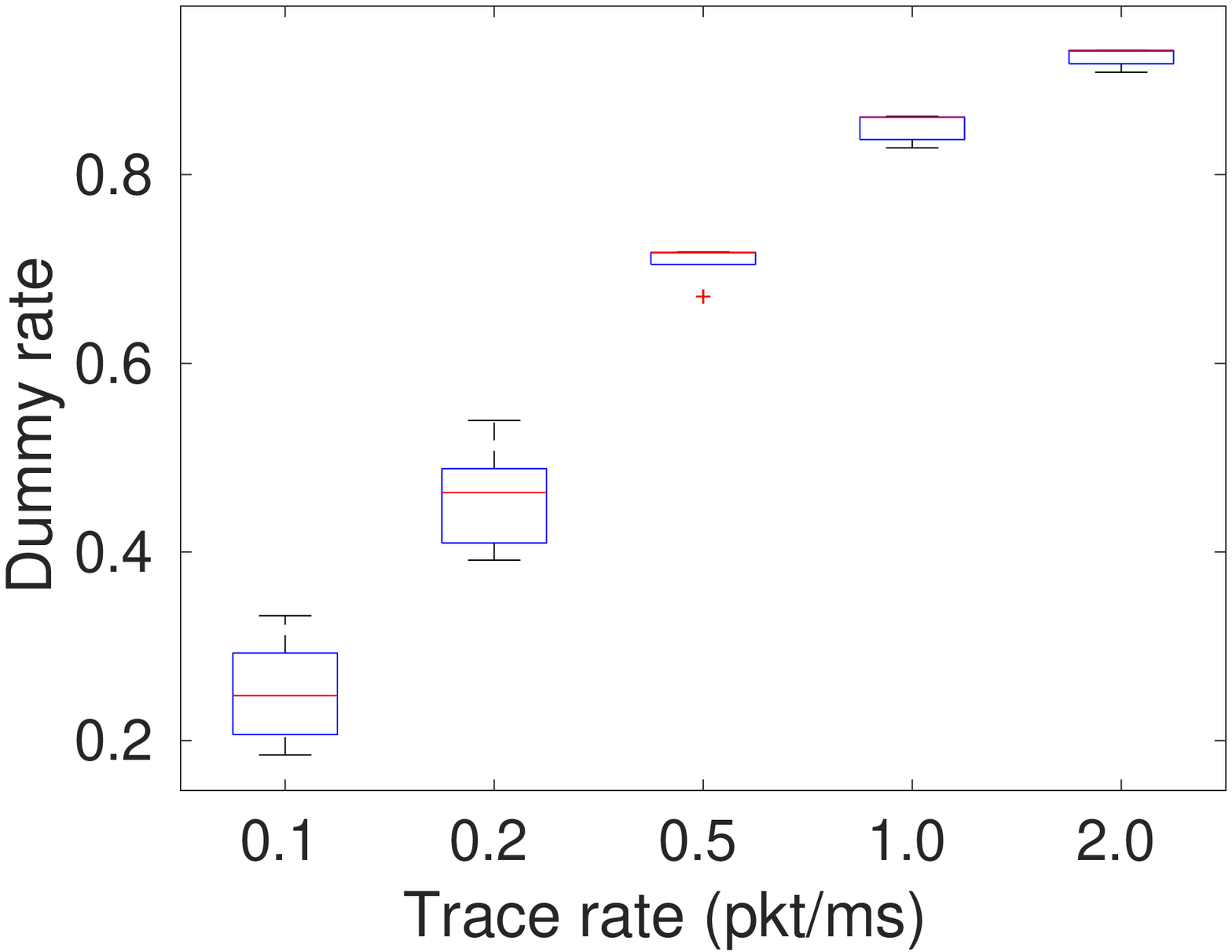}
  	\caption{Downlink}
\end{subfigure}
\caption[Effect of trace rate on dummy overhead]{Dummy overhead vs trace rate for a sample website.  For each rate data is shown for traces of duration 1, 2, 4, 8 and 16 seconds.  ($\gamma = 4096$).}
\label{fig:website_dummy_vs_rate}
\end{figure}

\begin{figure}
\centering
\begin{subfigure}[rate]{0.48\columnwidth}
\centering
  \includegraphics[clip, width=0.98\columnwidth]{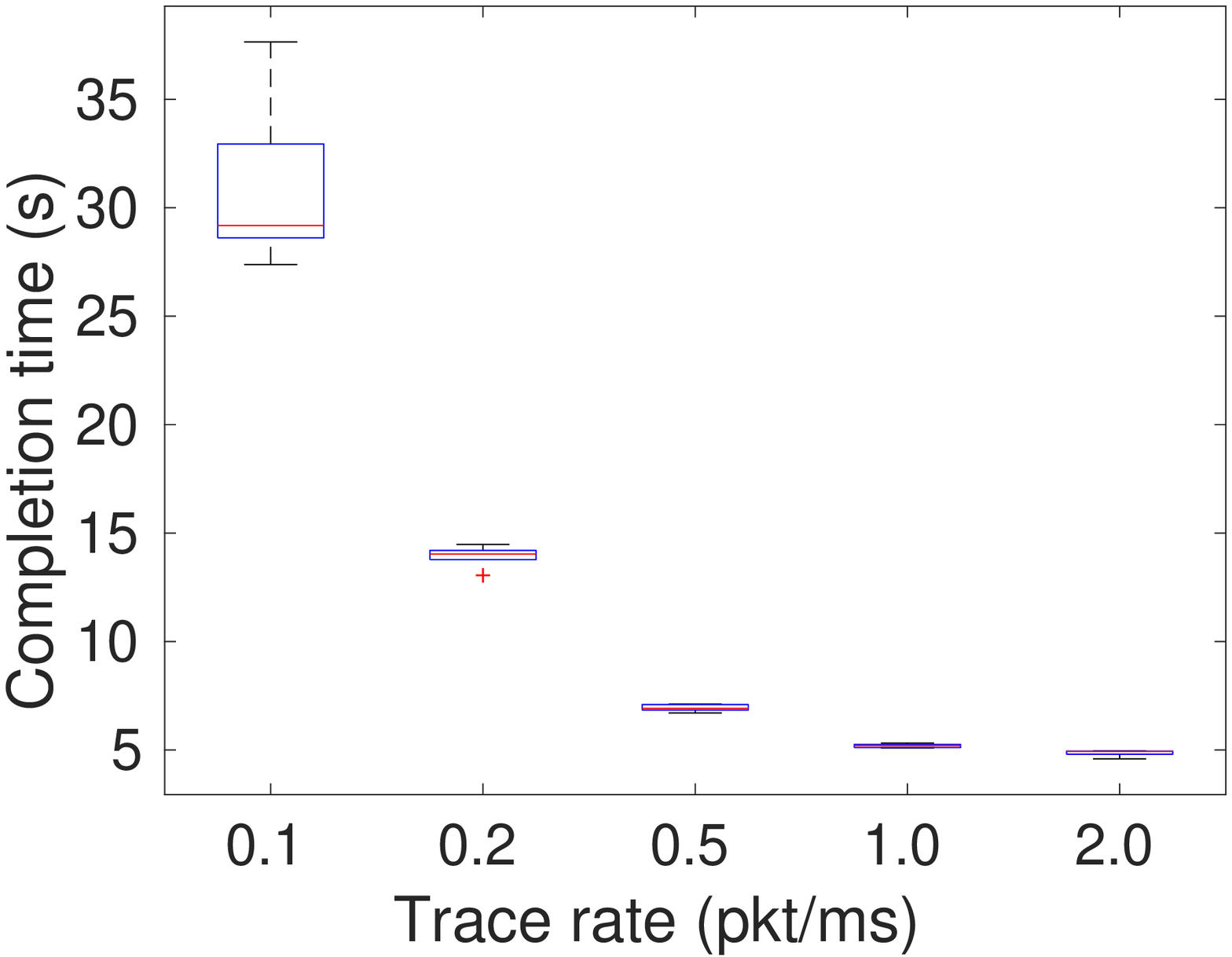}
  	\caption{Different rates}
\end{subfigure}%
\begin{subfigure}[duration]{0.48\columnwidth}
\centering
  \includegraphics[clip, width=0.98\columnwidth]{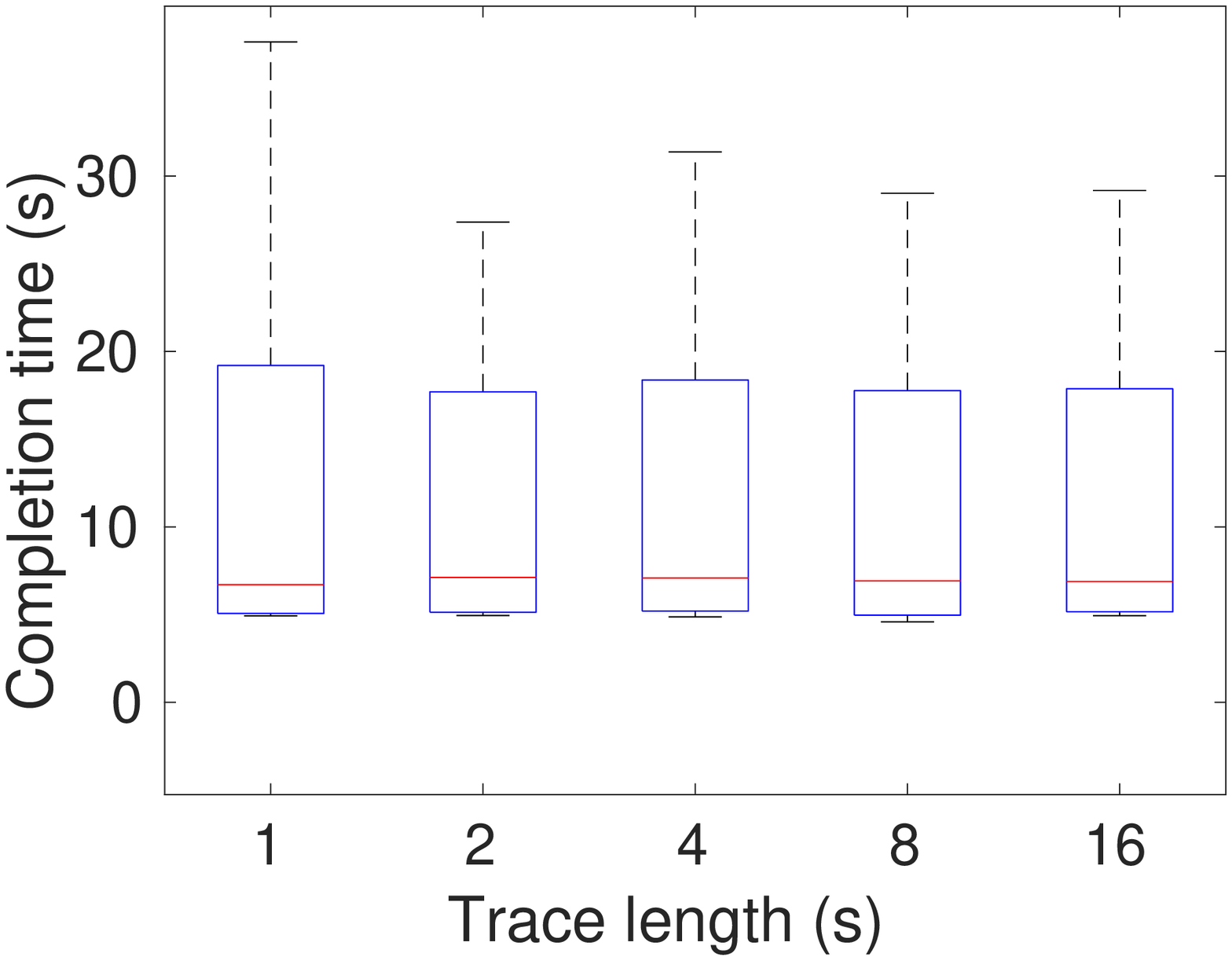}
  	\caption{Different durations}
\end{subfigure}
\caption[Effect of trace rate and duration on completion time]{Illustrating the effect of trace rate and duration on the completion time of a sample web page.  For each rate the data shown in (a) is for traces of duration 1, 2, 4, 8 and 16 second.  Similarly, for each duration the data shown in (b) is for rates of 0.1, 0.2, 0.5, 1.0 and 2.0 pkts/ms. ($\gamma = 4096$). }
\label{fig:website_completion}
\end{figure}

\begin{figure}
\centering
\begin{subfigure}[Uplink]{0.48\columnwidth}
\centering
  \includegraphics[clip, width=0.98\columnwidth]{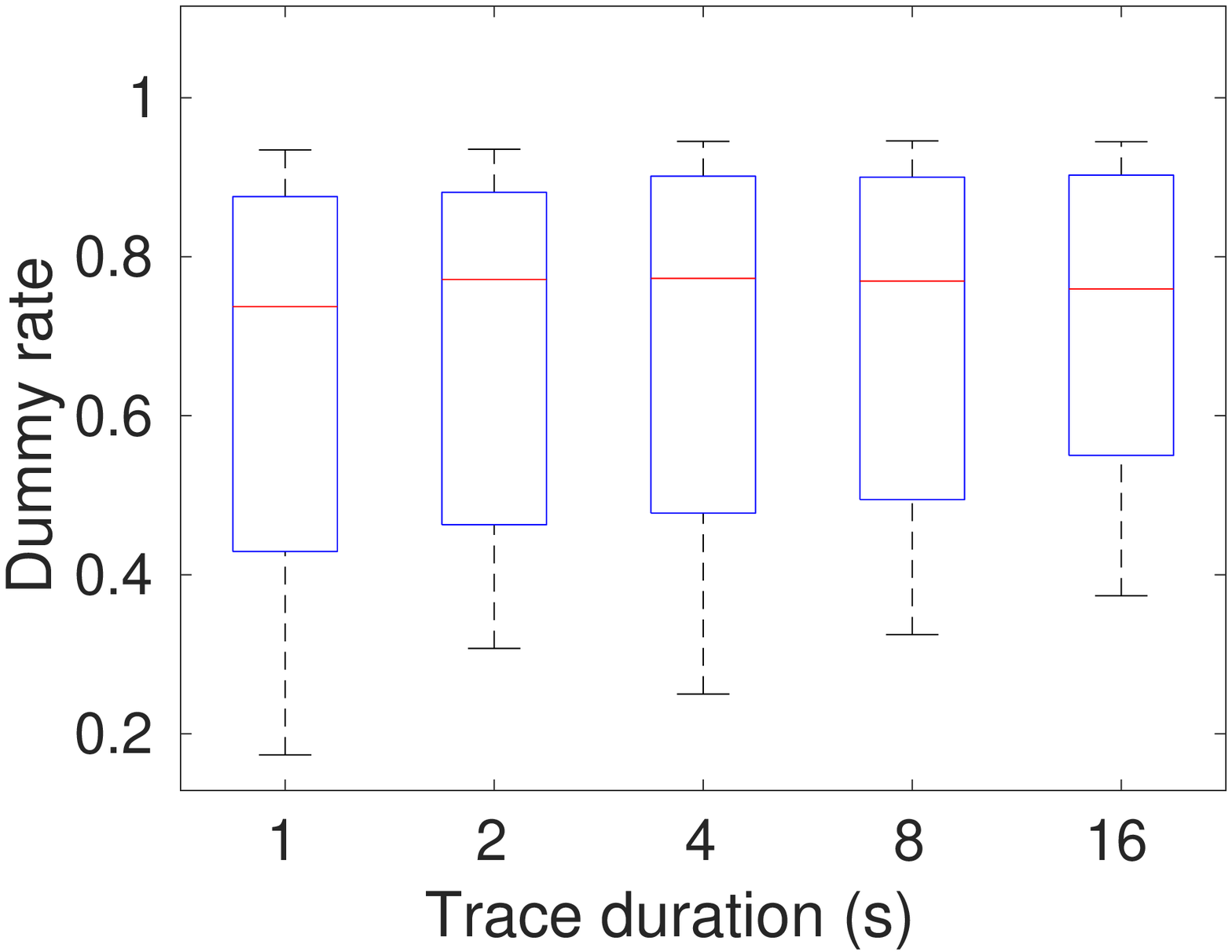}
  	\caption{Uplink}
\end{subfigure}%
\begin{subfigure}[Downlink]{0.48\columnwidth}
\centering
  \includegraphics[clip, width=0.98\columnwidth]{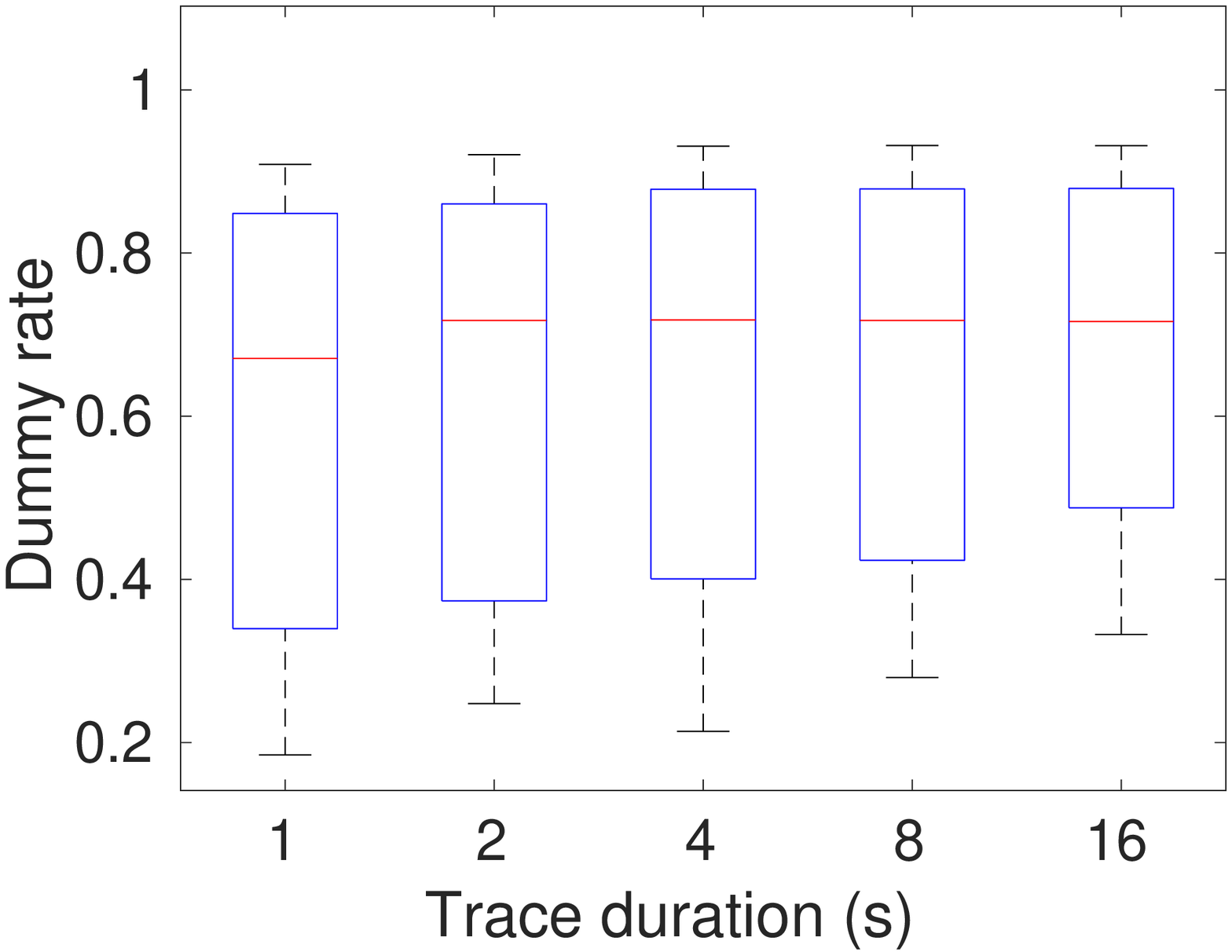}
  	\caption{Downlink}
\end{subfigure}
\caption[Effect of trace duration on dummy overhead]{Illustrating the effect of trace duration on dummy overhead for a sample website.  For each duration the data shown for rates of 0.1, 0.2, 0.5, 1.0 and 2.0 pkts/ms.  ($\gamma = 4096$).  }
\label{fig:website_dummy_vs_duration}
\end{figure}

We make use of traces which transmit at a constant rate for a defined duration.    Figure \ref{fig:website_dummy_vs_rate} shows measurements of the delay and fraction of dummy packets vs the trace rate and duration.  This data is for an example web fetch using TCP and shows measurements for both the downlink and uplink (the uplink carries the TCP ACKs).   Figure \ref{fig:website_completion} shows the corresponding impact of the trace rate and duration on the web fetch completion time.   It can be seen from these plots that increasing the trace rate reduces completion time but increases the dummy packet overhead, but the duration of trace has limited effect on completion time. The dummy packet overhead for different trace durations is also shown in Figure \ref{fig:website_dummy_vs_duration}.    

{The duration and rate of a trace also has an effect on the ability to defend against attacks. As explained in Section \ref{sec:traces}, pages that are served with the same pattern of traces are indistinguishable to an attacker.  Web pages that can be fetched using only one or two traces are likely to share the same pattern since there are few ways in which a small number of traces can be generated.  However when the number of traces needed to fetch a web page increases then the chance of having distinct patterns is increased.  Hence, better privacy is achieved by selecting the trace parameters such that the tunnel is able to serve most websites using only a small number of traces.   Based on these considerations and the data in Figure \ref{fig:website_dummy_vs_rate}-\ref{fig:website_dummy_vs_duration} we selected a trace of duration 9s and rate 0.2 pkts/ms ($n=9615$ slots and $P=1682$ packets, slots are of 1ms duration) and use these parameter values for the rest of our experiments.}

\subsection{Performance with CBR UDP Traffic}

\begin{figure}
\centering
  \includegraphics[clip, width=0.48\columnwidth]{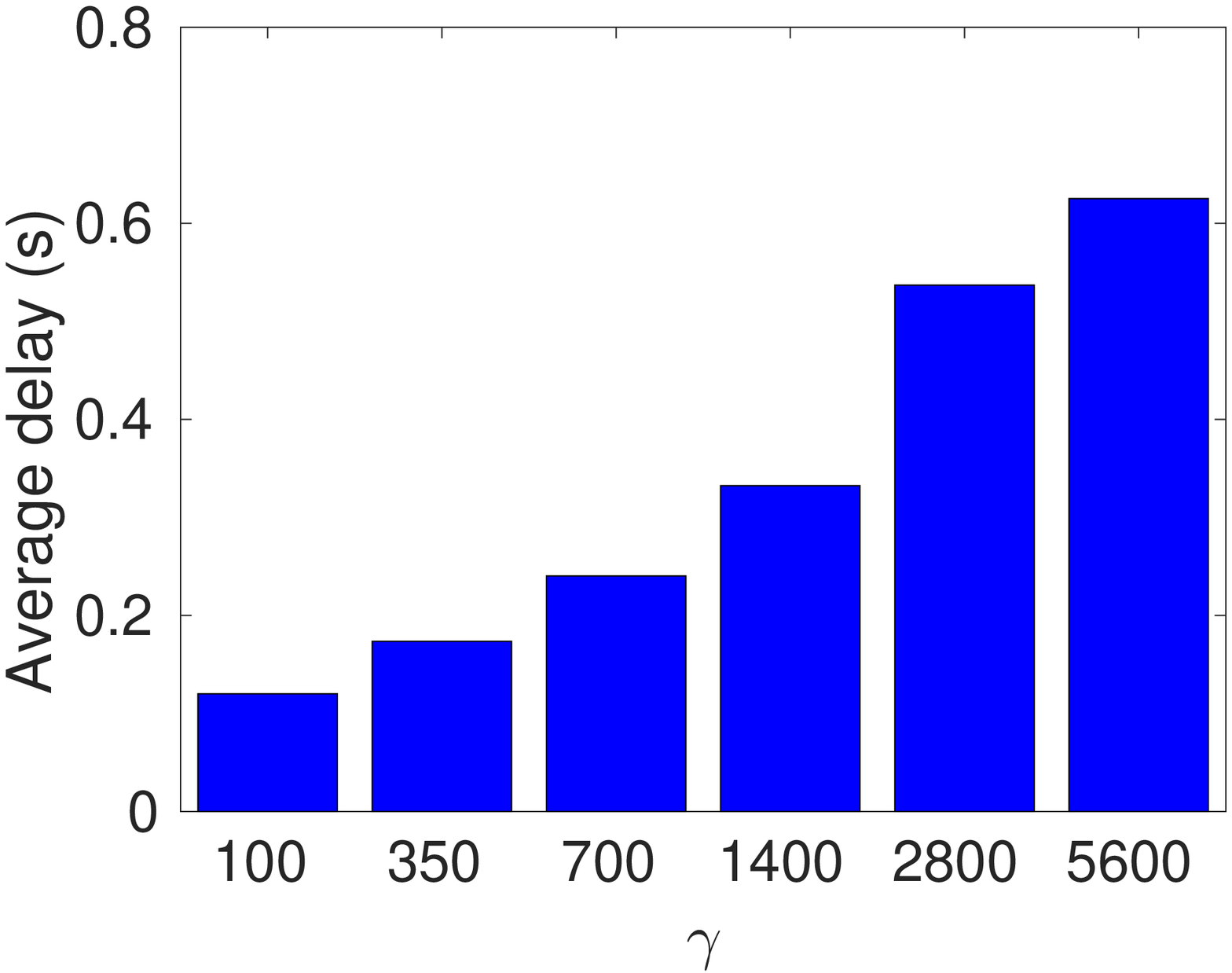}\hspace{0.03\columnwidth}%
  \includegraphics[clip, width=0.48\columnwidth]{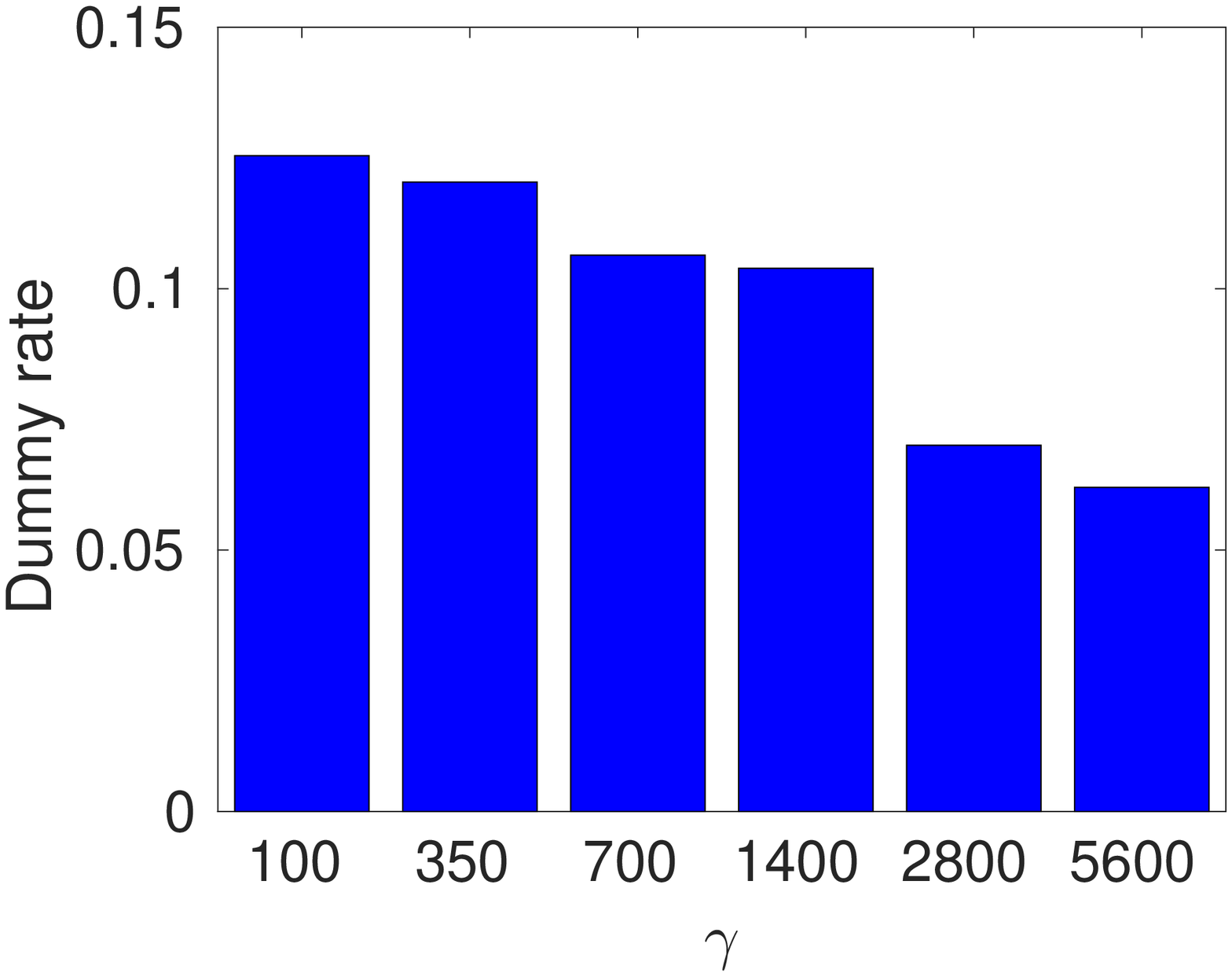}
\caption[Delay and dummy rates vs $\gamma$ for CBR UDP]{Comparison between dummy and delay rates for CBR UDP traffic and different choices of $\gamma$. Given the negligible variance, average over 5 tests are presented for each $\gamma$. (duration: 5min, $\bar{a} = 2800$, $\alpha = 1$, $P = 1682, n=9615, c = 20$). Values for each $\gamma$ are averaged over 5 tests (error bars are not shown since they are too small to be clearly visible).}
\label{fig:udp_diff_gamma}
\end{figure}


\begin{figure}
\centering 
  \includegraphics[clip, width=0.48\columnwidth]{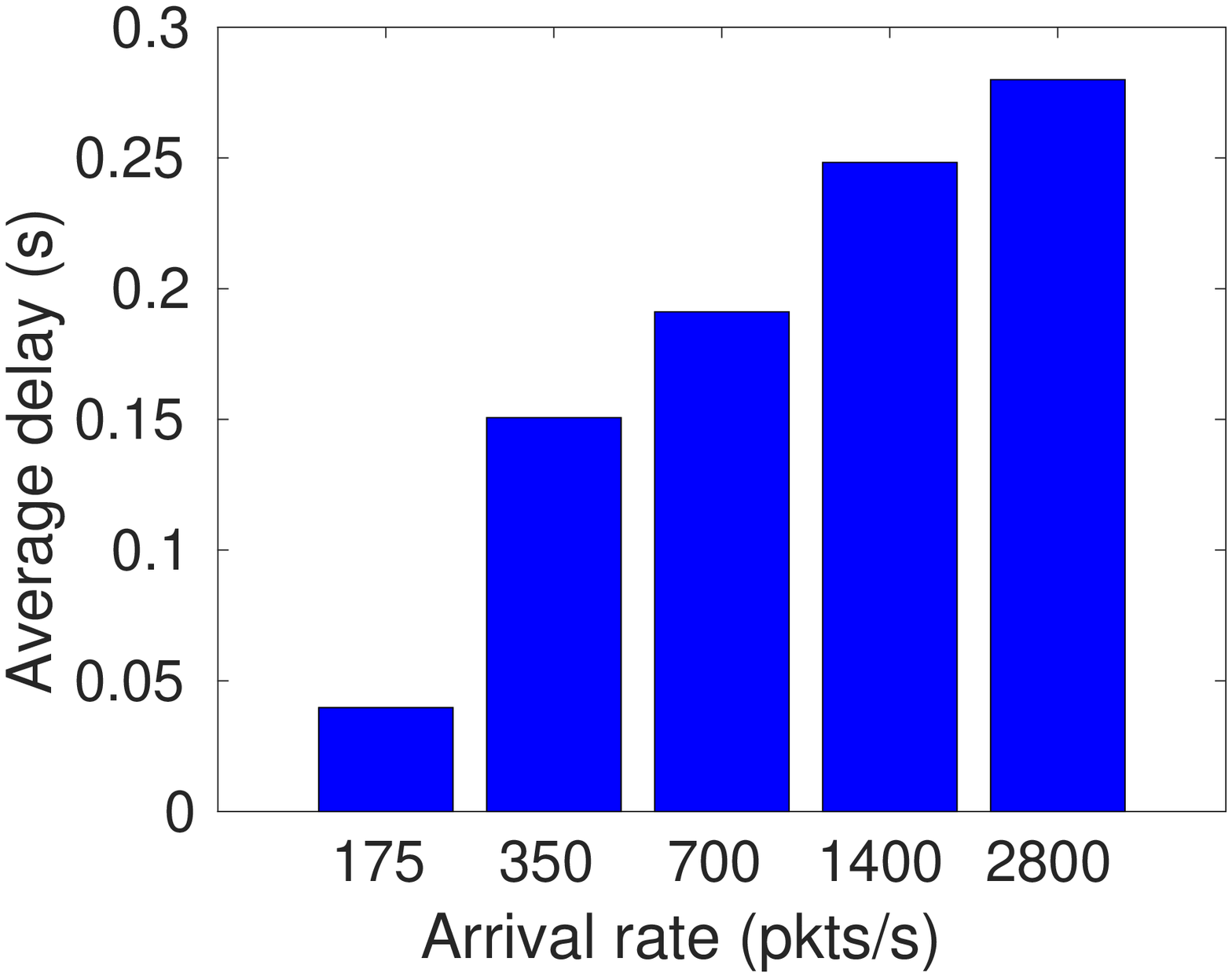}\hspace{0.03\columnwidth}%
  \includegraphics[clip, width=0.48\columnwidth]{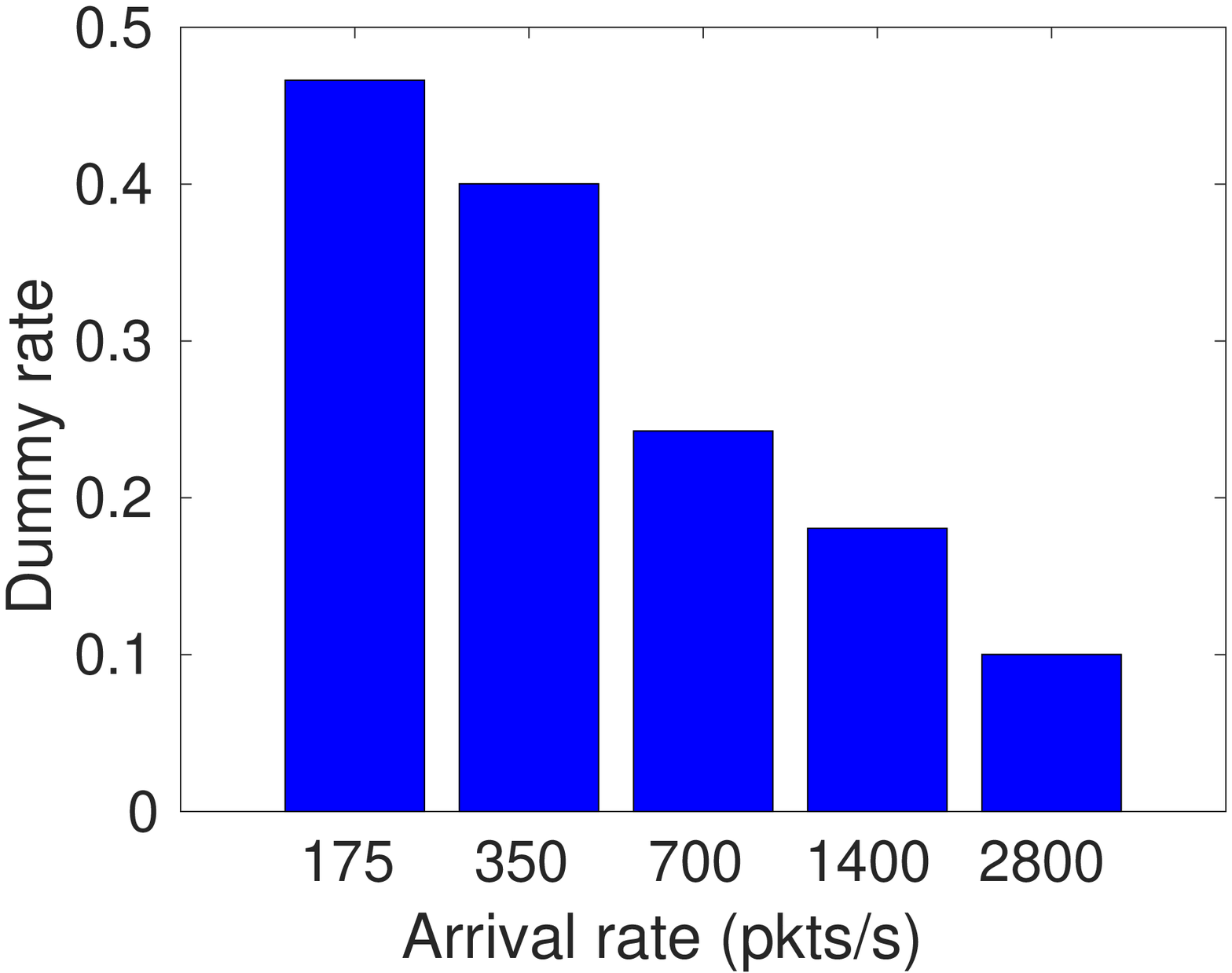}
\caption[Delay and dummy rates vs arrival rate for CBR UDP]{Comparison between delay and dummy rates  vs arrival rate for CBR UDP traffic with $\gamma = 1024$. (duration: 5min, $\alpha = 1$, $P = 1682, n=9615, c = 20$). Values are averaged over 5 tests.}
\label{fig:udp_diff_rates}
\end{figure}

In this section we study throughput efficiency and delay performance with constant bitrate (CBR) UDP arrivals.   By throughput efficiency we mean the mismatch between the transmit rate of the scheduler and the arrival rate of client packets.  Recall that when there are insufficient client packets buffered at the scheduler then the scheduler transmits dummy packets so that its transmissions can continue to follow the predefined trace pattern.  These dummy packets provide resistance to traffic analysis attacks but increase the load on the network and so we would like to minimise these.   There is a fairly direct trade-off between delay and the volume of dummy packets transmitted since by increasing the backlog of client packets buffered  at the input to the scheduler we reduce the need for dummy packets but increase the delay experience by client packets.  Conversely, reducing the number of client packets buffered tends to increase the need for dummy packets but decrease delay.

This trade-off between throughput efficiency and delay can be seen in Figure \ref{fig:udp_diff_gamma}, which plots measurements of the delay and dummy rate vs scheduler parameter $\gamma$ which directly influences the queue backlog (with increases in $\gamma$ increasing the backlog).  
The dummy rate value shown is the ratio of the dummy packets sent to the total number of packets sent (dummy plus processed packets).  Note that there is no packet loss within the scheduler in these tests.  It can be seen from Figure \ref{fig:udp_diff_gamma} that as $\gamma$ is increased the delay rises but the rate at which dummy packets are sent falls.   Observe, however, that even for relatively small values of $\gamma$ the dummy rate remains reasonably small, which is encouraging.

Figure \ref{fig:udp_diff_rates} shows corresponding measurements of delay and dummy rate as the arrival rate is varied (scheduler parameter $\gamma$ is held constant at $1024$ in these plots).    It can be seen that the delay tends to increase with arrival rate and the dummy rate to fall.   This can be understood by noting that as the arrival rate rises more user packets tend to be buffered at the scheduler.  Hence, the delay rises but also user packets are more likely to be available when a transmisison opportunity occurs and so the number of dummy packets needed falls.
Importantly, observe that the delay is consistently reasonable, rising to no more than 275ms at higher arrival rates, despite the extensive traffic shaping being carried out.   Also, the decrease in dummy rate with increasing arrival rate means that the throughput efficiency increases with increasing load, so mitigating the dummy packet cost incurred by the traffic shaping.     



\subsection{Performance with TCP Flows}
\label{sec:realTCP}


We finish this section by presenting data on the scheduler performance with TCP traffic (for the default Linux Cubic TCP variant).   Figure \ref{fig:tcp_diff_gamma_rates} plots measurements of delay and fraction of dummy packets vs design parameter $\gamma$ when fetching a 1024MB file from a campus server using \texttt{wget}.  Data is shown for both the downlink and uplink, the packet rate on the uplink being roughly half of that on the downlink due to delayed acking by TCP.   It can be seen that the delay increases with $\gamma$, as expected, but remains less than 100ms even for relatively large values of $\gamma$.  The fraction of dummy packets is relatively insensitive to $\gamma$ when using TCP and remains consistently below 20\%.

\begin{figure}[!t]
\centering
\begin{subfigure}[Uplink]{1.0\columnwidth}
  \includegraphics[clip, width=0.46\columnwidth]{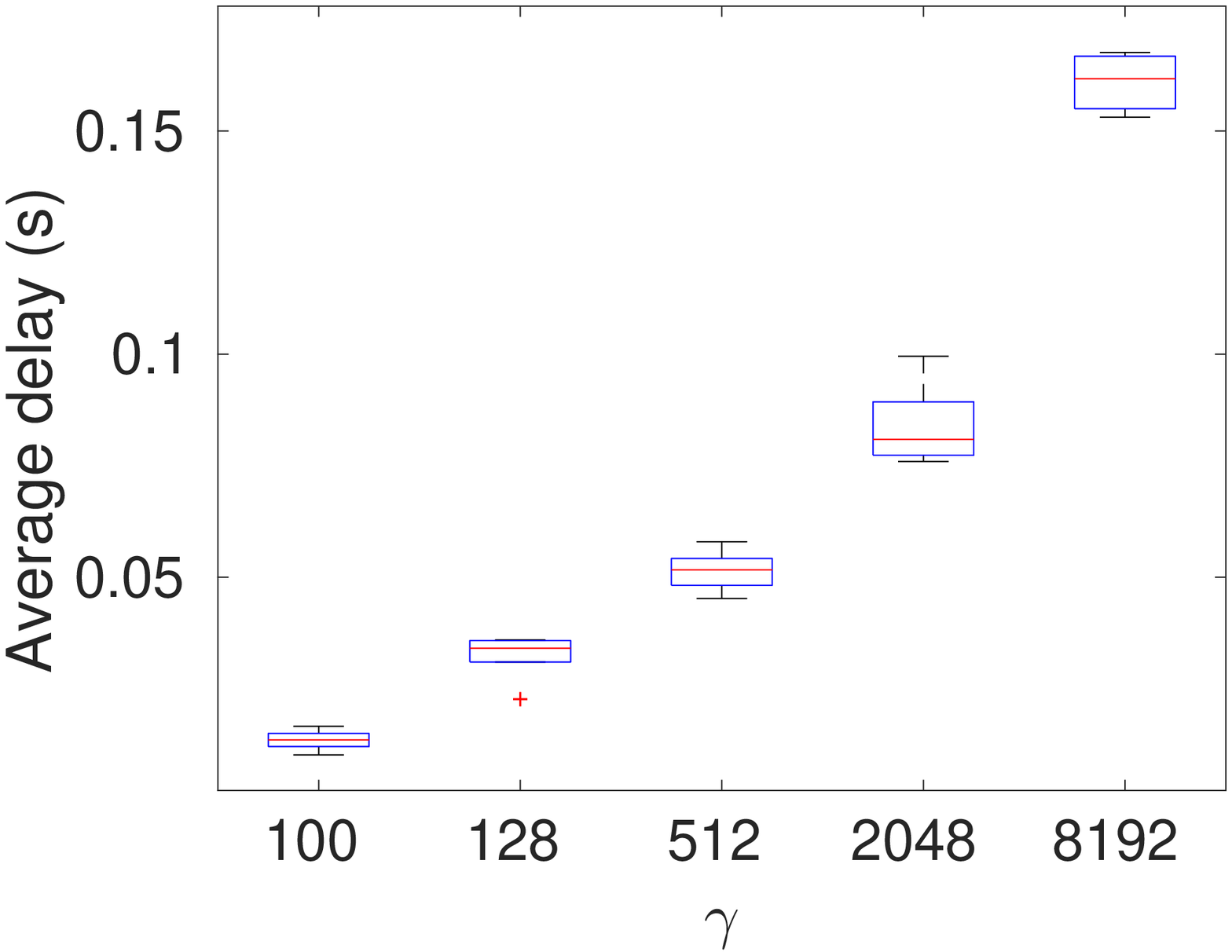}
  \includegraphics[clip, width=0.46\columnwidth]{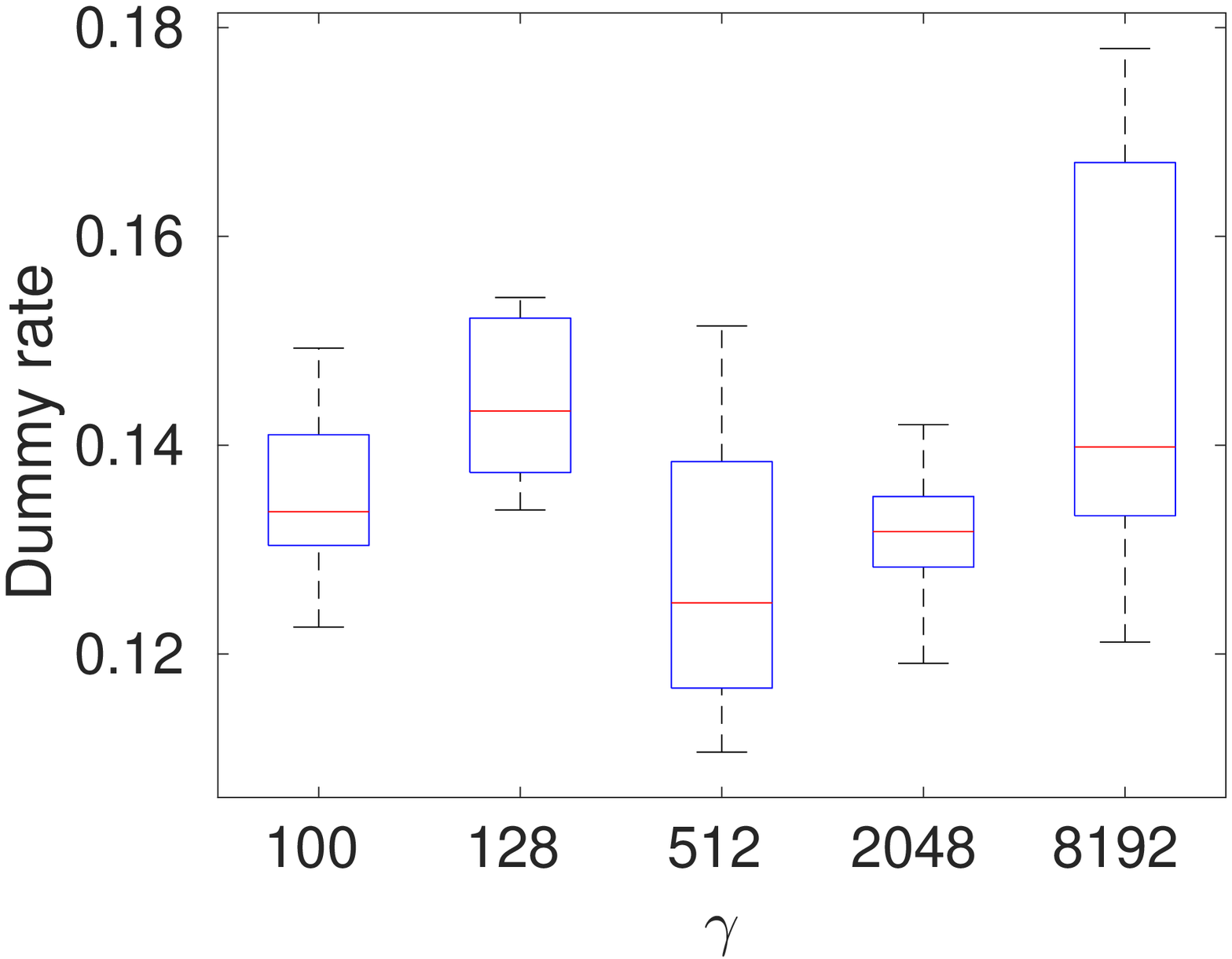}
	\label{fig:tcp_gamma_uplinkk}
	\caption{Uplink}
\end{subfigure}
\begin{subfigure}[Downlink]{1.0\columnwidth}
  \includegraphics[clip, width=0.46\columnwidth]{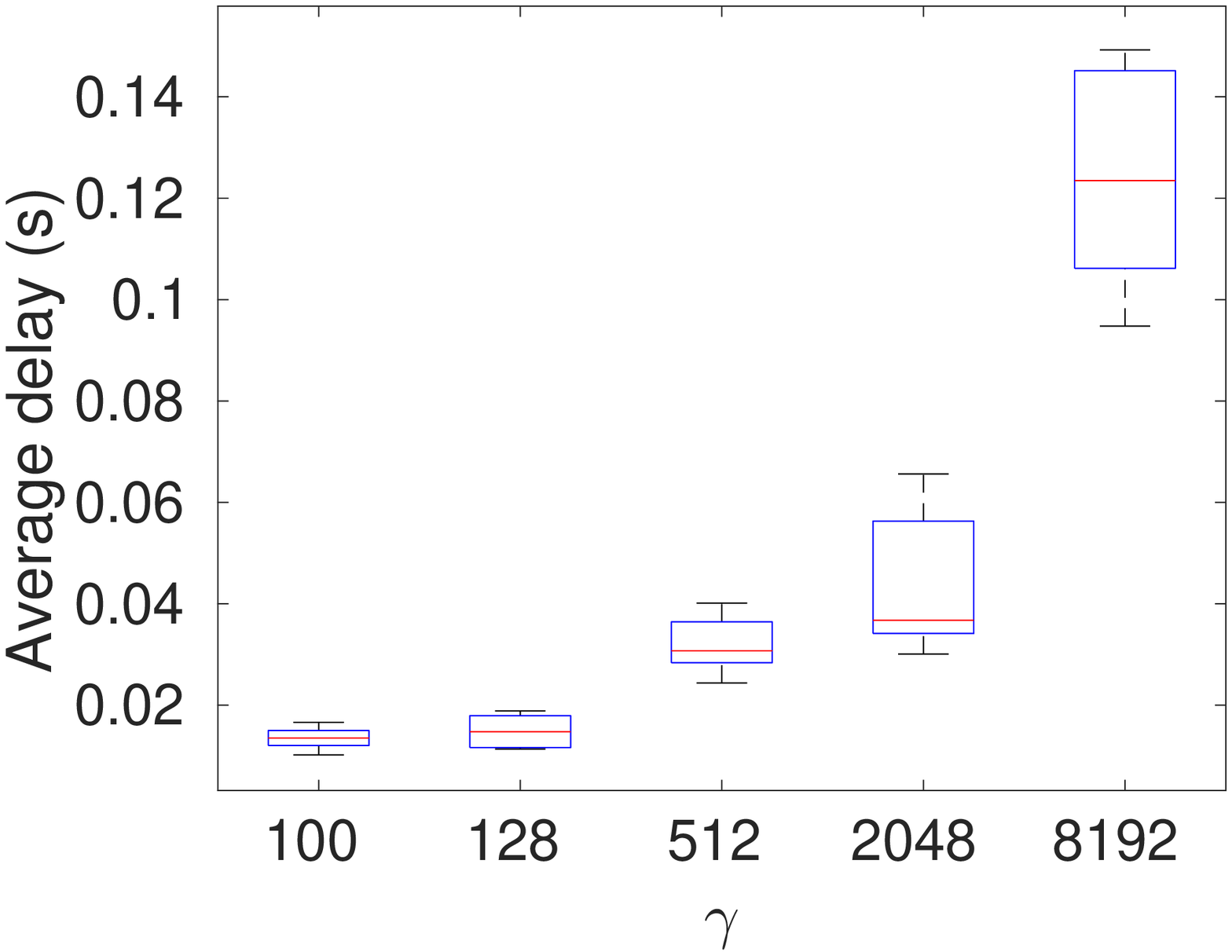}
  \includegraphics[clip, width=0.46\columnwidth]{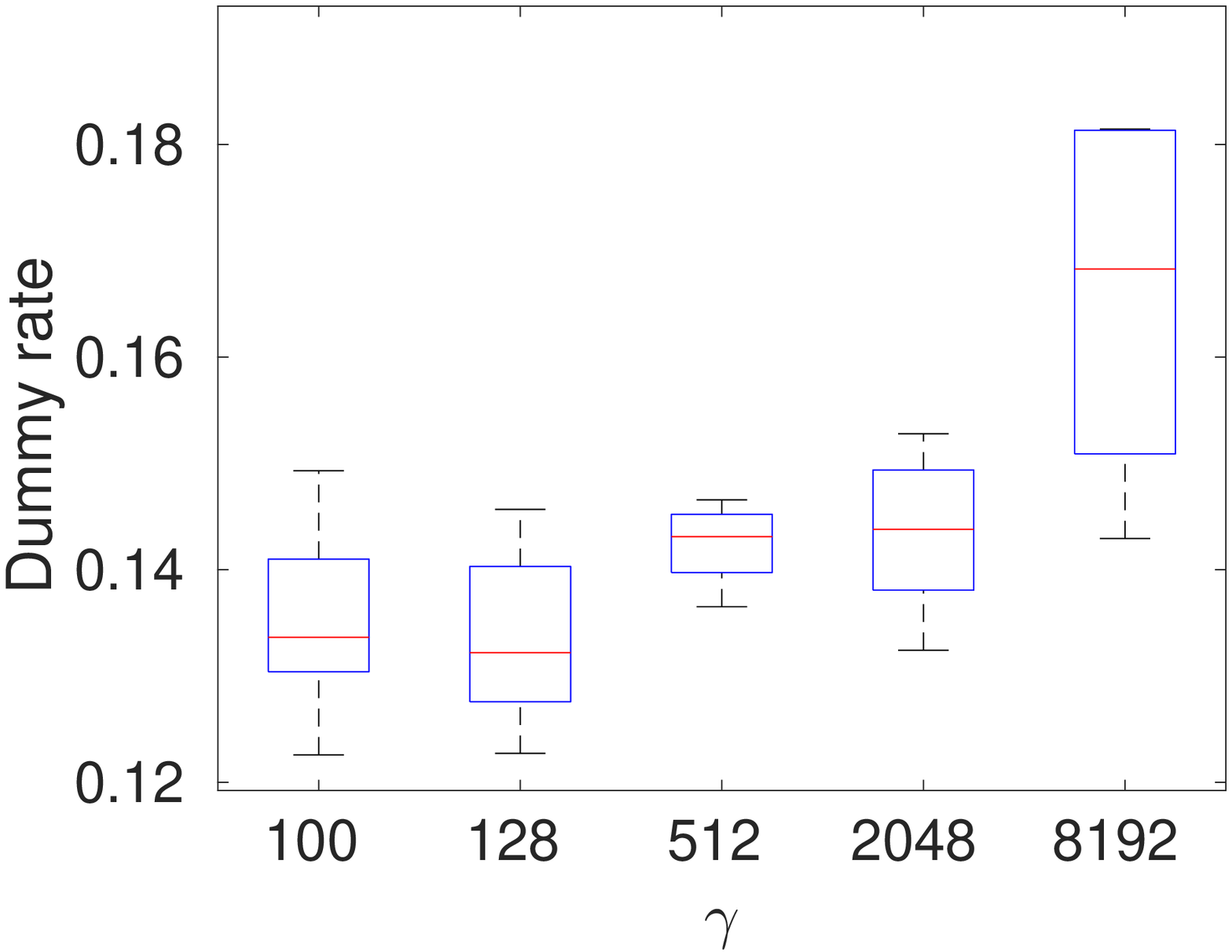}
	\label{fig:tcp_gamma_downlink}
	\caption{Downlink}
\end{subfigure}
\caption[Delay and dummy rates vs $\gamma$ for TCP]{Measurements of delay and fraction of dummy packets vs $\gamma$ for TCP traffic.  The data is measured fetching a 1024MB file, the average over 5 tests is presented for each $\gamma$ ($\alpha = 1$, $P = 1682, n=9615, c = 20$).}
\label{fig:tcp_diff_gamma_rates}
\end{figure}


\begin{figure}
\centering
  \includegraphics[clip, width=0.65\columnwidth]{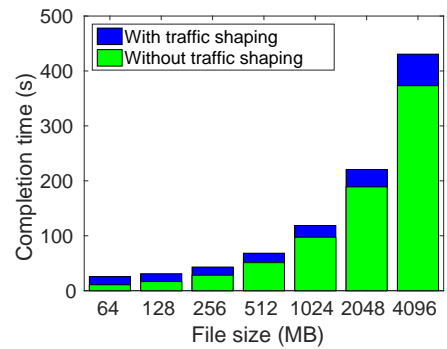}
\caption[Completion time vs file sizes when using TCP]{Completion time vs file sizes when using TCP, $\gamma = 1024$.  Results are shown both for the traffic shaped tunnel and for an unprotected link (no encryption, no traffic shaping). Values are averaged over 5 runs.}
\label{fig:tcp_diff_files_1024}
\end{figure}

Figure \ref{fig:tcp_diff_files_1024} plots the measured completion time vs file size fetched.    For comparison, the corresponding data is also shown for a link with no traffic shaping.   It can be seen that the cost, in terms of increase in completion time, is modest.   Figure \ref{fig:tcp_diff_files_rates} provides additional detail, showing measured data on delay and dummy rate vs file size.    It can be seen that the delay and dummy rate both tend to fall as the file size increases.   The effect here is due to the time that it takes the scheduler to adapt to the arrival of a new flow:  for longer flows this adaptation overhead gets washed out and amortised over many packets but for short flows its effect is more pronounced.  This can be seen in Figure \ref{fig:convergence}, which shows time histories of the number of active traces on both the downlink and uplink (uplink traces are carrying the TCP acks and so are fewer in number).  The experiment is conducted by fetching a file of size 8192MB. The link reaches steady state in about 20 seconds.


\begin{figure}[!t]
\centering
\begin{subfigure}[Uplink]{1.0\columnwidth}
  \includegraphics[clip, width=0.46\columnwidth]{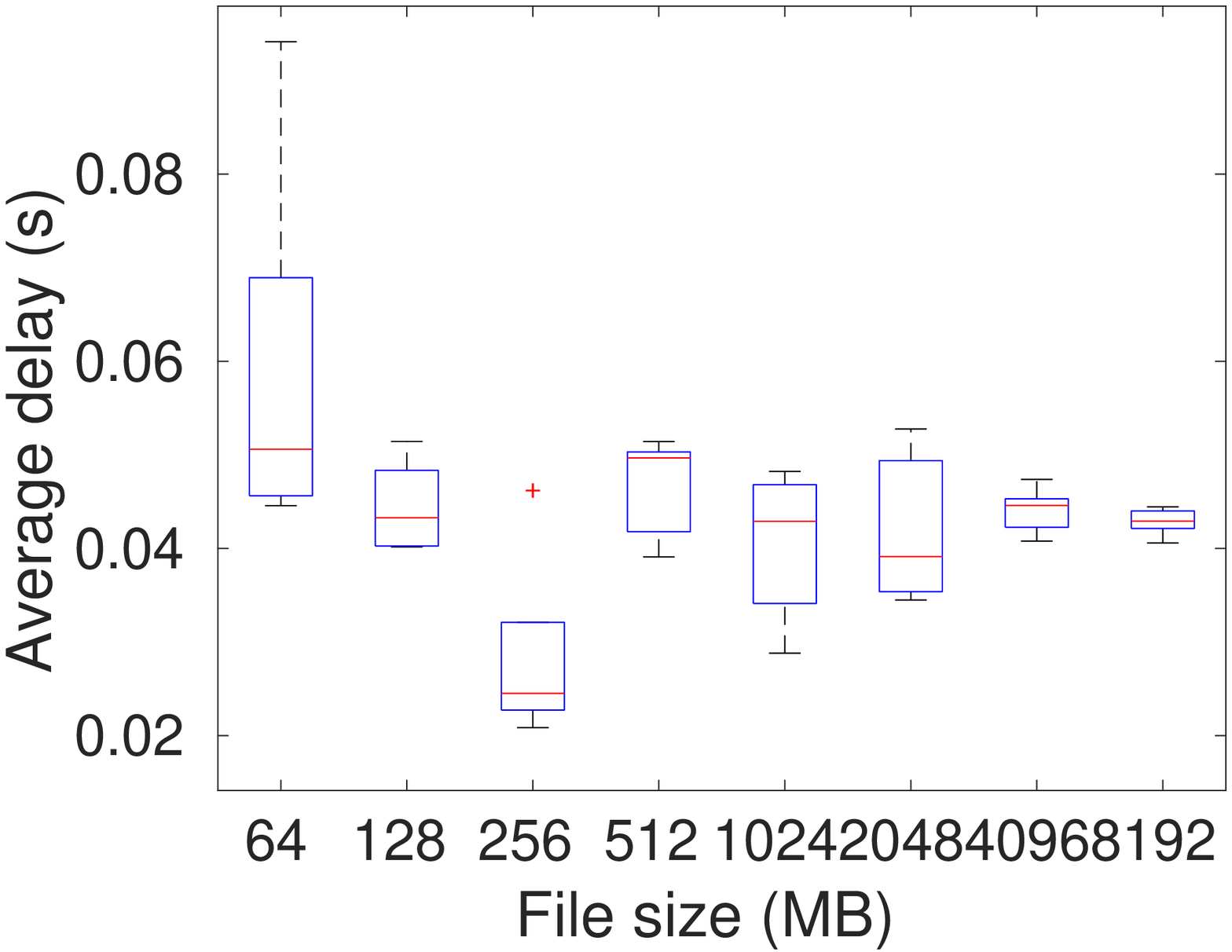}
  \includegraphics[clip, width=0.46\columnwidth]{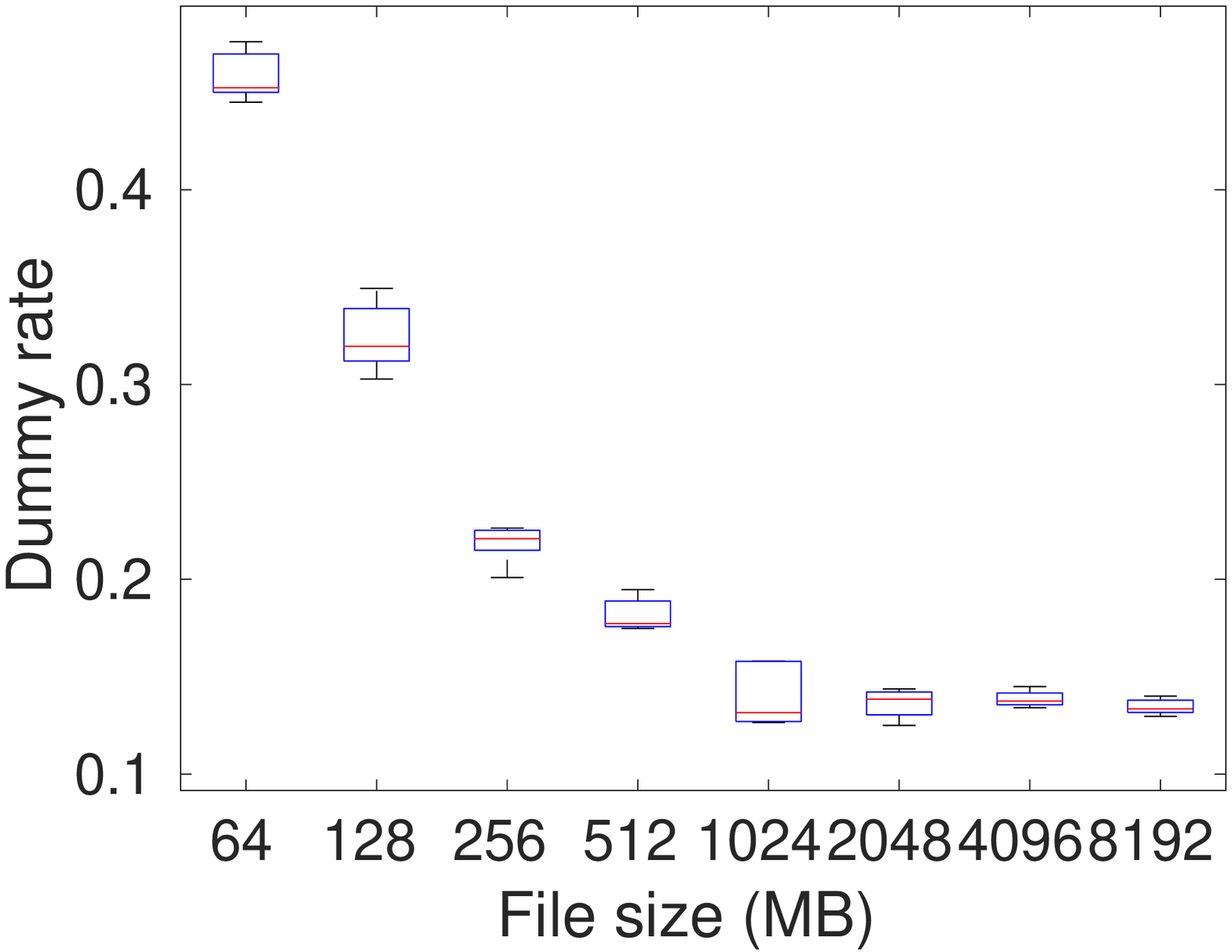}
	\label{fig:tcp_files_uplinkk}
	\caption{Uplink}
\end{subfigure}
\begin{subfigure}[Downlink]{1.0\columnwidth}
  \includegraphics[clip, width=0.46\columnwidth]{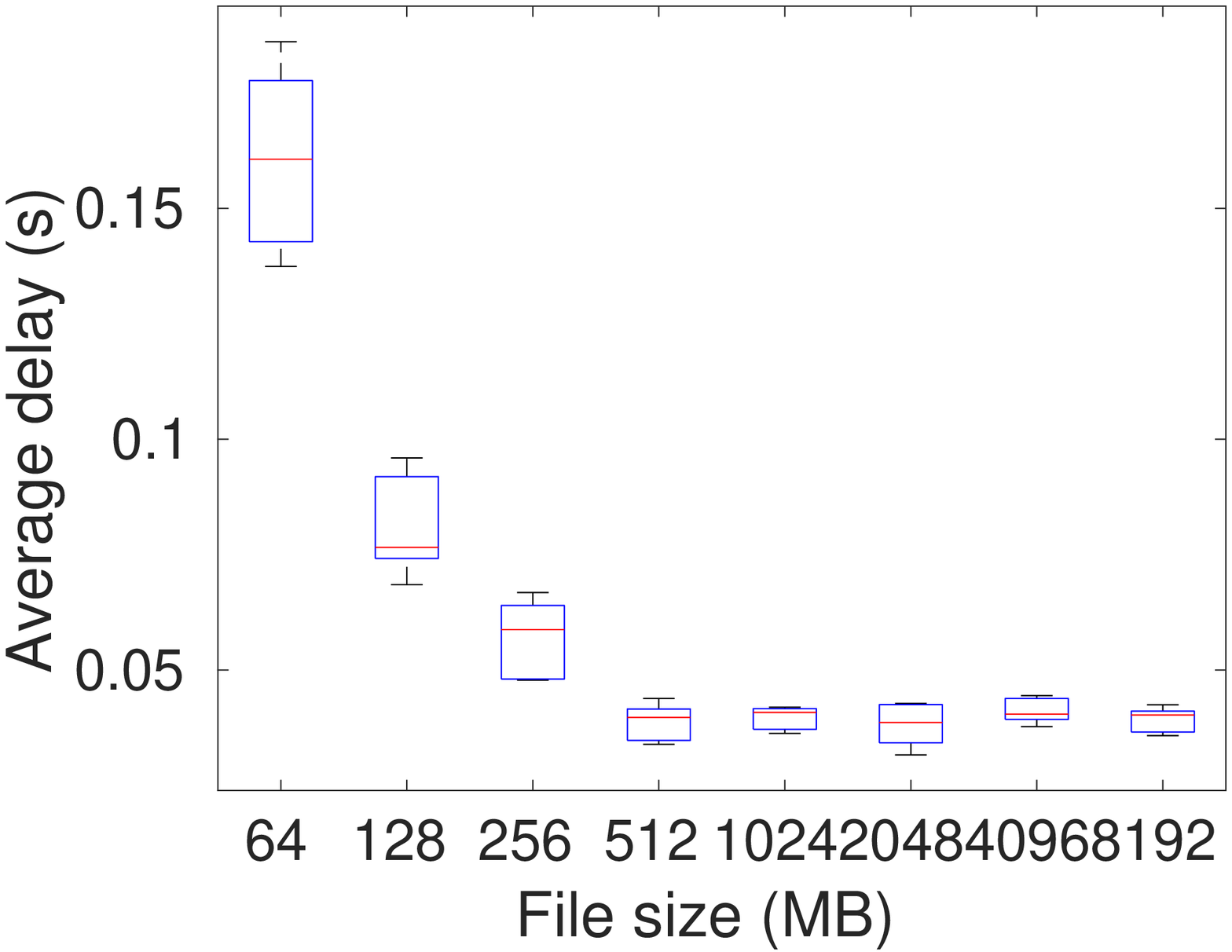}
  \includegraphics[clip, width=0.46\columnwidth]{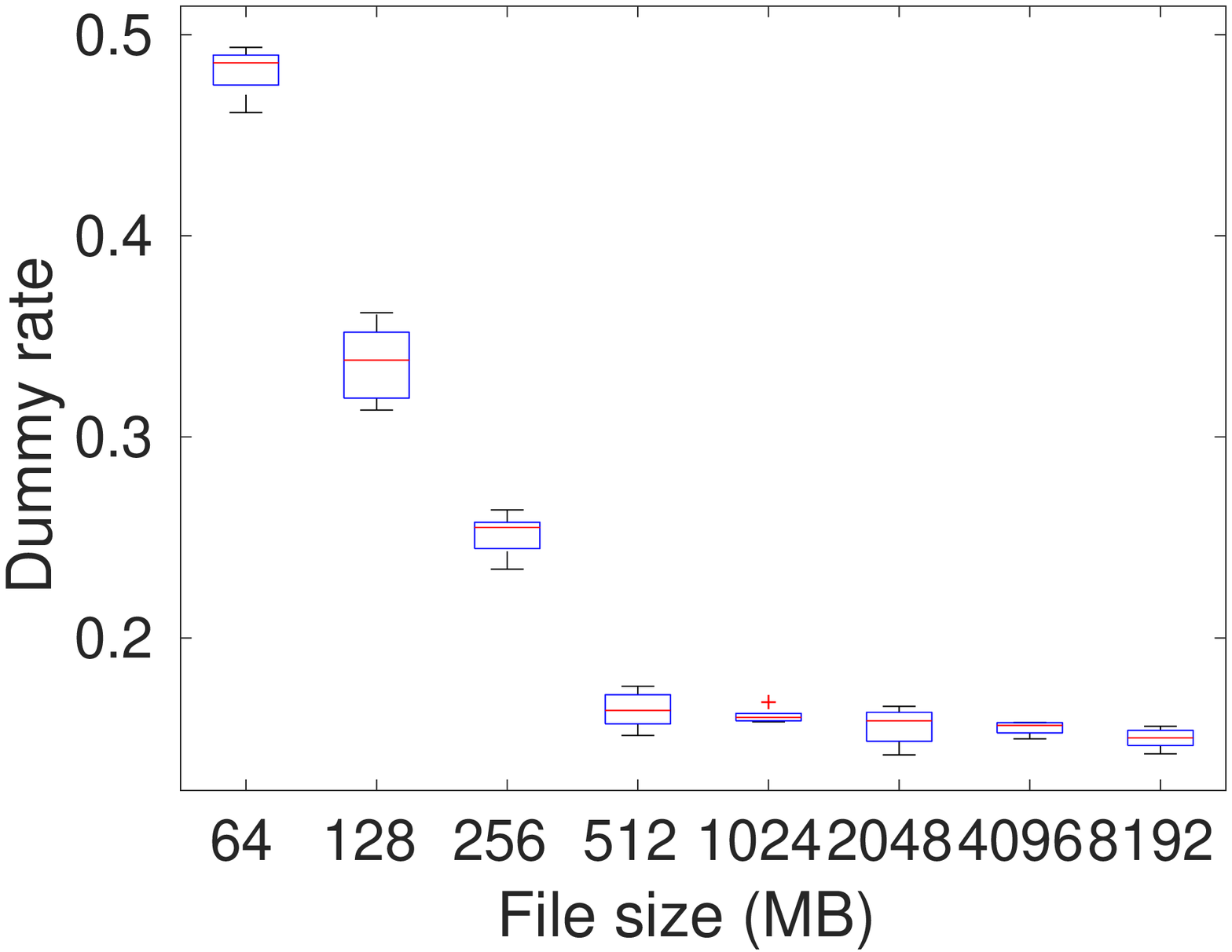}
	\label{fig:tcp_files_downlink}
	\caption{Downlink}
\end{subfigure}
\caption[Dummy and delay rates of fetching files of different sizes]{Comparison between dummy and delay rate when using TCP to fetch files of different sizes, for $\gamma = 1024$. Box plots for values of 5 runs are shown.}
\label{fig:tcp_diff_files_rates}
\end{figure}

\begin{figure}[!t]
\centering
\begin{subfigure}[Uplink]{0.46\columnwidth}
\centering
  \includegraphics[clip, width=0.95\columnwidth]{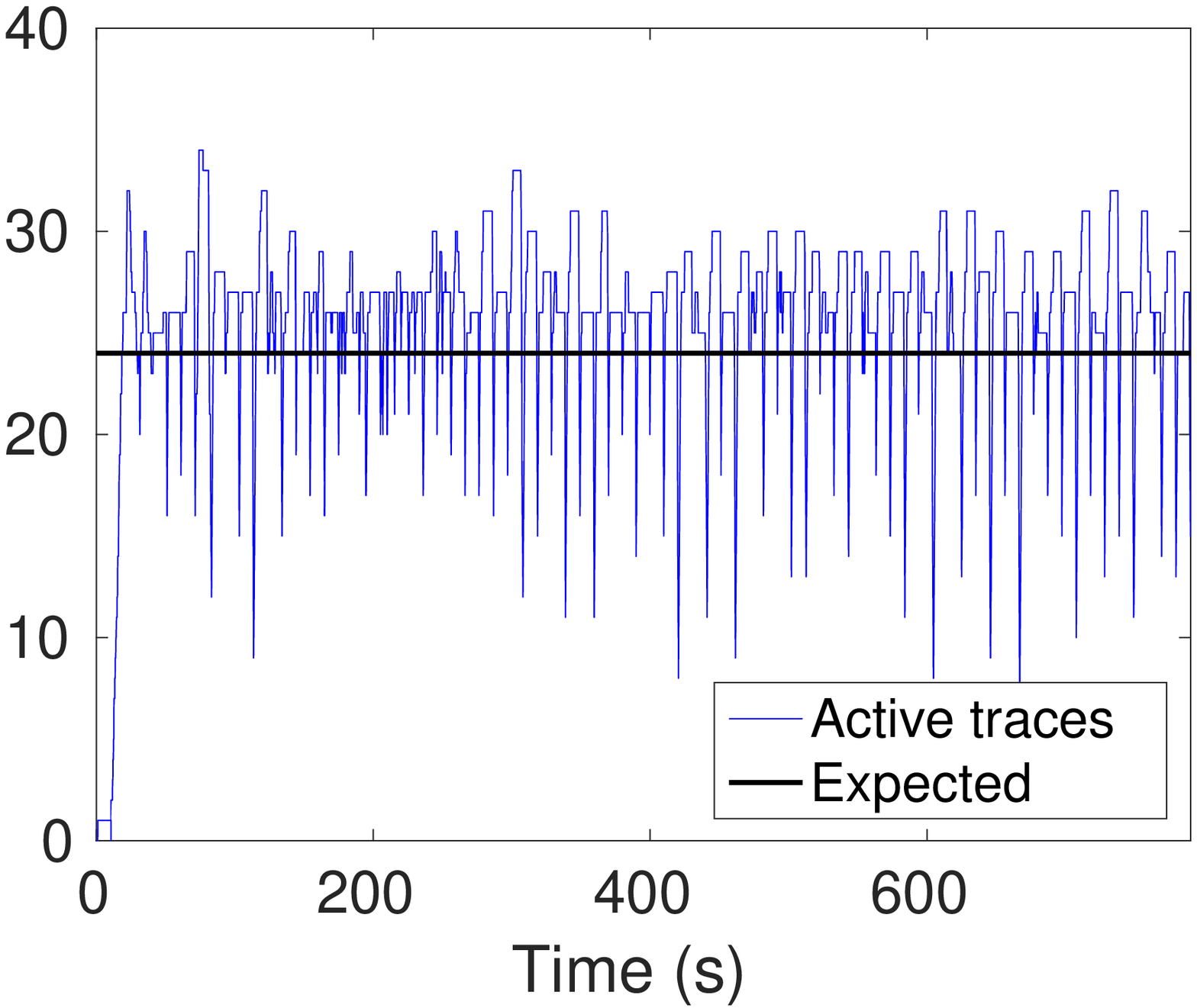}
  	\label{fig:convergence_uplink}
  	\caption{Uplink}
\end{subfigure}
\begin{subfigure}[Downlink]{0.46\columnwidth}
\centering
  \includegraphics[clip, width=0.95\columnwidth]{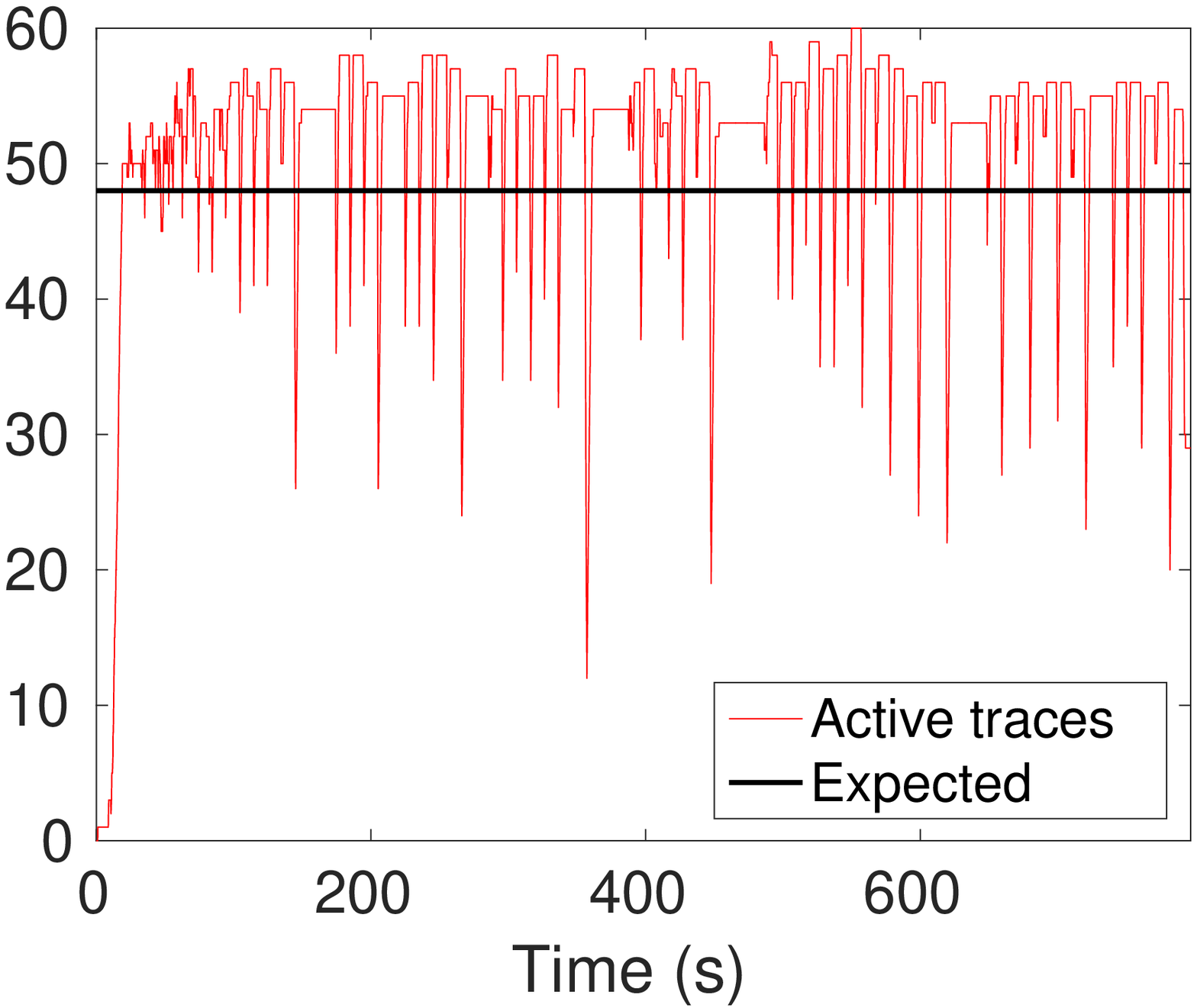}
	\label{fig:convergence_downlink}
	\caption{Downlink}
\end{subfigure}
\caption[Convergence speed of active traces for uplink and downlink]{Illustrating convergence speed of active traces for uplink and downlink. The active traces approach channel capacity only 20 seconds into the fetch. Note that expected value of uplink capacity is half the value of downlink and that is due to TCP protocol sending one ACK for every two data packets. }
\label{fig:convergence}
\end{figure}

{
\section{Evaluation of Privacy Performance}

In this section we begin by collecting experimental measurements of the sequence of traces generated by fetching of each of the Alexa top 100 finance and health web sites in Ireland.    Recall that the sequence of traces generated is affected by interaction between TCP and the traffic shaping applied, and so a simple replay approach is not appropriate. 
 {
Note also that the activation of traces is not synchronized with the start of a web fetch and there may be an offset between the beginning of a fetch and the activation of the first trace whose duration depends on the user's previous network activity and the duration of the silent period (if any) before the start of a fetch.  There is also a random offset between the activation of traces in the outgoing and incoming directions which depends on the link speed, previous traffic history and TCP protocol behaviour. These random offsets lead to variations in the sequence of traces generated by repeated fetches of the same web site.}  The sequence of traces generated also depends on the choice of scheduler and trace parameters, for which we use the values derived from the analysis in the previous section.

Using the collected experimental trace data we carry out a simplified indistinguishability analysis i.e. an analysis of how many combinations of web fetches are consistent with each measured trace.   We then evaluate the performance of the trace-based approach under the state of the art k-NN attack from \cite{wang14}.   This attack uses a large number of features and so is rather computationally demanding but as noted by Panchenko et al \cite{panchenko16} it provides benchmark that outperforms all previous attacks\footnote{The attack introduced in \cite{panchenko16}] uses fewer features and scales better but since it is reported to have broadly similar performance to the k-NN attack it seems reasonable to use the k-NN attack of \cite{wang14} as a baseline for comparison.}. {We also evaluate the performance of the approach against the attack from \cite{panchenko16} which uses fewer features and SVM to classify web pages.}

\subsection{Trace Sequences Generated by Alexa Top 100 Web Sites}
We fetched each of the home pages from the Alexa top 100 finance and health web sites in Ireland 100 times using the Seculink VPN (so 10,000 fetches in total).   Note that several of these pages include dynamic AJAX content as well as adverts etc which can change between fetches.}  Figure \ref{fig:active_vs_time} shows four example trace time histories recorded during these fetches.  It can be seen that the trace time histories in Figures \ref{fig:acvstime_w1} and  \ref{fig:acvstime_w3} are identical and so evidently these two web pages cannot be distinguished by an attacker.   Figure \ref{fig:profiles_vs_webpages} plots the number of distinct trace time histories measured on the uplink and downlink while fetching the 100 web pages and also the number of pages for which each trace time history is observed.   It can be seen that certain trace time histories are generated by 10-30 different web pages and so these web pages are indistinguishable to an attacker.  

\begin{figure}[!t]
\centering
\begin{subfigure}[]{0.48\columnwidth}
\centering
  \includegraphics[clip, width=0.95\columnwidth]{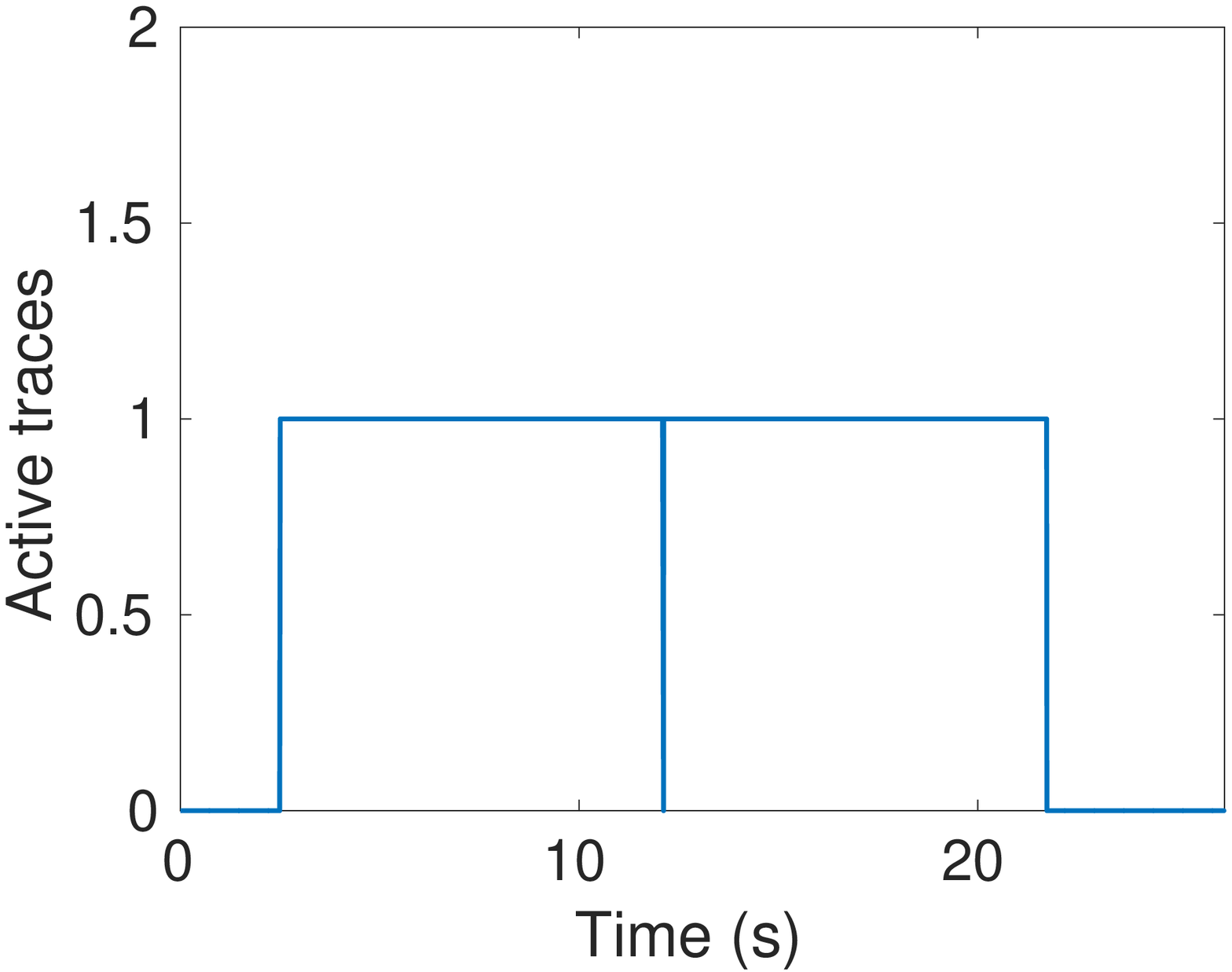}
  	\caption{}
	\label{fig:acvstime_w1}
\end{subfigure}%
\begin{subfigure}[]{0.48\columnwidth}
\centering
  \includegraphics[clip, width=0.95\columnwidth]{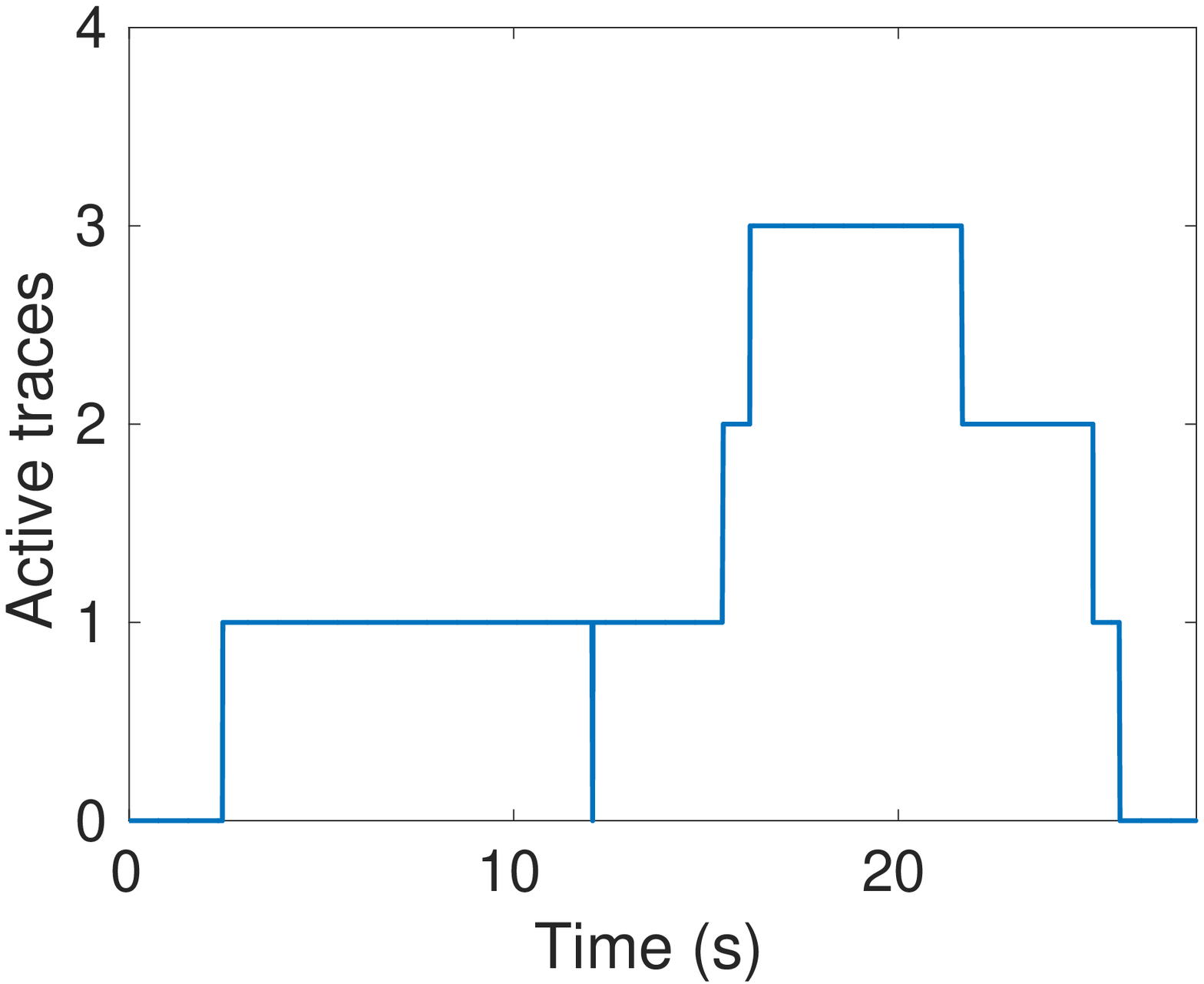}
	\caption{}
	\label{fig:acvstime_w2}
\end{subfigure}
\\
\begin{subfigure}[w3]{0.48\columnwidth}
\centering
  \includegraphics[clip, width=0.95\columnwidth]{traces-vs-time-uplink-test-126.eps}
  	\caption{}
	\label{fig:acvstime_w3}
\end{subfigure}%
\begin{subfigure}[w4]{0.48\columnwidth}
\centering
  \includegraphics[clip, width=0.95\columnwidth]{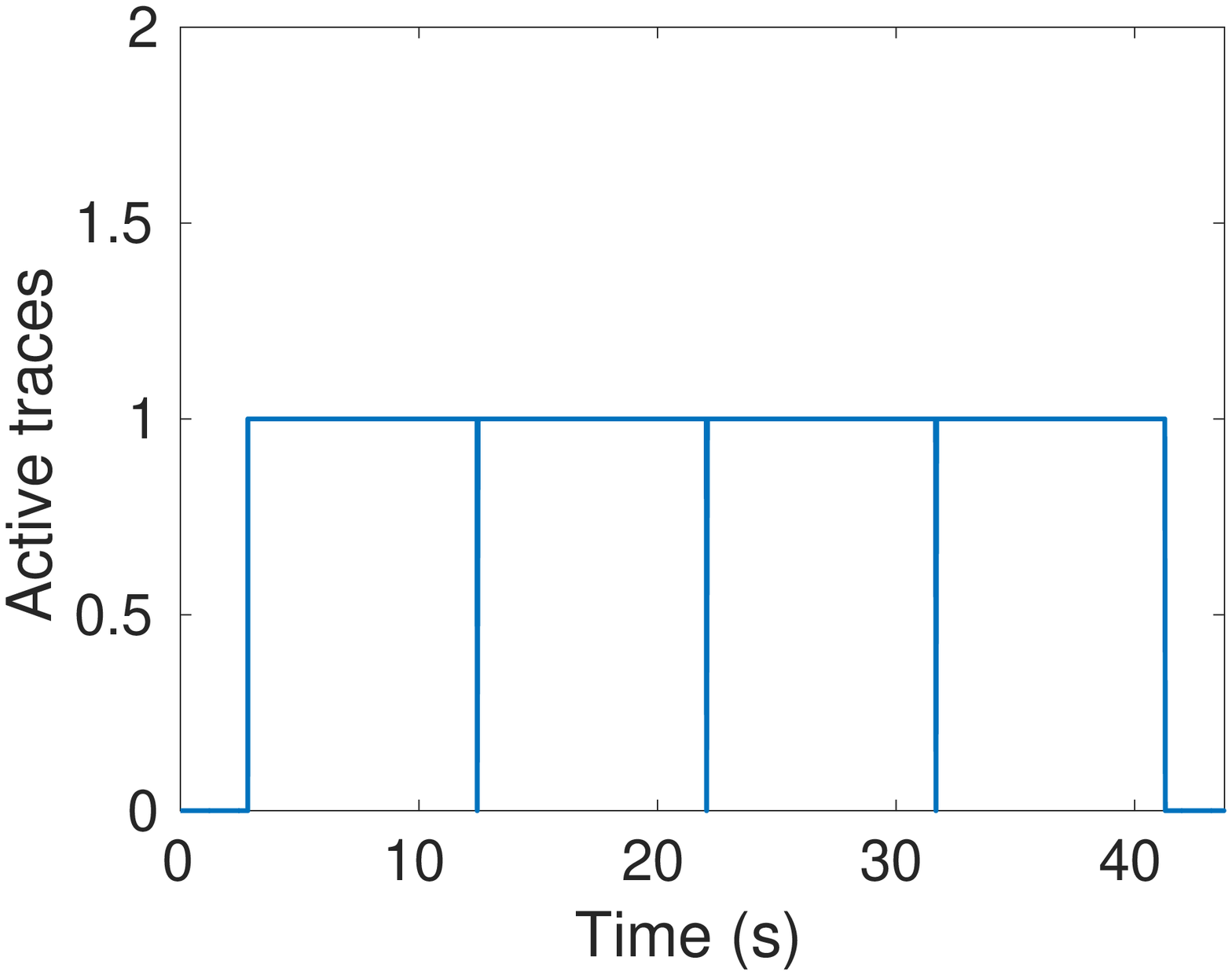}
	\caption{}
	\label{fig:acvstime_w4}
\end{subfigure}
\caption[Trace time histories of fetching four websites]{Examples of trace time histories recorded when fetching the home pages of four websites from the Alexa top 100 finance and health web sites in Ireland. ($P = 1682, n=9615, c = 20, \gamma = 1024$)}
\label{fig:active_vs_time}
\end{figure}

\begin{figure}
\centering
\begin{subfigure}[Uplink]{0.48\columnwidth}
\centering
  \includegraphics[clip, width=0.98\columnwidth]{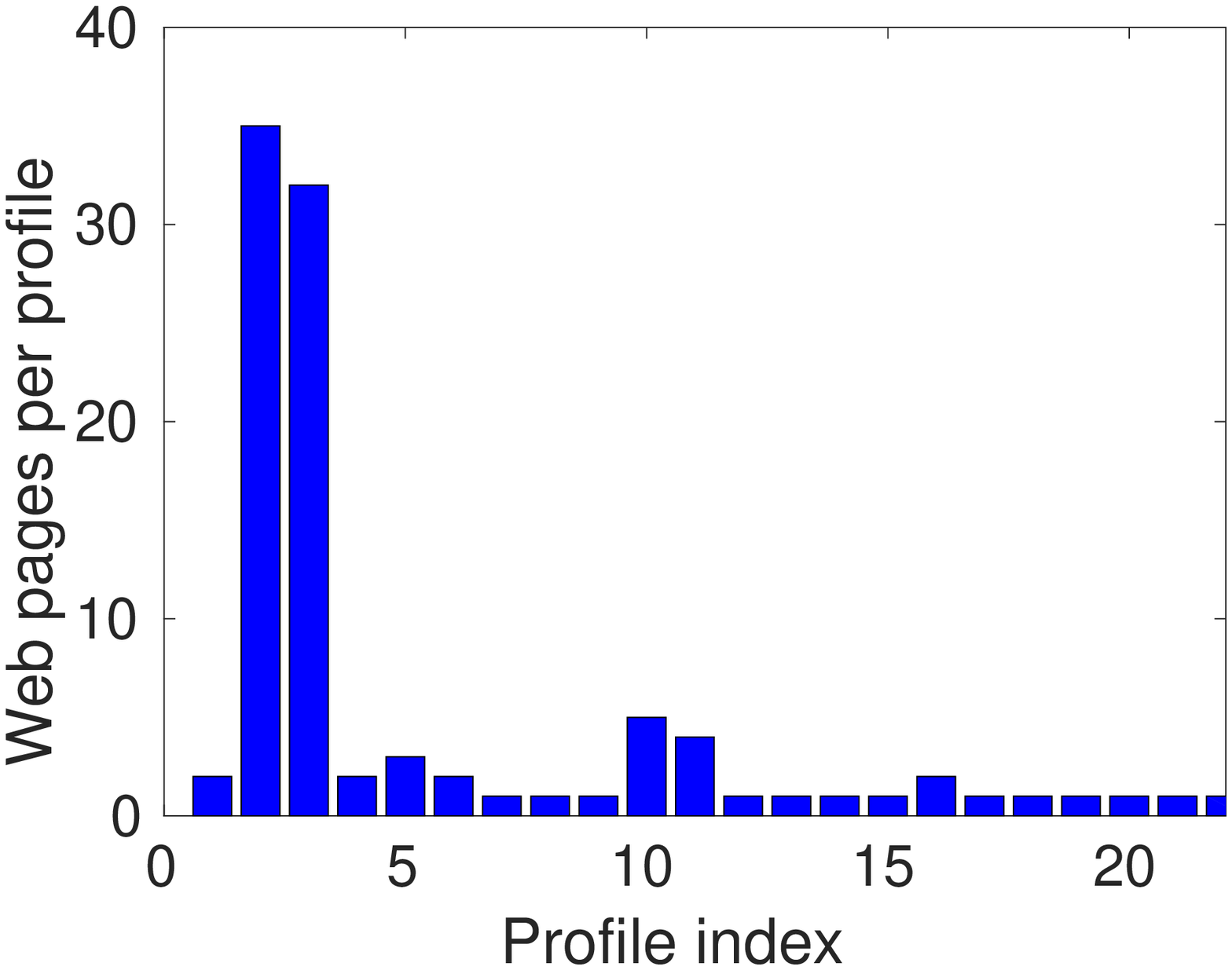}
  	\caption{Uplink}
\end{subfigure}%
\begin{subfigure}[Downlink]{0.48\columnwidth}
\centering
  \includegraphics[clip, width=0.98\columnwidth]{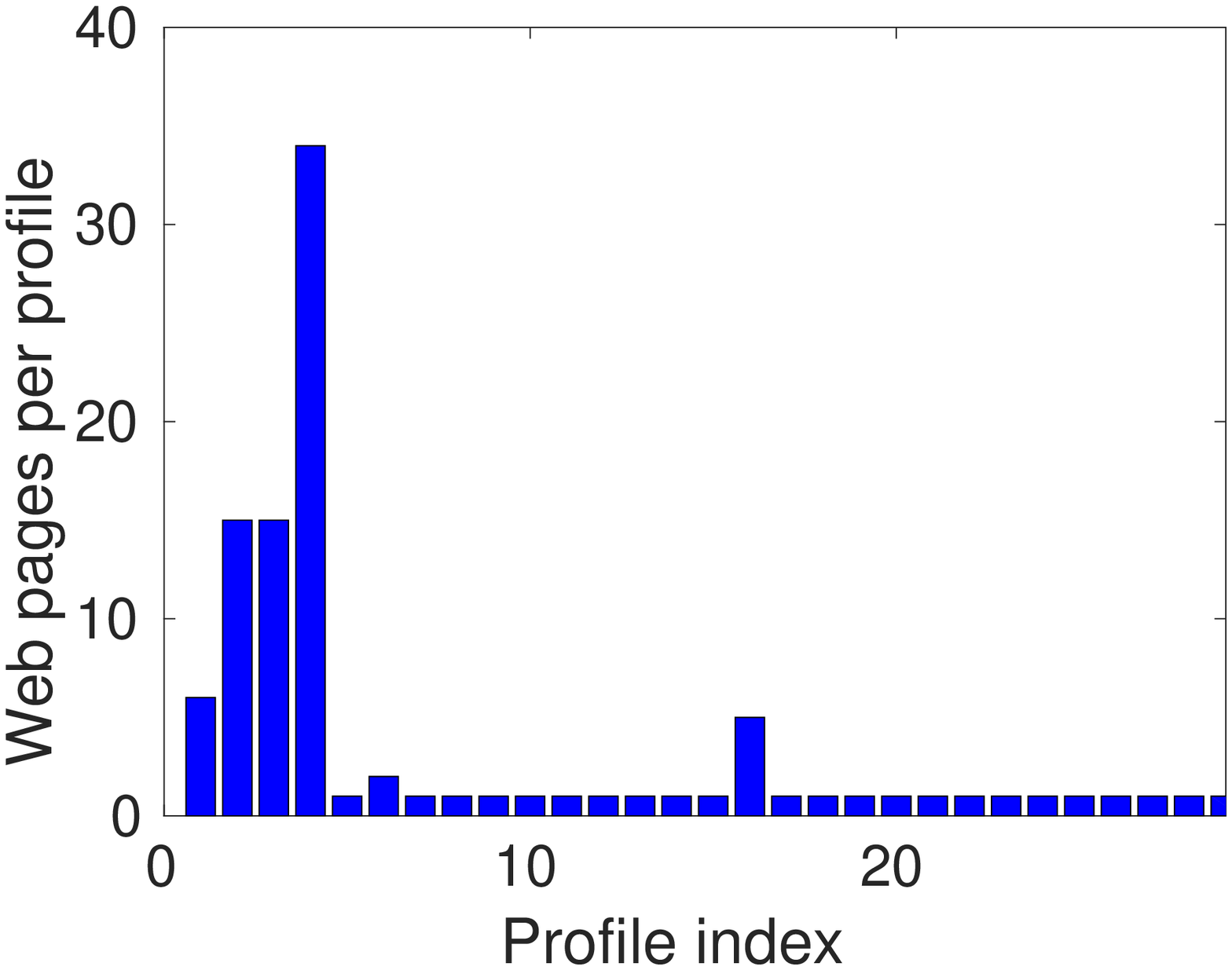}
  	\caption{Downlink}
\end{subfigure}
\caption[Number of pages generating each distinct web trace]{Counts of the number of web pages generating each distinct web trace time history.}
\label{fig:profiles_vs_webpages}
\end{figure}

It can also be seen in Figure \ref{fig:profiles_vs_webpages} that around 20 trace time histories are generated by a single web page, and so potentially vulnerable to attack.   An example is shown in Figure \ref{fig:acvstime_w4}.   Evidently the trace time history in Figure \ref{fig:acvstime_w4} cannot be distinguished from \emph{two} fetches of the web sites in Figures \ref{fig:acvstime_w1} and \ref{fig:acvstime_w3}.   The trace time history in Figure \ref{fig:acvstime_w2} is more complex, and raises the question of whether there exists one or more combinations of web page fetches that yield the same trace time history and so cannot be distinguished from this web fetch.  

\begin{figure}
\centering
\begin{subfigure}[single]{0.48\columnwidth}
\centering
  \includegraphics[clip, width=0.98\columnwidth]{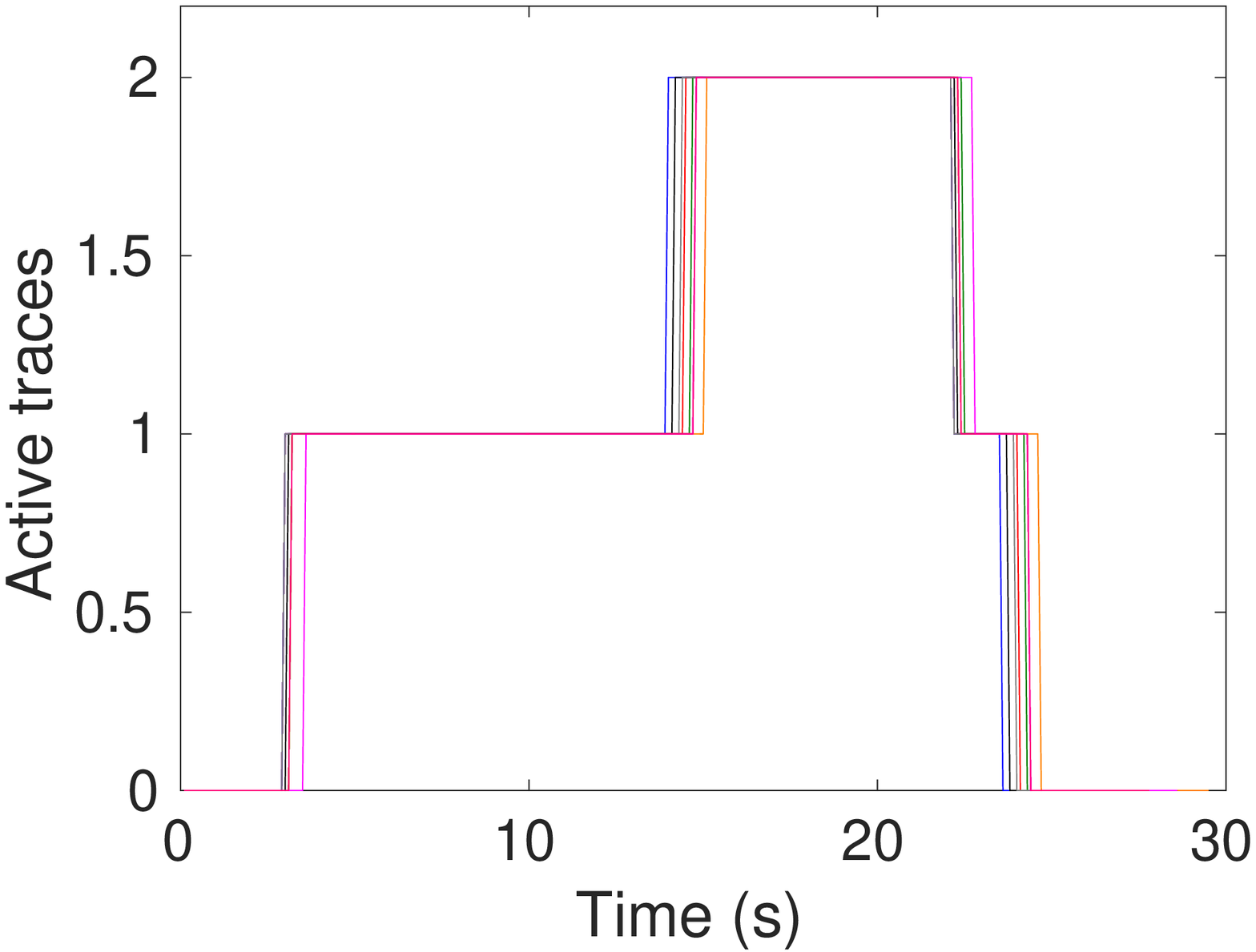}
  	\caption{Single web page}
\end{subfigure}
\begin{subfigure}[multi]{0.48\columnwidth}
\centering
  \includegraphics[clip, width=0.98\columnwidth]{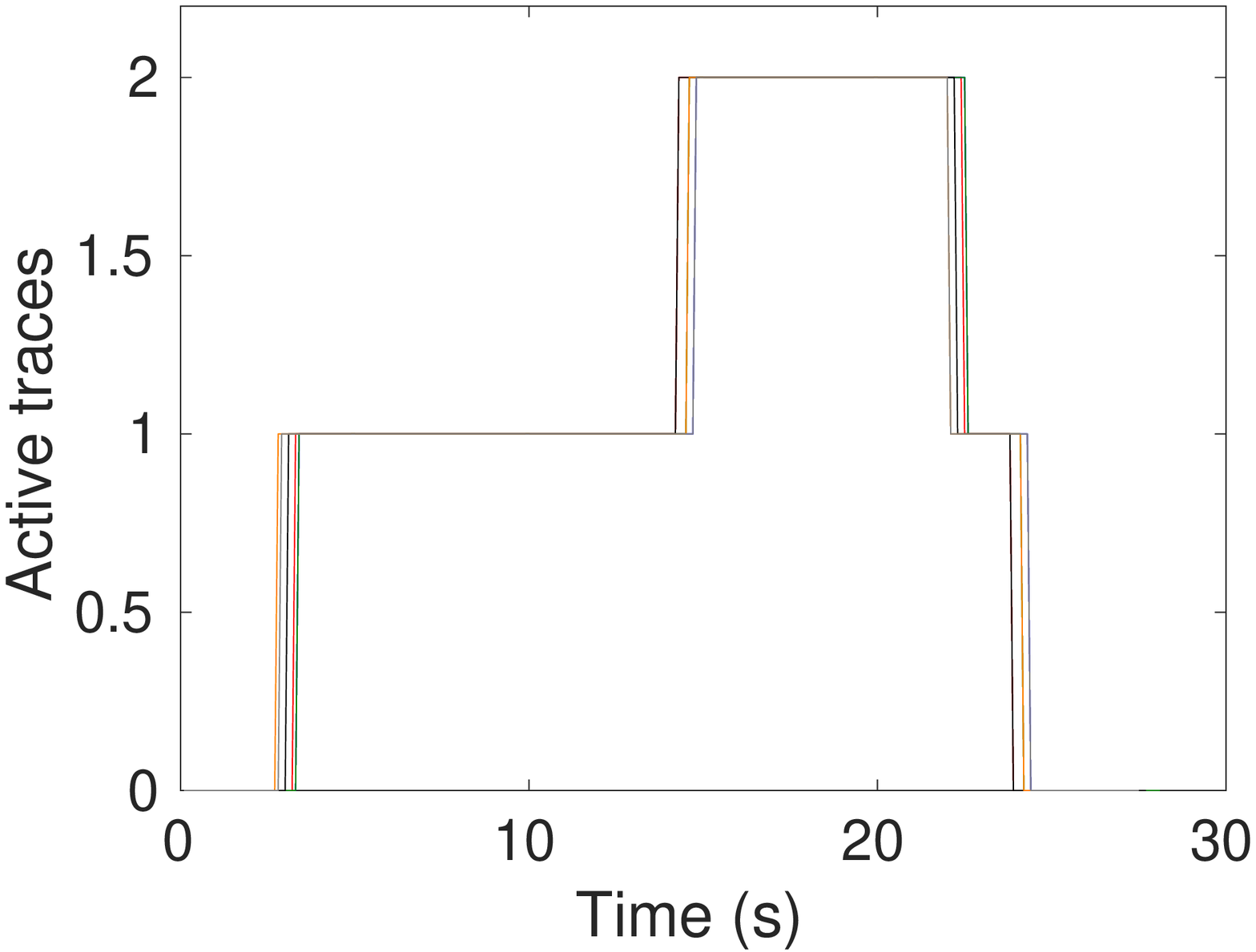}
  	\caption{Concurrent web search}
\end{subfigure}%
\caption[Comparison between trace time histories from two different browsing behaviours.]{{Comparison between trace time histories from two different browsing sessions.  In (a) a single web page is fetched while in (b) three Google search queries fare conducted on concurrent tabs with a 0.2s time gap between fetches. ($P = 1682$, $n = 9615$, $c = 20$, $\gamma = 1024$).}}
\label{fig:sim_patterns}
\end{figure}

{As an illustrative example we collected 10 samples for each of the two following browsing behaviours (i) a web page from our data set is fetched on a single tab and (ii) on three concurrent tabs, the words "invisible", "history" and "browsing" are queried from Google search with a 0.2 second gap between queries. The measured trace time histories are plotted in Figure \ref{fig:sim_patterns}.  It can be seen that the time histories look very similar, thus allowing a user to plausibly deny that they fetched the single sensitive web page since the observed trace time history might equally have been generated by non-sensitive google queries.

{
\subsection{Indistinguishability Analysis}

We can analyse the indistinguishability of the traces associated with the Alexa data more formally as follows.   The analysis is basically that of a packing problem, namely given an observed sequence of traces what combinations of web fetches are compatible with that sequence i.e. what combination of other trace sequences (since each web fetch generates a sequence) can be packed into a given observed sequence of traces.    Let $\mathcal{H}$ denote the set of measured trace sequences generated by the Alexa web pages.   For  trace sequence $h\in \mathcal{H}$ we can determine (by exhaustive search) the combinations $C_{h,1},C_{h,2},\dots,C_{h,m_h}$ of the other trace sequences that are indistinguishable from $h$, where combination $C_{h,i}=\{(h_{i,1},n_1), (h_{i,2},n_2),\dots\}$ with $h_{i,j}\in \mathcal{H}\setminus\{h\}$ and $n_j$ the number of times that $h_{i,j}$ is repeated in the combination.   


Of course not all combinations are equally likely and for plausible deniability we need there to be several combinations that have reasonably high probability.   As already noted, there are two main difficulties with actually carrying out probability calculations however.  One is that the probability of fetching a page is usually unknown and user dependent, although it might be estimated from historical data.  The other is that the pages fetched may be correlated, and again would have to be estimated from historical data.   With these caveats in mind, we can still gain insight under some simplifying assumptions.}

Let $\mathcal{W}_h$ denote the set of single web pages that generate trace sequence $h$, so $|\mathcal{W}_h|$ is the number of single web pages that generate $h$ (see Figure \ref{fig:profiles_vs_webpages}).   Let the web page fetched $W$ be a random variable that takes values in $\cup_{h\in\mathcal{H}}\mathcal{W}_h$.    Assume, for simplicity, that the web pages in $\mathcal{W}_h$ are fetched independently and are equally likely to be fetched.  Then when combination $C_{h,i}$ is observed the probability that page $\omega$ was fetched is,
\begin{align}
&\textsc{P}(W=\omega | C_{h,i}) \notag\\
&\qquad= \sum_{(h,n)\in C_i: \omega\in \mathcal{W}_h} \sum_{j=1}^{n}  {n \choose j} \frac{1}{|\mathcal{W}_{h}|}^j(1-\frac{1}{|\mathcal{W}_{h}|})^{n-j}
\end{align}
Let $C$ be a random variable which is the combination used.  Assume, again for simplicity, that each combination $C_{h,1},C_{h,2},\dots,C_{h,m_h}$ is equally probable when trace $h$ is observed, i.e $\textsc{P}(C=C_{h,i})=1/m_h$.  Then the probability that web page $\omega$ was fetched given that trace time history $h$ was observed is,
\begin{align}
\textsc{P}(W=\omega | h) = \sum_{i=1}^{m_h} \textsc{P}(W=\omega | C_{h,i})/m_h\label{eq:pomega}
\end{align}

When the traces $h\in\mathcal{H}$ are equally likely to be observed then $\textsc{P}(W=\omega ) = \frac{1}{|\mathcal{H}|}\sum_{h\in\mathcal{H}}\textsc{P}(W=\omega | h)$.   Figure \ref{fig:websites_probabilities_mean} plots $\textsc{P}(W=\omega)$ calculated using (\ref{eq:pomega}) for each of the 100 web pages, sorted in increasing order.   It can be seen that this probability is less than 0.08 on the downlink and less than 0.05 on the uplink.   We can conclude therefore that with the foregoing assumptions the user can reasonably deny that page $\omega$ was fetched despite an attacker observing the transmitted packet trace.

\begin{figure}
\centering
\begin{subfigure}[Uplink]{0.48\columnwidth}
\centering
  \includegraphics[clip, width=0.98\columnwidth]{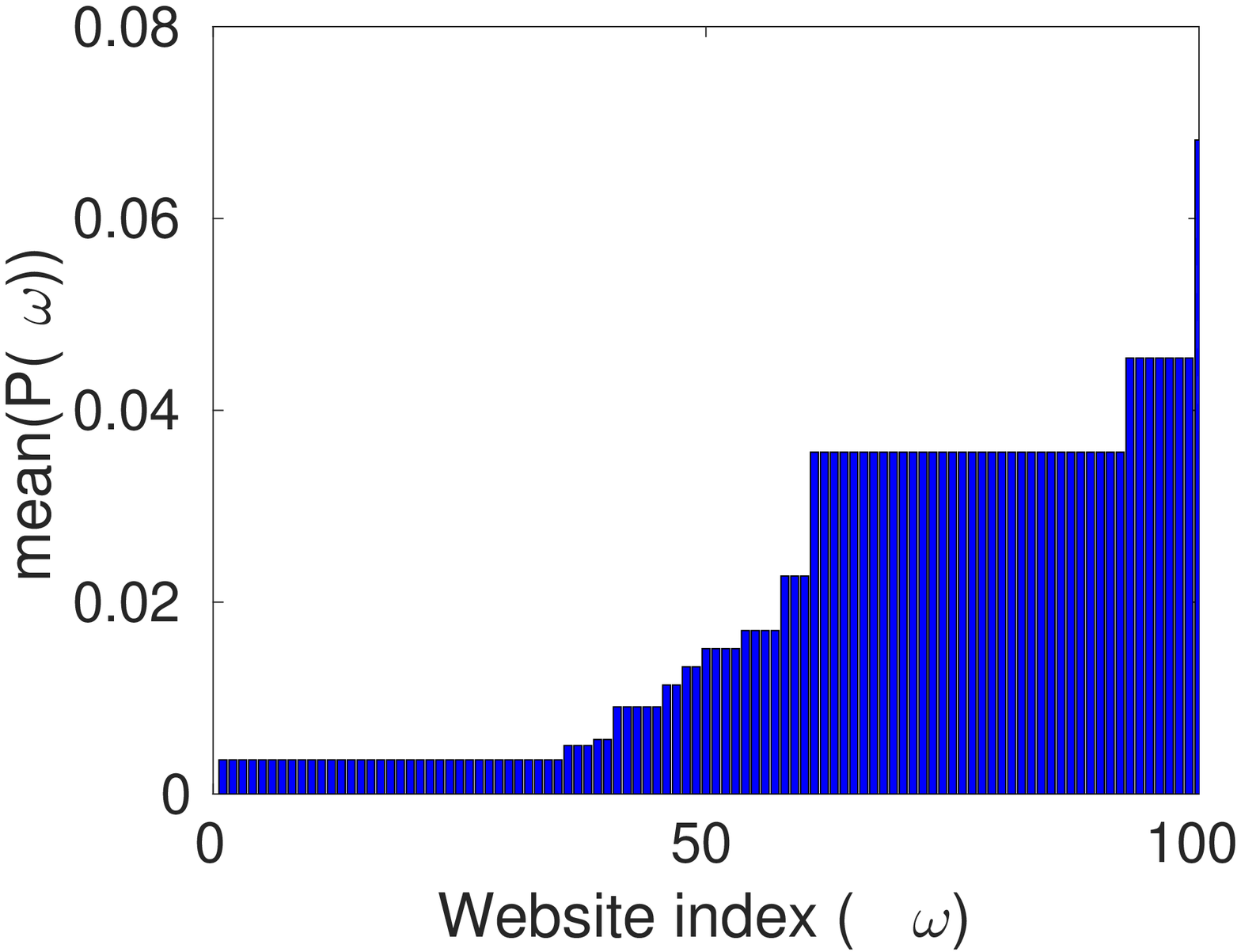}
  	\caption{Uplink}
\end{subfigure}%
\begin{subfigure}[Downlink]{0.48\columnwidth}
\centering
  \includegraphics[clip, width=0.98\columnwidth]{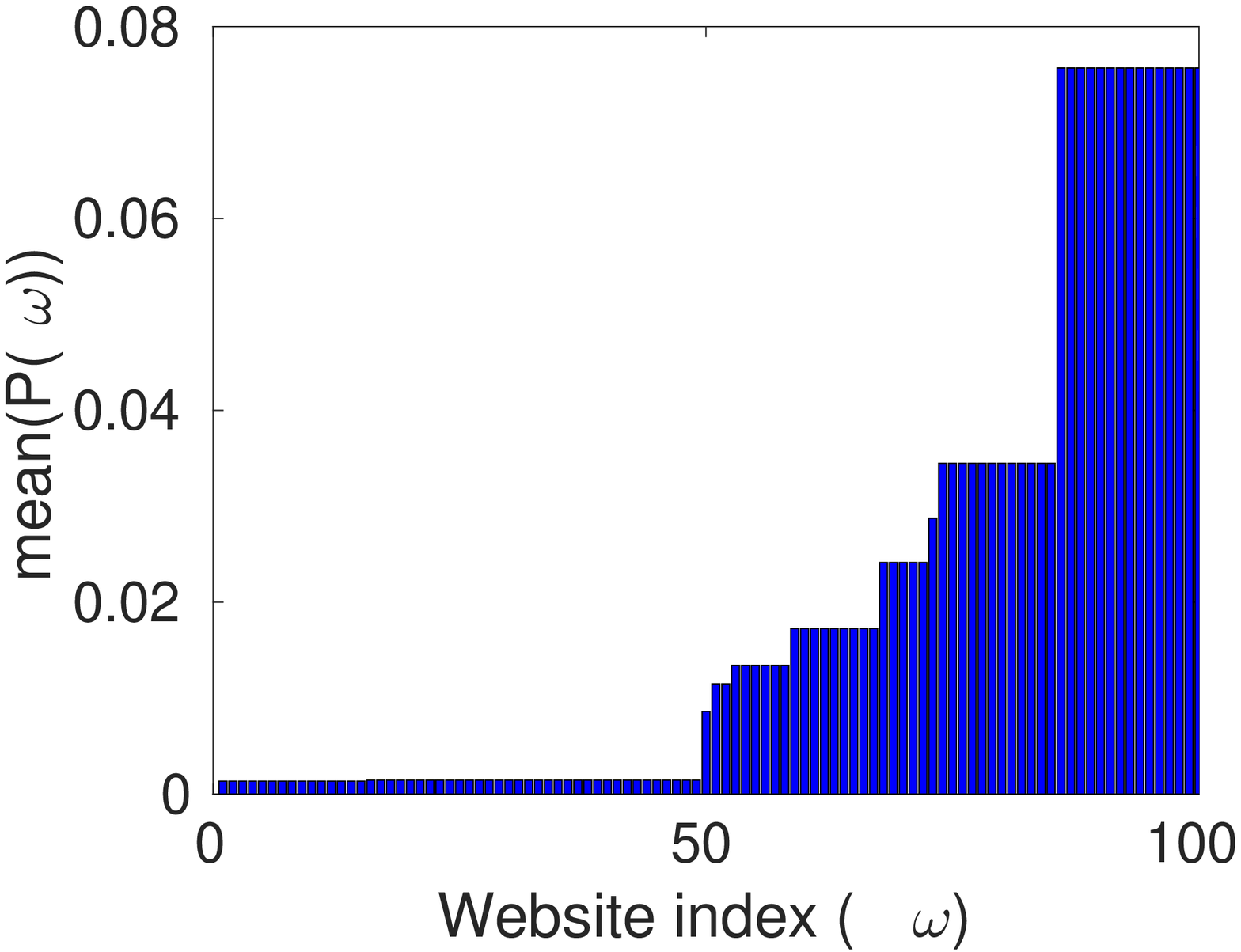}
  	\caption{Downlink}
\end{subfigure}
\caption[Observation probability of $\omega$ assuming equiprobable combinations]{Probability $\textsc{P}(\omega)$ that web page $\omega$ was fetched given an observed trace time history vs $\omega$ and assuming equiprobable combinations.}
\label{fig:websites_probabilities_mean}
\end{figure}

To get a sense of the sensitivity of these values to the assumption of equiprobability, for comparison, let $h^*\in\arg\min_{h\in\mathcal{H}} \textsc{P}(W=\omega | h)$ and suppose that this worst case trace $h^*$ is always observed i.e. $\textsc{P}(W=\omega ) = \textsc{P}(W=\omega | h^*)$.  Figure \ref{fig:websites_probabilities_max} shows the calculated $\textsc{P}(W=\omega|h^*)$ for our measured web fetches.  It can be seen that $\textsc{P}(W=\omega|h^*)$ is higher than in Figure \ref{fig:websites_probabilities_mean}, as expected, and indeed reaches one but only for a small number of web pages.   Hence, even in this worst case a user has fairly level of strong deniability.

\begin{figure}
\centering
\begin{subfigure}[Uplink]{0.48\columnwidth}
\centering
  \includegraphics[clip, width=0.98\columnwidth]{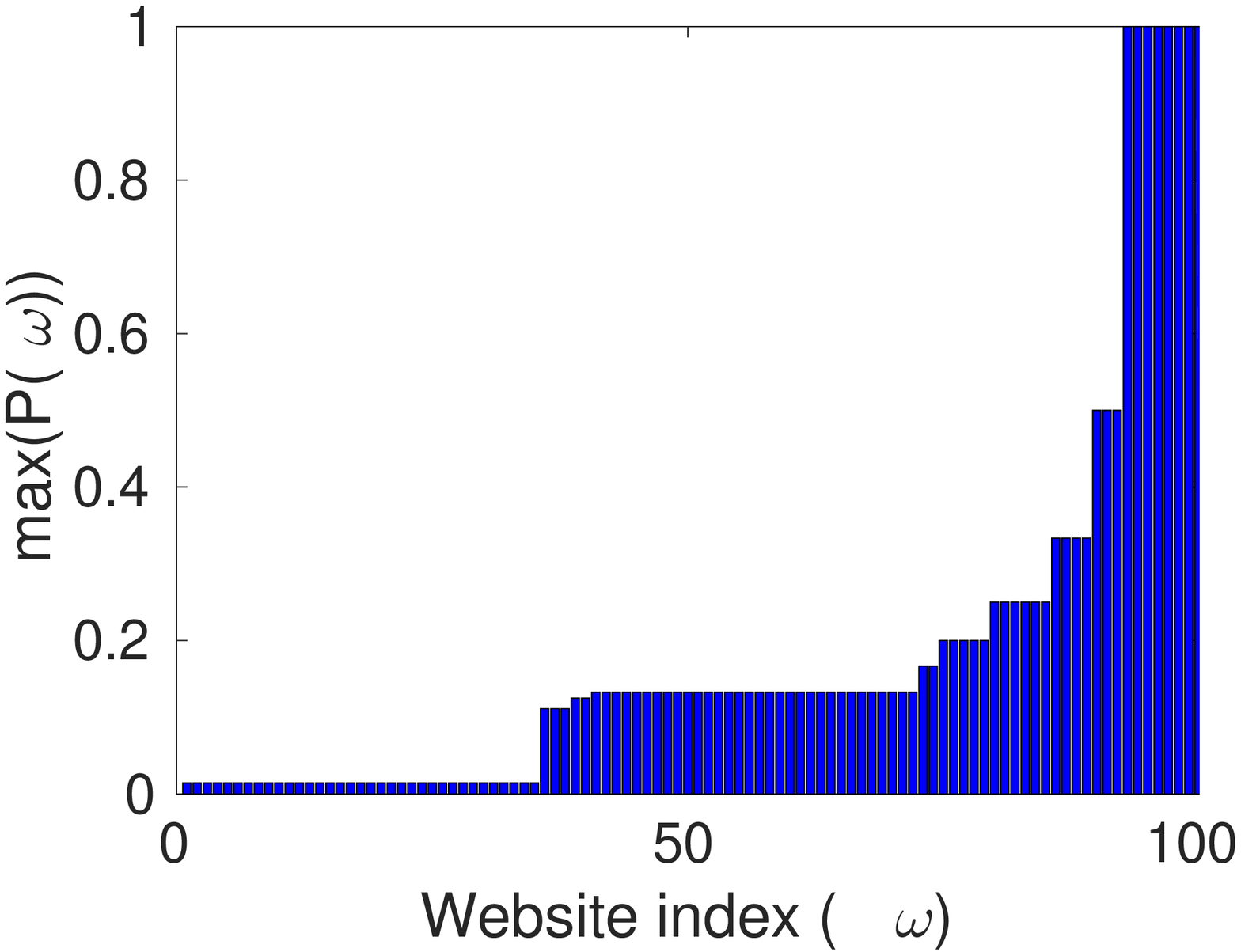}
  	\caption{Uplink}
\end{subfigure}%
\begin{subfigure}[Downlink]{0.48\columnwidth}
\centering
  \includegraphics[clip, width=0.98\columnwidth]{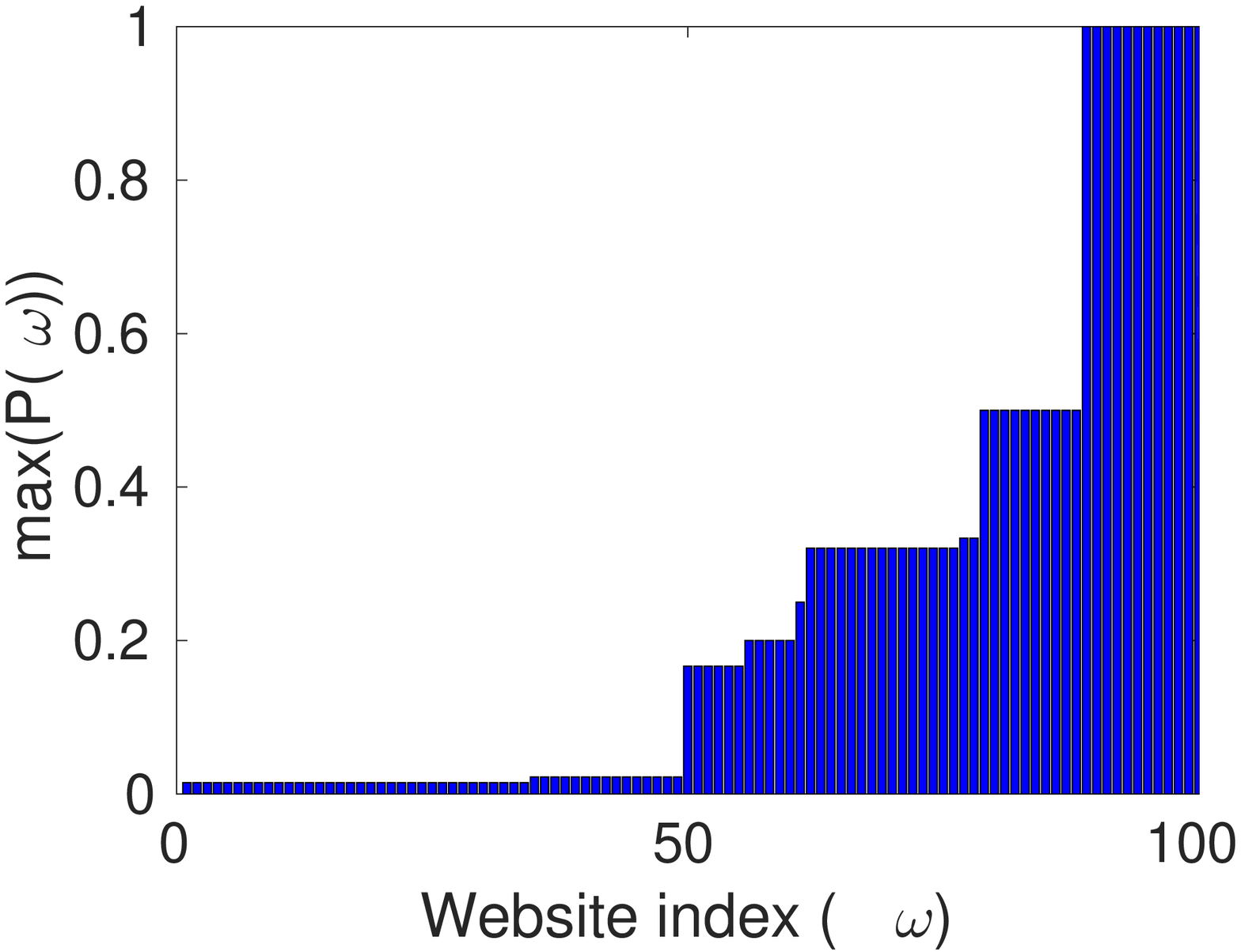}
  	\caption{Downlink}
\end{subfigure}
\caption[Observation probability of $\omega$ assuming worst combination has probability one]{Probability $\textsc{P}(\omega)$ that web page $\omega$ was fetched given an observed trace time history vs $\omega$ and assuming worst combination has probability one.}
\label{fig:websites_probabilities_max}
\end{figure}

{
\subsection{Performance Against State of the Art Attacks}
\label{sec:attack}

In this section we evaluate the performance of the proposed trace-based approach against two state of the art attacks namely the k-NN attack from \cite{wang14} and the SVM attack from \cite{panchenko16}.

The attack from \cite{wang14}} uses a large number of features, including packet timing information and the duration of transmission, and has state of the art performance.   We note that the attack assumes (i) web pages are fetched individually and sequentially, one at a time, and (ii) the start and end of each web fetch is known.   The assumption that web pages are fetched sequentially is a strong one and can be expected to greatly reduce the indistinguishability of the trace sequences generated by the Seculink tunnel (since it excludes from consideration all combinations of web fetches consistent with an observed trace sequence other than those involving a single web fetch).   Knowledge of the start and end of each web fetch also need not be straightforward to obtain in practice, especially when a trace-based defence is used, and we assist the attack by explicitly providing this information.    

\begin{figure}
\centering
\includegraphics[clip,angle=270,width=0.7\columnwidth]{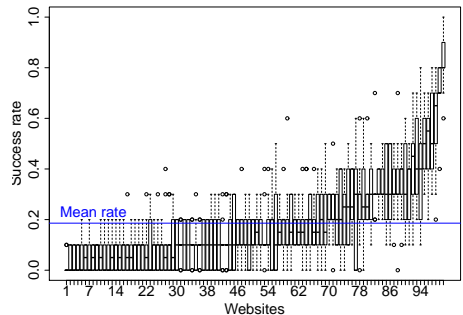}
\caption[Performance of Wang's attack]{{Performance of the attack in \cite{wang14} against the Alexa top 100 websites when using the Seculink tunnel.  Each web site is fetched 100 times and 10-fold cross validation is used (splitting training and test data).   The box and whiskers indicate the median, 25 and 75 percentiles plus outliers for each page.  The mean success rate of the attack is 18\%, indicated by the horizontal line.}}\label{fig:wang-result}
\end{figure}

Figure \ref{fig:wang-result} shows the measured performance of the attack against the Alexa pages fetched via the Seculink tunnel.   Each web site is fetched 100 times and 10-fold cross validation is used where training and test data are both collected from the tunnel.   The median, 25 and 75 percentiles plus outliers of the attack success rate are shown for each page.   We note that the pages considered here include dynamic content such as AJAX, adverts, tweets etc.  The relatively low spread in success rates for each web page seen in Figure \ref{fig:wang-result} suggests that this content has little impact on the performance of proposed attack or defence.

\begin{figure}
\centering
\begin{subfigure}[common]{0.48\columnwidth}
\centering
  \includegraphics[clip, width=0.98\columnwidth]{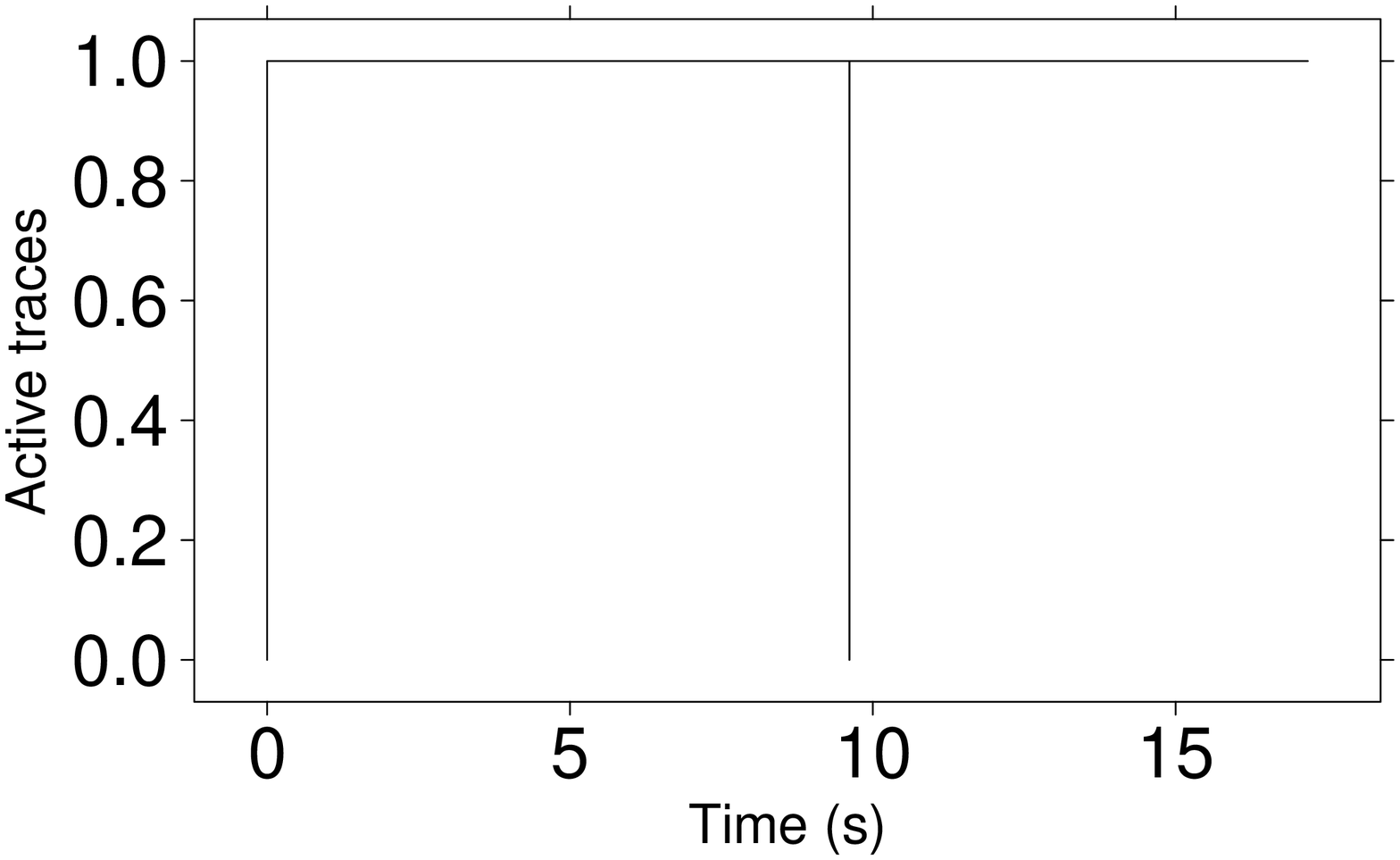}
  	\caption{Most common trace sequence}
	\label{fig:traces-common}
\end{subfigure}
\begin{subfigure}[uncommon]{0.48\columnwidth}
\centering
  \includegraphics[clip, width=0.98\columnwidth]{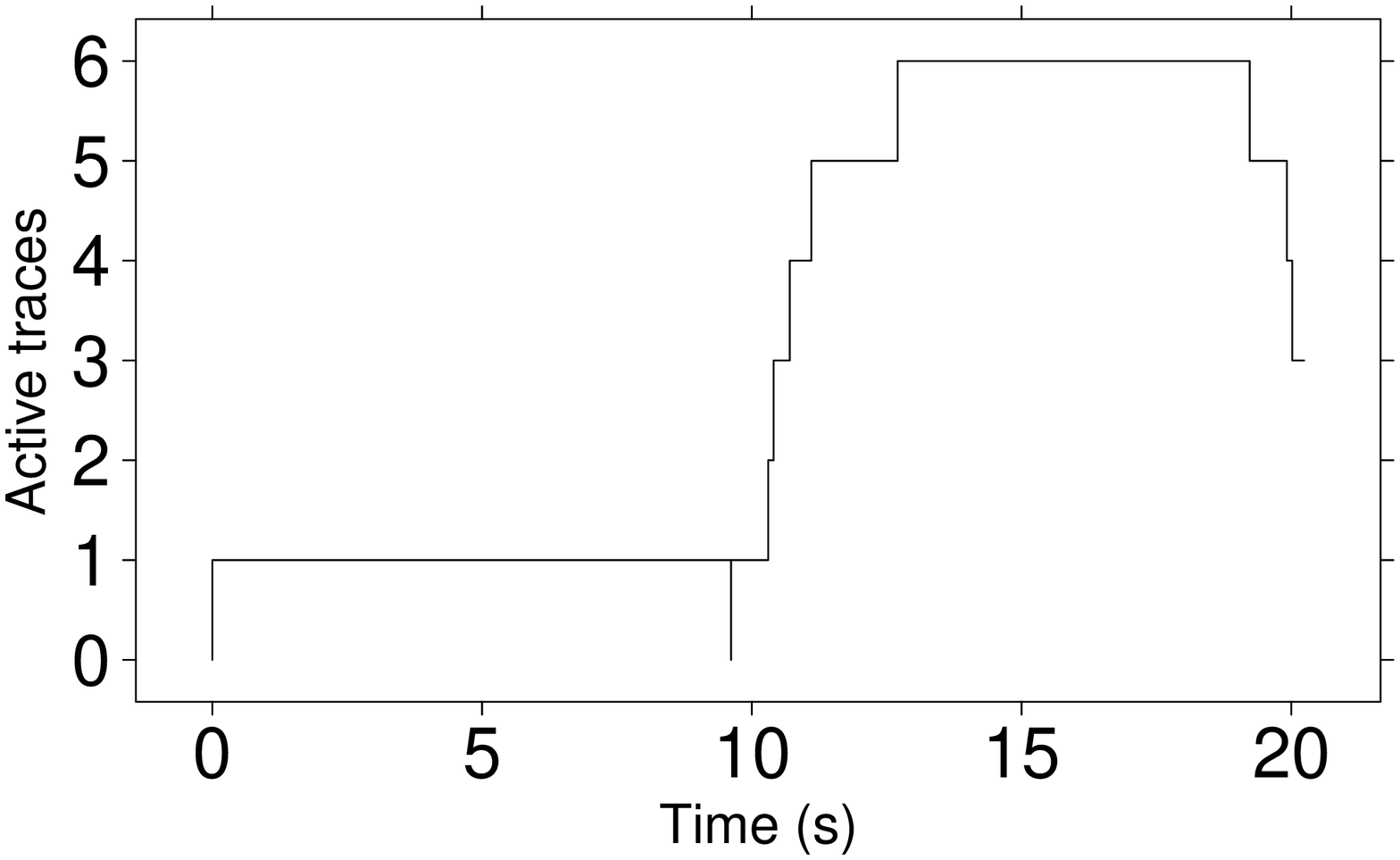}
  	\caption{An uncommon trace pattern}
	\label{fig:traces-uncommon}
\end{subfigure}%
\caption{{Examples of common and complex traces sequences generated by the Alexa web pages. }}
\label{fig:traces-shapes}
\end{figure}

The mean success rate of the attack is 18\%.   This is surprisingly low given the extra help given to the attacker by the assumption of sequential web page fetches and knowledge of the start/end times of each fetch, and so quite encouraging.   

It can be seen that a handful of websites suffer from a higher success rate (the web sites are ordered by success rate in Figure \ref{fig:wang-result}) and it is these which pull the average success rate up.  To try to gain more insight into why the attack has much lower success rates against some pages compared to others we plot the trace sequences of example pages in Figure \ref{fig:traces-shapes}.    Most of websites are served with two back to back traces,  see Figure \ref{fig:traces-common}, which is the most common timing pattern among all of the websites considered.   The websites against which the attack is successful tend to generate more complex trace sequences, see for example Figure \ref{fig:traces-uncommon}.    Since the attacker has the additional information that web pages are fetched sequentially these more complex traces can be used to identify a single web page.  However, a less powerful attacker (and perhaps a more realistic one) which could not rely on this knowledge would be faced with the fact that many combinations of multiple web fetches would also be consistent with such an observed trace -- this is easy to see since most web fetches generate only two back to back traces as shown in Figure \ref{fig:traces-common}.   Hence the attack success rates reported here are likely overly optimistic.

\begin{figure}
\centering
\includegraphics[clip,angle=270,width=0.7\columnwidth]{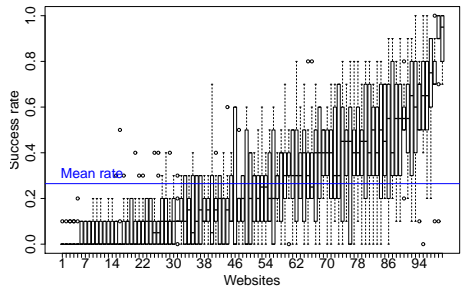}
\caption[Performance of Panchenko's attack]{{Best performance of the attack in \cite{panchenko16} against our baseline Alexa top 100 websites when using the Seculink tunnel. The SVM parameters are $c=2048$ and $\gamma=0.5$. The mean success rate of the attack is 26\%, indicated by the horizontal line.}}\label{fig:panchenko-result}
\end{figure}

{
The attack from \cite{panchenko16} uses around 100 features to capture size, direction and ordering of packets in web samples. Although this attack does not explicitly capture packet timing information, features regarding packet ordering may capture some of this information. Similarly to the attack from \cite{wang14} this attack also assumes that web pages are fetched individually and sequentially and that the beginning and ending of each capture is known to the attacker. Using fingerprints generated from features, the attack uses an SVM with parameters $c$ and $\gamma$ to classify web site samples. The algorithm tries different values for $c$ and $\gamma$ to find those which achieve the highest accuracy.

Figure \ref{fig:panchenko-result} shows the highest measured performance of the attack against the same dataset used with the previous k-NN attack. The SVM parameters for this highest accuracy measurement are $c=2048, \gamma=0.5$. We used 100 features following a recommendation from authors in \cite{panchenko16} that suggests no significant increase in classification is observed beyond this value. Again 10-fold cross validation is used where training and test data are both collected from the tunnel and the median, 25 and 75 percentiles plus outliers of the attack success rate are shown for each page.

The mean success rate of the attack is 26\%.   Again given the extra help given to the attacker by the assumption of sequential web page fetches and knowledge of the start/end times of each fetch, the results are surprisingly low.  For comparison, in \cite{panchenko16} the authors report a success rate of 91\% against Tor traffic in a closed world scenario using 104 features with 90 instances from each website.}

\subsection{Comparison with A State of the Art Defence}
We compare the dummy packet overhead of the trace-based approach against that of the Supersequence defence proposed in \cite{wang14}, which also adopts an indistinguishability approach (referred to as an anonymity set approach in \cite{wang14}) to enhancing privacy.   The Supersequence defence first clusters websites by similarity and then represents each cluster using a candidate supersequence. This requires the prior knowledge of web site packet traces and so as well as requiring the maintenance of a trace database not just for each user/device but also for each location since the packet trace generated by a fetch depends on the properties of the access link) it also need not work well with previously unseen web pages.   We use code kindly supplied by the authors of \cite{wang14} but note that this uses a simulation/replay approach which, as discussed previously, does not take account of the impact of traffic shaping on the web fetch packet trace and so some caution is needed when comparing the overheads reported for the Supersequence defence and for the full Seculink implementation of the trace-based defence.

\begin{figure}[!t]
\centering
\begin{subfigure}[Proposed defence]{0.48\columnwidth}
  \includegraphics[clip, width=1.0\columnwidth]{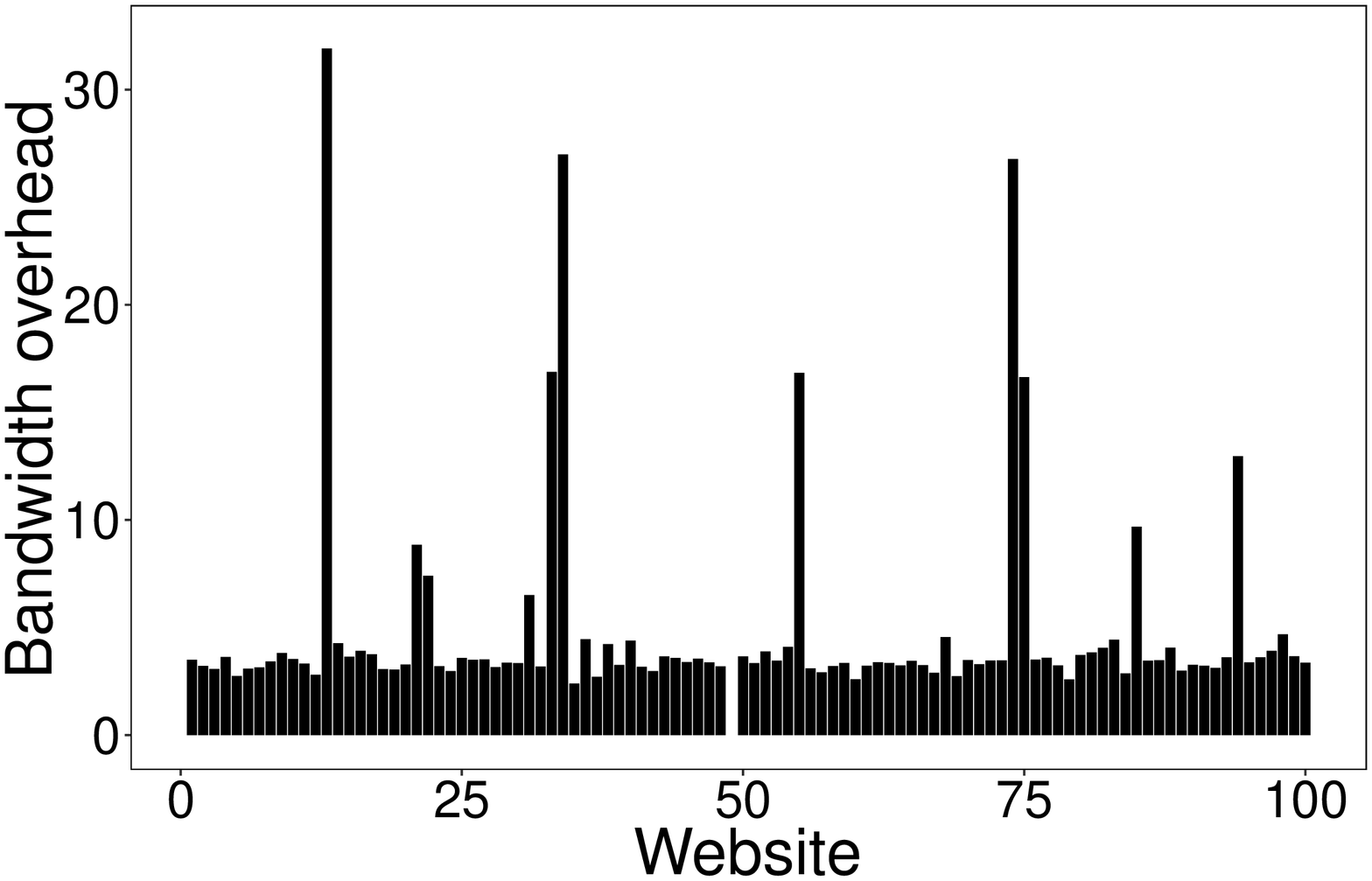}
	\label{fig:proposed-pverhead}
	\caption{Proposed trace-based defence}
\end{subfigure}
\begin{subfigure}[Supersequence defence]{0.48\columnwidth}
  \includegraphics[clip, width=1.0\columnwidth]{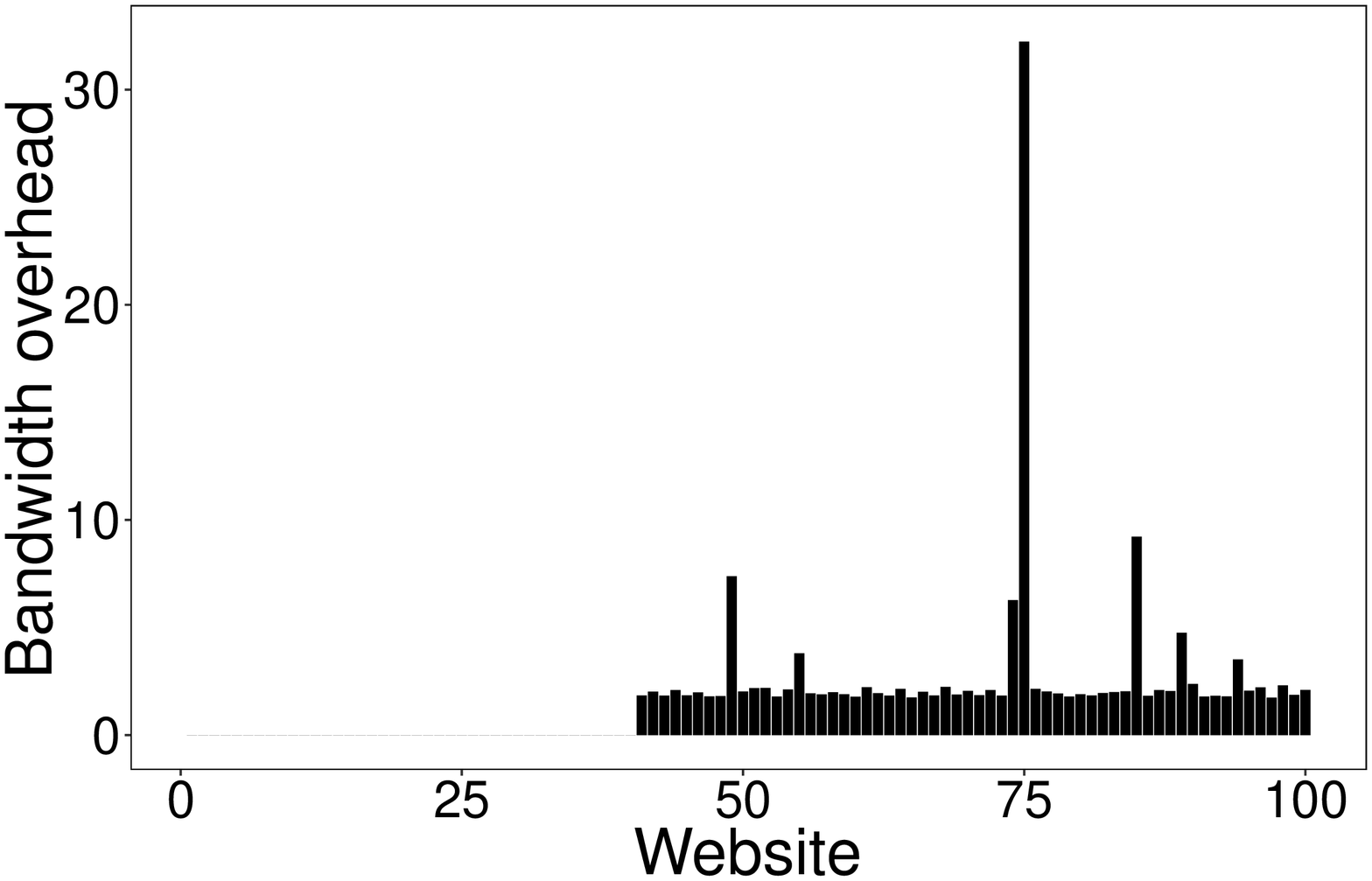}
	\label{fig:wang-overhead}
	\caption{Supersequence defence}
\end{subfigure}
\caption{{Measured bandwidth overhead for the Supersequence defence and the proposed trace-based defence. The Supersequence defence uses the first 40 websites to learn clusters and then defends the following 60 websites.}}
\label{fig:dummy-overhead}
\end{figure}

Figure \ref{fig:dummy-overhead} shows the measured bandwidth overhead (the measure used in \cite{wang14}) of both approaches for the Alexa web sites.   It can be seen that the overheads are broadly similar for both, although generally somewhat higher for the trace-based defence (although bear in mind the caveat that this uses a full implementation rather than a replay-based simulation).    

\section{Conclusions}
In this paper we introduce a low overhead capacity-achieving tunnel that is resistant to traffic analysis in the sense that it provides deniability to users that any specified web page was fetched given that a specified packet trace is observed on the tunnel.   We present a scheduler design for managing the transmission of traces to satisfy user traffic demand while maintaining reasonably low delay and throughput overhead due to dummy packets.   Experimental results are also presented demonstrating the effectiveness of this scheduler under a range of realistic network conditions and real web page fetches.

\section{Acknowledgement}
The assistance of Dr. Tao Wang in facilitating the implementation of existing attacks and defences is gratefully acknowledged.

\bibliography{references}{}
\bibliographystyle{plain}

\begin{IEEEbiography}{Saman Feghhi}
\end{IEEEbiography}

\begin{IEEEbiography}{Douglas Leith}
\end{IEEEbiography}

\end{document}